\documentclass[12pt]{article}
\usepackage{amsmath,amssymb,amsthm}
\usepackage{graphics,epsfig}
\usepackage{hyperref}
\usepackage[numbers]{natbib}
\usepackage{color}
\usepackage{graphicx}
\usepackage{caption}
\usepackage{subcaption}
\usepackage{float}
\usepackage{mathrsfs}
\usepackage{booktabs}
\usepackage{multirow}
\usepackage{comment}
\usepackage{setspace}
\usepackage{bbm}
\usepackage{bm}
\usepackage{amsfonts}
\usepackage{xcolor}
\usepackage{pslatex}
\usepackage{setspace}
\usepackage{bbm}
\usepackage{listings}
\usepackage{algorithm}
\usepackage{algpseudocode}

\def \ll {\boldsymbol{l}}

\def \E {\mathbb{E}}

\newtheorem{definition}{\bf Definition}
\newtheorem{defn}[definition]{\bf Definition}

\newtheorem{remark}{\bf Remark}
\newtheorem{rmk}[remark]{\bf Remark}

\newtheorem{theorem}{\bf Theorem}
\newtheorem{prop}[theorem]{\bf Proposition}
\newtheorem{lem}[theorem]{\bf Lemma}
\newtheorem{cor}[theorem]{\bf Corollary}


\begin{document}

	\title{\bfseries
		{Dynamic Physical Hedging amid Jump Losses, Reconstruction-Price Uncertainty, Population Interactions
		}}

	\author{
	Paramahansa Pramanik$^{1,2}$ \and 
	Michael Bowdin$^{3,4}$
	}
	
	\date{
		\small
		$^{1}$ Department of Mathematics and Statistics, University of South Alabama, Mobile, AL 36688, United States.\\
		$^{2}$
		Corresponding author, \texttt{ppramanik@southalabama.edu}\\[0.5em]
		$^{3}$Department of Mathematics and Statistics, University of South Alabama, Mobile, AL 36688, United States.\\
		$^{4}$ \texttt{mb2139@jagmail.southalabama.edu}
	}
	
	\maketitle

	\begin{abstract}
		We study dynamic physical hedging for insurers exposed jointly to catastrophe losses and stochastic reconstruction costs. Surplus evolves
		as a controlled jump diffusion whose loss amplitude combines marked catastrophe severity, an exogenous mean-reverting cost factor, and endogenous mitigation. We establish well-posedness, moment and stability
		estimates, and a stopping-time dynamic programming principle, and prove that the value function is the unique viscosity solution of the resulting nonlocal Hamilton-Jacobi-Bellman (HJB) equation Strategic
		interaction is introduced through a mean field game (MFG) with reduced-form vulnerability costs, yielding a coupled backward-forward HJB-Kolmogorov system. We establish relaxed equilibrium existence,
		Markovian realization, and uniqueness under appropriate compactness and monotonicity conditions. Numerical experiments show that reconstruction
		costs and capitalization materially affect optimal hedging and that cross-sectional vulnerability alters equilibrium costs. Tail-family robustness calculations further assess the sensitivity of these conclusions to alternative catastrophe-severity specifications.
	\end{abstract}


\section{Introduction}
\label{sec:introduction}

The economic problem considered in this paper is the allocation of physical mitigation when an insurer faces two distinct sources of uncertainty: discontinuous catastrophe losses and continuously fluctuating reconstruction prices. For $(t,x,y)\in [0,T]\times\mathbb R^2$, an admissible predictable policy $u\in\mathcal U_t$ is evaluated through the discounted criterion
	\begin{equation}
		\label{eq:intro-objective}
		V(t,x,y)
		=
		\inf_{u\in\mathcal U_t}
		\E_{t,x,y}\!\left[
		\int_t^T e^{-\delta(s-t)}
		f(s,X_s,Y_s,u_s)\,ds
		+
		e^{-\delta(T-t)}g(X_T,Y_T)
		\right].
	\end{equation}
	The controlled state is specified by
	\begin{equation}
		\label{eq:intro-state}
		\begin{aligned}
			dY_s
			&=
			\kappa(\vartheta-Y_s)\,ds+\sigma_I\,dW_s^I,\\
			dX_s
			&=
			b(X_{s-},Y_{s-},u_s)\,ds
			+
			\sigma_X(X_{s-},Y_{s-},u_s)\,dW_s^X
			-
			\int_E e^{Y_{s-}}a(z)h(u_s)\,
			\widetilde N(ds,dz),
		\end{aligned}
		\qquad
		(X_t,Y_t)=(x,y).
	\end{equation}
	The first coordinate $X$ represents insurer surplus, whereas $I=e^Y$ is an exogenously evolving reconstruction-cost index whose logarithm reverts toward the long-run level $\vartheta$. Catastrophe
	events and their marks are generated by the marked Poisson random measure $N$. The factor $h(u)$ measures the fraction of physical loss remaining after the insurer has committed resources to pre-arranged reconstruction capacity. Consequently, mitigation enters the model at
	the level of catastrophe severity itself: conditional on $(Y,z)$, the monetary jump transmitted to surplus is scaled by $h(u)$ rather than being offset only after the loss has occurred.
	
	A central modeling distinction is that the reconstruction-cost process and the catastrophe point process are separate primitive risks. Equation~\eqref{eq:intro-state} contains no catastrophe-triggered impulse in $Y$, and a realization of $N$ does not cause the replacement-cost index to increase. The framework therefore analyzes the adverse joint exposure created when a catastrophe occurs while
	reconstruction inputs are already expensive; it does not impose a causal mechanism through which the catastrophe itself generates a price surge. Within this specification, we establish well-posedness and stability of the controlled jump system, obtain a stopping-time
	dynamic programming principle, and identify the value function with the unique viscosity solution of the corresponding nonlocal Hamilton-Jacobi-Bellman (HJB) equation. We subsequently introduce cross-sectional strategic interaction through a mean field game of
	insurers, thereby linking individual mitigation choices to the evolution of aggregate vulnerability. This perspective complements evidence that catastrophe recovery may be accompanied by reconstruction-cost inflation and demand-surge effects
	\citep{DohrmannGuertlerHibbeln2017,Kousky2019,OlsenPorter2011,RothKunreuther1998}, while deliberately isolating the control problem generated by stochastic replacement-cost exposure from the
	separate economic question of how post-catastrophe reconstruction prices are endogenously formed.

The first strand of related work concerns catastrophe insurance,mitigation, reconstruction costs, and stochastic loss models. Classical insurance-risk theory studies reserves, ruin, dividends, and reinsurance under compound-Poisson and diffusion-perturbed surplus processes \citep{Asmussen2000,Avanzi2009,Azcue2010,DeFinetti1957,Dufresne1991,Garrido1989,Gerber1969,Hogaard1999,PanjerWillmot1992,Schmidli2008}.  Catastrophe-risk research emphasizes that insurance and ex ante mitigation jointly affect recovery and loss reduction \citep{Kousky2019,RothKunreuther1998}, while empirical work on reconstruction markets documents post-event price pressure, contractor scarcity, and insured-loss inflation
	\citep{DohrmannGuertlerHibbeln2017,OlsenPorter2011}.  These studies motivate replacement-cost exposure as an economically distinct source of catastrophe severity.  The present model is deliberately narrower
	$Y$ is an exogenous replacement-cost factor rather than a
	market-clearing reconstruction price, and aggregate rebuilding demand does not determine its dynamics \citep{pramanik2020optimization}.  On the probabilistic side, catastrophe clustering and transient changes in claim intensity have been modeled through shot-noise, contagion, and related point processes \citep{Albrecher2006,Dassios2003,Dassios2012,Macci2011,Liu2023,Pojer2023}, while marked point processes accommodate heterogeneous event severity \citep{Zeller2022}.  Our mechanism differs
	from intensity-based models because $e^Y$ scales the monetary loss generated by a catastrophe mark, and the control $h(u)$ acts on that same nonlocal loss amplitude.  Thus physical hedging changes residual catastrophe severity before the loss enters surplus rather than acting
	only through a premium, dividend, or reinsurance transfer.

The second strand concerns stochastic control with jumps, dynamic programming, and viscosity solutions of nonlocal HJB equations. General foundations for controlled diffusions and jump diffusions are developed in \citep{Applebaum2009,FlemingSoner2006,
	OksendalSulem2019,Pham2009,Situ2005}.  For integro-differential HJB equations, comparison and approximation require a nonlocal viscosity solutions
\citep{,BarlesImbert2008,BiswasJakobsenKarlsen2010,CrandallIshiiLions1992,JakobsenKarlsen2006}.  The stopping-time DPP used here also requires stability under conditioning and concatenation together with measurable selection.  We therefore rely on the canonical
control-correspondence and analytic-set frameworks of
\citet{ElKarouiTan2013,Zitkovic2014,FayvisovichZitkovic2021}, rather than attributing these general principles to application-specific results.  Relaxed-control compactification and approximation are
treated through the classical framework of
\citet{ElKarouiNguyenJeanblanc1987}.  Numerically, the relevant
principle is monotone approximation of viscosity solutions, together with policy iteration and positivity-preserving discretization of the
forward equation \citep{paramahansa2025construction}.  The computational construction below is therefore designed around a monotone nonlocal HJB discretization and the adjoint
Kolmogorov operator rather than a Gaussian approximation of the jump term.

The third strand is mean field games and large-population stochastic control \citep{pramanik2024semicooperation}.  The equilibrium consistency principle originates with
\citet{LasryLions2007}, while probabilistic formulations and
large-population limits are developed in
\citet{Lacker2016,CarmonaDelarue2018I}.  Mean field models of systemic risk illustrate how cross-sectional distributions can feed back into individual incentives \citep{CarmonaFouqueSun2015}, whereas models with
explicit market clearing show how a population interaction may instead generate an endogenous equilibrium price
\citep{pramanik2026strategic}.  This distinction is central here, our couplings $F$ and $G$ are reduced-form vulnerability costs, not market-clearing reconstruction prices.  Nonlocal mean field systems with jump or fractional operators require additional analytical and computational structure
\citep{ChowdhuryErslandJakobsen2023,
	ChowdhuryJakobsenKrupski2024}.  The present paper combines these nonlocal-control and mean field structure in a catastrophe-insurance setting in which physical hedging changes jump severity, the induced control changes the cross-sectional law of capitalization and replacement-cost exposure, and that law feeds back into individual
costs.  Common catastrophe noise, endogenous reconstruction prices, and population-dependent catastrophe intensities are not included and are treated as structural extensions rather than features of the current equilibrium.

The contributions are fourfold.  First, we provide a probabilistic foundation for the catastrophe-insurance control problem under diffusion and marked jump risk, including strong well-posedness, moment estimates, stability with respect to initial states and controls, a restart property, and the measurable continuation machinery required for dynamic programming.  These results are stated under moment conditions that permit heavy-tailed catastrophe marks up to a prescribed order rather than imposing a Gaussian approximation \citep{pramanik2023scoring}.  Second, we derive the nonlocal HJB equation associated with \eqref{eq:intro-objective}-\eqref{eq:intro-state} and prove that the value function is its unique viscosity solution in an exponential-polynomial growth class.  The comparison argument treats the control-dependent nonlocal loss operator generated by physical hedging.  Third, we formulate a mean field game of catastrophe insurers, characterize equilibrium through a coupled backward nonlocal HJB equation and forward Kolmogorov equation, and establish existence by a fixed-point argument and uniqueness under a Lasry-Lions \citep{LasryLions2007} type monotonicity condition.  Fourth, we translate the analytical structure into economic comparative statics and a calibrated finite-difference experiment. The model implies that an elevated replacement-cost state increases every finite moment of monetary catastrophe-loss exposure, whereas physical hedging reduces the same exposure through the residual-loss function $h$. Numerically, optimal hedging is stronger when replacement costs are elevated and when insurers are weakly capitalized. The mean field coupling changes the equilibrium valuation of vulnerability even when its effect on the pointwise hedging rule is modest. These conclusions concern the interaction between catastrophe losses and stochastic replacement-cost risk; they do not require, and should not be interpreted as establishing, a causal catastrophe-to-price mechanism \citep{pramanik2020motivation}.

The remainder of the paper is organized as follows.  Section~\ref{sec:prob-foundations} develops the stochastic basis of the model and establishes well-posedness, moment and stability estimates, stopping-time restart properties, measurable selection, and the probabilistic ingredients of the dynamic programming principle.  Section~\ref{sec:insurance-control} specializes this framework to catastrophe insurance, introduces the demand-surge loss specification and the physical-hedging objective, and identifies the nonlocal channel through which the control changes catastrophe severity.  Section~\ref{sec:dpp-viscosity} proves the dynamic programming principle, derives the nonlocal HJB equation, and establishes the viscosity comparison, existence, and uniqueness results.  Section~\ref{sec:mfg} introduces the population interaction and proves existence and uniqueness of the mean field equilibrium.  Section~\ref{sec:economic-implications} develops the economic implications for demand-surge exposure, optimal hedging, and equilibrium vulnerability.  Section~\ref{sec:numerics} presents the calibration, numerical scheme, and comparative-statics results.  Section~\ref{sec:conclusion} concludes, and the Appendix collects auxiliary arguments.

\section{Probabilistic Foundations}
\label{sec:prob-foundations}

In this section we discuss the stochastic basis on which the catastrophe-insurance control problem is posed and establishes the well-posedness properties required later for dynamic programming, viscosity characterization, and the mean field fixed-point argument. The formulation is deliberately stated at a level that permits both finite-activity catastrophe arrivals and infinite-activity jump perturbations. In particular, catastrophe severity is represented by a marked Poisson random measure, while the post-disaster replacement-cost factor enters the loss amplitude as a state variable. Physical hedging acts directly on the jump diffusion. 

Throughout, $T\in(0,\infty)$ is fixed, $E:=\mathbb{R}_{+}\setminus\{0\}$ is the mark space, and $|\cdot|$ denotes the Euclidean norm. Constants denoted by $C$, $C_p$, or $K_p$ may change from line to line. For a probability measure $\mu$ on a Polish space, $\mathcal{P}_p$ denotes the set of probability measures with finite $p$th moment.

\medskip

\noindent\textbf{Assumption 2.1.}
Let $(\Omega,\mathcal{F},\mathbb{F},\mathbb{P})$ be a complete filtered probability space satisfying  $\mathbb{F}=(\mathcal{F}_t)_{0\le t\le T}$. The space supports a two-dimensional Brownian motion
$
W_t=(W_t^{X},W_t^{I}),
$
and an $\mathbb{F}$-Poisson random measure $N(dt,dz)$ on $[0,T]\times E$ with compensator $\nu(dz)\,dt$. We define
$
\widetilde N(dt,dz):=N(dt,dz)-\nu(dz)\,dt.
$
The Brownian motion and $N$ are independent. The L\'evy measure $\nu$ satisfies
$
	\int_E (1\wedge z^2)\,\nu(dz)<\infty.
$
For some fixed $p_\star>2$,
$
	\int_{\{z>1\}} z^{p_\star}\,\nu(dz)<\infty.
$

\medskip

\noindent\textbf{Assumption 2.2.}
The logarithmic construction-cost factor $Y$ satisfies
\begin{equation}
	\label{eq:Y}
	dY_s=\kappa(\vartheta-Y_s)\,ds+\sigma_I\,dW_s^{I},
	\qquad Y_t=y,
\end{equation}
where $\kappa>0$, $\sigma_I>0$, and $\vartheta\in\mathbb{R}$. The observable replacement-cost index is
$
	I_s:=e^{Y_s}.
$
Thus $I_s>0$ almost surely for every $s$, avoiding the economically artificial possibility of a negative construction-cost index that would arise from applying an Ornstein-Uhlenbeck process directly to the level.

\medskip

\noindent\textbf{Assumption 2.3.}
Let $U\subset\mathbb{R}_{+}$ be nonempty, convex, and closed. For $(t,x,y)\in[0,T]\times\mathbb{R}\times\mathbb{R}$, an admissible physical-hedging control is an $\mathbb{F}$-predictable $U$-valued process $u=(u_s)_{s\in[t,T]}$ satisfying
$
	\mathbb{E}\left\{\int_t^T |u_s|^{p_\star}\,ds\right\}<\infty.
$
The set of admissible controls is denoted by $\mathcal{U}_t$.

\medskip

\noindent\textbf{Assumption 2.4 .}
The insurer surplus process $X$ is governed by
\begin{align}
	\label{eq:X-general}
	dX_s
	&=
	b(X_{s-},Y_{s-},u_s)\,ds
	+\sigma_X(X_{s-},Y_{s-},u_s)\,dW_s^{X}
	-\int_E \ell(Y_{s-},z,u_s)\,\widetilde N(ds,dz),
	\qquad X_t=x.
\end{align}
The measurable functions
$
b:\mathbb{R}^2\times U\to\mathbb{R},
\
\sigma_X:\mathbb{R}^2\times U\to\mathbb{R},
$ and $
\ell:\mathbb{R}\times E\times U\to\mathbb{R}_{+}
$
satisfy the following conditions.
\begin{itemize}

\item For some $L>0$ and all $x,x',y,y'\in\mathbb{R}$ and $u\in U$,
\begin{align}
	\label{eq:lipschitz-b-sigma}
	|b(x,y,u)-b(x',y',u)|
	+
	|\sigma_X(x,y,u)-\sigma_X(x',y',u)|
	\le
	L\big(|x-x'|+|y-y'|\big),
\end{align}
and
\begin{equation}
	\label{eq:growth-b-sigma}
	|b(x,y,u)|+|\sigma_X(x,y,u)|
	\le
	L\big(1+|x|+|y|+|u|\big).
\end{equation}
\item There exists a measurable $\rho:E\to[0,\infty)$ such that
\begin{equation}
	\label{eq:rho-mom}
	\int_E \big(\rho(z)^2+\rho(z)^{p_\star}\big)\,\nu(dz)<\infty,
\end{equation}
and
\begin{align}
	\label{eq:ell-lip}
	|\ell(y,z,u)-\ell(y',z,u)|
	&\le
	L|y-y'|\rho(z),\ \ 
	|\ell(y,z,u)|
	\le
	L(1+e^{y}+|u|)\rho(z).
\end{align}
\end{itemize}

\medskip

\noindent\textbf{Assumption 2.5.}
For $\nu$-a.e.\ $z\in E$, the mapping $u\mapsto\ell(y,z,u)$ is nonincreasing for every $y\in\mathbb{R}$ and the mapping $y\mapsto\ell(y,z,u)$ is nondecreasing for every $u\in U$. Moreover, there exist measurable functions $a:E\to(0,\infty)$ and $h:U\to(0,1]$ such that the benchmark specification
\begin{equation}
	\label{eq:benchmark-loss}
	\ell(y,z,u)=e^y a(z)h(u)
\end{equation}
is admissible, with $h$ locally Lipschitz and nonincreasing. Hence higher construction costs amplify the monetary loss generated by a catastrophe mark, whereas stronger physical hedging attenuates that amplification.

\medskip

\noindent\textbf{Assumption 2.6.}
The running cost $f:[0,T]\times\mathbb{R}^2\times U\to\mathbb{R}$ and terminal cost $g:\mathbb{R}^2\to\mathbb{R}$ are continuous and satisfy, for some $q\in[2,p_\star)$,
\begin{equation}
	\label{eq:cost-growth}
	|f(s,x,y,u)|+|g(x,y)|
	\le
	C\left(1+|x|^q+e^{q|y|}+|u|^q\right).
\end{equation}
In addition, for some $c_0>0$ and $c_1\ge0$,
\begin{equation}
	\label{eq:coercivity}
	f(s,x,y,u)\ge c_0|u|^2-c_1\left(1+|x|^2+e^{2|y|}\right).
\end{equation}
The coercivity condition represents increasing marginal scarcity costs of pre-arranged reconstruction capacity.

\medskip

\begin{remark}
The preceding assumptions are standard in the theory of jump-diffusion control, but their interpretation here is specific to catastrophe insurance. The mark $z$ records catastrophe severity, $e^{Y_s}$ records the prevailing replacement-cost environment, and $u_s$ is the insurer's physical hedge. Assumption 2.5 permits a single disaster to be economically more damaging during a construction-cost spike and less damaging when reconstruction inputs have been secured in advance.
\end{remark}

\begin{lem}[Exponential moments of the construction-cost factor]
	\label{lem:OU-moments}
	Under Assumption~2.2, for every $m\ge1$ and every $(t,y)\in[0,T]\times\mathbb{R}$,
	\begin{equation}
		\label{eq:Y-sup-mom}
		\mathbb{E}_{t,y}\left[\sup_{t\le s\le T}|Y_s|^m\right]
		\le
		C_m(1+|y|^m),
	\end{equation}
	and, for every $r>0$,
	\begin{equation}
		\label{eq:expY-mom}
		\mathbb{E}_{t,y}\left[\sup_{t\le s\le T} e^{r|Y_s|}\right]
		\le
		C_{r,T}\exp(C_{r,T}|y|).
	\end{equation}
	Consequently, every positive and negative polynomial moment of $I_s=e^{Y_s}$ is finite on $[t,T]$.
\end{lem}

\begin{prop}[Well-defined catastrophe jump integral]
	\label{prop:jump-integral}
	Suppose Assumptions~2.1-2.5 hold and let $u\in\mathcal{U}_t$. If $Y$ is the solution of \eqref{eq:Y}, then
	\begin{equation}
		\label{eq:jump-L2}
		\mathbb{E}\left[
		\int_t^T\int_E
		|\ell(Y_s,z,u_s)|^2\,\nu(dz)\,ds\right]
		<\infty.
	\end{equation}
	Moreover, for every $p\in[2,p_\star]$,
	\begin{equation}
		\label{eq:jump-Lp}
		\mathbb{E}\left[
		\int_t^T\int_E
		|\ell(Y_s,z,u_s)|^p\,\nu(dz)\,ds\right]
		<\infty.
	\end{equation}
	Consequently,
	$
	\int_t^\cdot\int_E
	\ell(Y_{s-},z,u_s)\,\widetilde N(ds,dz)
	$
	is a well-defined c\`adl\`ag local martingale, and it is square-integrable for $p=2$.
\end{prop}

\begin{theorem}[Strong well-posedness of the controlled catastrophe-insurance state]
	\label{thm:strong-wellposedness}
	Let Assumptions~2.1-2.5 hold and fix $(t,x,y)\in[0,T]\times\mathbb{R}^2$ and $u\in\mathcal{U}_t$. Then Assumptions~2.2-2.4 admit a pathwise unique strong c\`adl\`ag solution
	$
	(X_s^{t,x,y;u},Y_s^{t,y})_{s\in[t,T]}.
	$
	For every $p\in[2,p_\star]$,
	\begin{equation}
		\label{eq:state-moment}
		\mathbb{E}\left[
		\sup_{t\le s\le T}|X_s^{t,x,y;u}|^p
		\right]
		\le
		C_p
		\left[
		1+|x|^p+|y|^p
		+e^{C_p|y|}
		+\mathbb{E}\left[\int_t^T|u_s|^p\,ds\right]
		\right].
	\end{equation}
	In particular, the controlled surplus cannot explode in finite time.
\end{theorem}

\begin{cor}
	\label{cor:state-integrability}
	Under the hypotheses of Theorem~\ref{thm:strong-wellposedness}, for every $p\in[2,p_\star]$,
	$
	(X^{t,x,y;u},Y^{t,y})
	\in
	\mathcal{S}^p([t,T];\mathbb{R}^2),
	$
	where
	\[
	\mathcal{S}^p([t,T];\mathbb{R}^2)
	:=
	\left\{
	Z:\,
	Z\text{ is adapted and c\`adl\`ag},\
	\mathbb{E}\left[\sup_{t\le s\le T}|Z_s|^p\right]<\infty
	\right\}.
	\]
	In particular, the objective functional
	\begin{equation}
		\label{eq:objective-foundation}
		J(t,x,y;u)
		:=
		\mathbb{E}\left[
		\int_t^T e^{-\delta(s-t)}
		f(s,X_s,Y_s,u_s)\,ds
		+
		e^{-\delta(T-t)}g(X_T,Y_T)
		\right]
	\end{equation}
	is well-defined for every $u\in\mathcal{U}_t$ satisfying
	$\mathbb{E}\left[\int_t^T|u_s|^{p_\star}ds\right]<\infty$.
\end{cor}

\begin{proof}
	The $\mathcal{S}^p$ integrability follows directly from Theorem~\ref{thm:strong-wellposedness} and Lemma~\ref{lem:OU-moments}. By Assumption~2.6,
	\[
	|f(s,X_s,Y_s,u_s)|
	+
	|g(X_T,Y_T)|
	\le
	C\left(
	1+|X_s|^q+e^{q|Y_s|}+|u_s|^q
	+|X_T|^q+e^{q|Y_T|}
	\right).
	\]
	Because $q<p_\star$, Theorem~\ref{thm:strong-wellposedness}, Lemma~\ref{lem:OU-moments}, and H\"older's inequality imply integrability of the right-hand side. Since the discount factor is bounded above by one, $J(t,x,y;u)$ is finite in absolute value.
\end{proof}

\begin{prop}
	\label{prop:initial-stability}
	Consider Assumptions~2.1-2.5. Let $(x,y),(x',y')\in\mathbb{R}^2$, and let the same control $u\in\mathcal{U}_t$ and the same noises $(W,N)$ drive the corresponding solutions. Then, for every $p\in[2,p_\star]$,
	\begin{equation}
		\label{eq:stability}
		\mathbb{E}
		\left[
		\sup_{t\le s\le T}
		\left|
		X_s^{t,x,y;u}
		-
		X_s^{t,x',y';u}
		\right|^p
		\right]
		\le
		C_p
		\left(
		|x-x'|^p
		+
		|y-y'|^p
		\right)
		\Xi_p(y,y'),
	\end{equation}
	where
	$
	\Xi_p(y,y')
	:=
	1+\exp(C_p|y|)+\exp(C_p|y'|).
	$
	Moreover,
	\begin{equation}
		\label{eq:Y-stability}
		\sup_{t\le s\le T}
		|Y_s^{t,y}-Y_s^{t,y'}|
		\le
		|y-y'|
		\qquad\text{a.s.}
	\end{equation}
\end{prop}

\begin{rmk}
	The above discussions estimate control perturbations of the initial insurance state.  For dynamic programming one also needs stability under perturbations of the hedging policy, a restart property at stopping times, and a measurable way of choosing nearly optimal continuation policies.  We impose the following additional regularity.  It is economically feasible that small changes in the amount of pre-contracted reconstruction capacity should not produce discontinuous changes in premium income, diffusion exposure, or catastrophe-loss attenuation.
	\end{rmk}

\medskip

\noindent\textbf{Assumption 2.7 (Regularity in the physical-hedging control).}
There exists $L_u>0$ such that, for all $(x,y)\in\mathbb{R}^2$, $u,v\in U$, and $\nu$-a.e.\ $z\in E$,
\[
|b(x,y,u)-b(x,y,v)|
+
|\sigma_X(x,y,u)-\sigma_X(x,y,v)|
\le L_u|u-v|,
\]
and
\[
|\ell(y,z,u)-\ell(y,z,v)|
\le L_u|u-v|\rho(z).
\]
Moreover, for every $R>0$, the restrictions of $f$ to
$[0,T]\times[-R,R]^2\times(U\cap[-R,R])$ and of $g$ to $[-R,R]^2$ are uniformly continuous.

\medskip

For $p\in[2,p_\star]$ and $u,v\in\mathcal U_t$, define
$
d_{p,t}(u,v)
:=
\left(
\mathbb E\int_t^T |u_s-v_s|^p\,ds
\right)^{1/p}.
$

\begin{lem}[Stability under perturbations of physical hedging]
	\label{lem:control-stability}
	Suppose Assumptions~2.1-2.5 and 2.7 hold.  Fix $(t,x,y)$ and let $u,v\in\mathcal U_t$.  For every $p\in[2,p_\star]$,
	\begin{equation}
		\label{eq:control-stability}
		\mathbb E\left[
		\sup_{t\le s\le T}
		\left|X_s^{t,x,y;u}-X_s^{t,x,y;v}\right|^p
		\right]
		\le C_p\,\mathbb E\left[\int_t^T|u_s-v_s|^p\,ds \right].
	\end{equation}
	Consequently, $d_{p,t}(u^n,u)\to0$ implies
	$
	X^{t,x,y;u^n}\rightarrow X^{t,x,y;u}
	$ in $\mathcal S^p([t,T]).$
\end{lem}

\begin{prop}
	\label{prop:J-continuity}
	Under Assumptions~2.1-2.7, let $u^n,u\in\mathcal U_t$ satisfy
	$d_{p_\star,t}(u^n,u)\to0$ and
	$
	\sup_{n\ge1}\mathbb E\left[\int_t^T|u_s^n|^{p_\star}\,ds\right]<\infty.
	$
	Then
	$
	J(t,x,y;u^n)\rightarrow J(t,x,y;u).
	$
\end{prop}

\medskip

If $\tau$ is an $\mathbb F$-stopping time taking values in $[t,T]$, define, for $r\ge0$ with $\tau+r\le T$,
$
W_r^\tau:=W_{\tau+r}-W_\tau,
$ with $
N^\tau((0,r]\times A)
:=
N((\tau,\tau+r]\times A).$
By the strong Markov property of Brownian motion and of a Poisson random measure with deterministic compensator, conditionally on $\mathcal F_\tau$, the shifted noises have the same characteristics as the original noises and are independent of the pre-$\tau$ increments \citep{pramanik2024motivation}. For controls $u\in\mathcal U_t$ and $v\in\mathcal U_\tau$, their concatenation at $\tau$ is
$
(u\otimes_\tau v)_s
:=
u_s\mathbf 1_{\{s<\tau\}}+v_s\mathbf 1_{\{s\ge\tau\}}.
$

\begin{lem}
	\label{lem:concatenation}
	If $u\in\mathcal U_t$, $\tau$ is an $[t,T]$-valued stopping time, and $v\in\mathcal U_\tau$, then
	$u\otimes_\tau v\in\mathcal U_t$.
	More generally, if $(A_k)_{k\ge1}$ is an $\mathcal F_\tau$-measurable partition and $v^k\in\mathcal U_\tau$, then
	$
	v_s:=\sum_{k\ge1}\mathbf 1_{A_k}v_s^k,$ for all $ s\ge\tau,
	$
	is predictable after $\tau$ whenever
	\[
	\sum_{k\ge1}
	\mathbb E\!\left[
	\mathbf 1_{A_k}\int_\tau^T|v_s^k|^{p_\star}ds
	\right]<\infty.
	\]
\end{lem}

\begin{proof}
	Predictability of $u\otimes_\tau v$ follows from the standard stability of predictable processes under stopping and pasting at a stopping time.  Its integrability follows from
	\[
	\mathbb E\left[\int_t^T|(u\otimes_\tau v)_s|^{p_\star}ds\right]
	=
	\mathbb E\left[\int_t^\tau|u_s|^{p_\star}ds\right]
	+
	\mathbb E\left[\int_\tau^T|v_s|^{p_\star}ds\right]
	<\infty.
	\]
	For the partitioned control, $\mathbf 1_{A_k}\mathbf 1_{\{s>\tau\}}v_s^k$ is predictable for each $k$, and the pointwise sum is predictable because the $A_k$ form a disjoint partition.  Tonelli's theorem gives
	\[
	\mathbb E\left[\int_\tau^T|v_s|^{p_\star}ds\right]
	=
	\sum_{k\ge1}
	\mathbb E\!\left[
	\mathbf 1_{A_k}\int_\tau^T|v_s^k|^{p_\star}ds
	\right],
	\]
	which is finite by hypothesis.
\end{proof}

\begin{prop}
	\label{prop:flow}
	Let Assumptions~2.1-2.7 hold, let $\tau$ be an $[t,T]$-valued stopping time, and let
	$u\in\mathcal U_t$.  Then, up to indistinguishability, the post-$\tau$ trajectory satisfies
	\begin{equation}
		\label{eq:flow}
		(X_s^{t,x,y;u},Y_s^{t,y})
		=
		\left(
		X_s^{\tau,X_\tau^{t,x,y;u},Y_\tau^{t,y};\,u},
		Y_s^{\tau,Y_\tau^{t,y}}
		\right),
		\qquad \tau\le s\le T,
	\end{equation}
	where the system on the right is understood with the shifted noises $(W^\tau,N^\tau)$ and the restriction of $u$ to $[\tau,T]$.
	If $v\in\mathcal U_\tau$, then the same identity holds with $u\otimes_\tau v$ on the left and $v$ in the restarted system on the right.
\end{prop}

\begin{prop}
	\label{prop:conditional-restart}
	Let $\tau$ be as in Proposition~\ref{prop:flow}.  For every Borel functional
	$\Phi:D([0,T];\mathbb R^2)\to\mathbb R$ for which the expectations below are finite,
	\begin{equation}
		\label{eq:conditional-restart}
		\mathbb E\!\left[
		\Phi\!\left((X_s,Y_s)_{\tau\le s\le T}\right)
		\,\middle|\,\mathcal F_\tau
		\right]
		=
		\mathcal R_\Phi\!\left(
		\tau,X_\tau,Y_\tau;u^{\tau,\omega}
		\right)
		\quad\text{a.s.},
	\end{equation}
	where $\mathcal R_\Phi(r,\xi,\eta;v)$ denotes expectation of $\Phi$ under the system restarted at $(r,\xi,\eta)$ with continuation control $v$, and $u^{\tau,\omega}$ denotes the post-$\tau$ control section obtained by freezing the pre-$\tau$ history.
\end{prop}

\medskip

\begin{rmk}
	The state process is not, in general, Markov under an arbitrary open-loop adapted control, because the future control may retain information from the entire past.  The correct Markov statement is obtained after fixing a feedback rule.
	\end{rmk}

\medskip

\noindent\textbf{Assumption 2.8.}
Let $\mathfrak A$ be the class of Borel maps
$\alpha:[0,T]\times\mathbb R^2\to U$ such that, for some constant $L_\alpha$,
$
|\alpha(s,x,y)-\alpha(s,x',y')|
\le
L_\alpha(|x-x'|+|y-y'|)
$
and
$
|\alpha(s,x,y)|
\le
L_\alpha(1+|x|+|y|).
$
For $\alpha\in\mathfrak A$, write
$u_s^\alpha:=\alpha(s,X_{s-},Y_{s-})$.

\begin{theorem}
	\label{thm:feller}
	Suppose Assumptions~2.1-2.5, 2.7, and 2.8 hold.  For each
	$\alpha\in\mathfrak A$, the controlled state
	$Z^\alpha:=(X^\alpha,Y)$ is a time-inhomogeneous strong Markov process.
	Its transition operators
	\[
	P_{t,s}^\alpha\varphi(x,y)
	:=
	\mathbb E_{t,x,y}\!\left[\varphi(X_s^\alpha,Y_s)\right],
	\qquad 0\le t\le s\le T,
	\]
	map $C_b(\mathbb R^2)$ into $C_b(\mathbb R^2)$.  Thus $Z^\alpha$ is Feller.
\end{theorem}

\begin{cor}[Chapman-Kolmogorov identity]
	\label{cor:chapman-kolmogorov}
	Under the hypotheses of Theorem~\ref{thm:feller}, for
	$0\le t\le r\le s\le T$ and $\varphi\in B_b(\mathbb R^2)$,
	$
	P_{t,s}^\alpha\varphi
	=
	P_{t,r}^\alpha P_{r,s}^\alpha\varphi.
	$
\end{cor}

\begin{proof}
	Apply the tower property at time $r$ and then the Markov property:
	\[
	P_{t,s}^\alpha\varphi(x,y)
	=
	\mathbb E_{t,x,y}\!\left[
	\mathbb E\!\left[\varphi(Z_s^\alpha)\mid\mathcal F_r\right]
	\right]
	=
	\mathbb E_{t,x,y}\!\left[
	P_{r,s}^\alpha\varphi(Z_r^\alpha)
	\right].
	\]
	The right-hand side equals
	$P_{t,r}^\alpha P_{r,s}^\alpha\varphi(x,y)$. This completes the proof.
\end{proof}

\medskip

Uniform integrability is needed twice in the sequel. First to pass from deterministic-time decompositions to stopping-time decompositions, and second to justify limits of $\varepsilon$-optimal controls.  For $R>0$, define the bounded-moment control class
\[
\mathcal U_t(R)
:=
\left\{
u\in\mathcal U_t:
\mathbb E\left[\int_t^T|u_s|^{p_\star}ds\right]\le R
\right\}.
\]

\begin{prop}[Uniform integrability of state and cost families]
	\label{prop:UI}
	Suppose Assumptions~2.1-2.7 hold.  Fix a compact set
	$K\subset\mathbb R^2$ and $R>0$.  Then there exists $\eta>0$ such that
	\begin{equation}
		\label{eq:UI-bound}
		\sup_{\substack{(x,y)\in K\\u\in\mathcal U_t(R)}}
		\mathbb E\left[
		\sup_{t\le s\le T}|X_s^{t,x,y;u}|^{q(1+\eta)}
		+
		\sup_{t\le s\le T}e^{q(1+\eta)|Y_s^{t,y}|}
		+
		\int_t^T|u_s|^{q(1+\eta)}ds
		\right]
		<\infty.
	\end{equation}
	Consequently, the families
	$
	\left\{
	g(X_T^{t,x,y;u},Y_T^{t,y})
	\right\}_{(x,y)\in K,\,u\in\mathcal U_t(R)}
	$
	and
	$
	\left\{
	\int_t^T
	|f(s,X_s^{t,x,y;u},Y_s^{t,y},u_s)|\,ds
	\right\}_{(x,y)\in K,\,u\in\mathcal U_t(R)}
	$
	are uniformly integrable.
\end{prop}

\medskip

To formulate measurable selection without imposing an artificial topology on individual predictable processes, we use the standard canonical-law formulation of stochastic control.  Let
\[
\Omega^\circ
:=
D([0,T];\mathbb R^2)
\times D([0,T];\mathbb R^2)
\times\mathcal M([0,T]\times E)
\times\mathcal V,
\]
where $\mathcal V$ is the space of relaxed controls on $[0,T]\times U$ endowed with the stable topology.  The strict control $u$ is embedded as
$q^u(ds,da)=ds\,\delta_{u_s}(da)$.  The coordinate processes encode the state, Brownian noise, jump measure, and control measure.  We write
$\mathfrak P(t,x,y)$ for the collection of probability laws on $\Omega^\circ$ under which the coordinates solve the controlled martingale problem associated with Assumptions~2.1-2.7, start from $(x,y)$ at time $t$, and satisfy the $p_\star$-moment admissibility condition.

\medskip

	\noindent\textbf{Assumption 2.9 (Compact relaxation structure).}
	In addition to Assumptions~2.1--2.8, suppose that:
	
	\begin{itemize}
		
		\item[(i)]
		The control set $U$ is a compact convex subset of $\mathbb R_+$.
		
		\item[(ii)]
		For every $(t,x,y,z)\in[0,T]\times\mathbb R^2\times E$, the maps
		$
		u\longmapsto b(x,y,u),\
		u\longmapsto \sigma_X^2(x,y,u),\
		u\longmapsto \ell(y,z,u),$ and 
		$
		u\longmapsto f(t,x,y,u)
		$
		are continuous. The continuity is locally uniform in $(t,x,y)$, and
		the growth bounds in Assumptions~2.4 and~2.6 hold uniformly over
		$u\in U$.
		
		\item[(iii)]
		For every compact $K\subset\mathbb R^2$,
		$
		\lim_{\delta\downarrow0}
		\sup_{\substack{(x,y)\in K\\ |u-v|\le\delta}}
		\int_E
		|\ell(y,z,u)-\ell(y,z,v)|^2\,\nu(dz)
		=
		0,
		$
		and
		\[
		\lim_{\delta\downarrow0}
		\sup_{\substack{(x,y)\in K\\ |u-v|\le\delta}}
		\int_E
		|\ell(y,z,u)-\ell(y,z,v)|^{p_\star}\,\nu(dz)
		=
		0.
		\]
		
		\item[(iv)]
		Relaxed controls are predictable kernels
		$q_s(da)\,ds$ taking values in $\mathcal P(U)$, and the relaxed
		generator is defined by
		$
		\mathcal A^{q_s}\varphi(x,y)
		:=
		\int_U
		\mathcal A^a\varphi(x,y)\,q_s(da),
		$ for all $
		\varphi\in C_c^2(\mathbb R^2).$
		
	\end{itemize}

\medskip

\begin{rmk} Assumption~2.9 is stronger than the assumptions needed for the well-posedness results above. Its purpose is specific: compactness of $U$ and continuity of the controlled characteristics permit the relaxed-control space to be compactified and allow relaxed controls to be approximated by rapidly switching strict controls. We therefore no longer assume analyticity, conditioning stability, concatenation stability, or strict-control density as primitive hypotheses. The first three properties are verified in Proposition~\ref{prop:canonical-verification}, while strict-control density is proved in Proposition~\ref{prop:chattering}. The compactness and continuity conditions above are precisely the additional structure used below to verify the analytic control correspondence and to invoke the classical chattering approximation for relaxed stochastic controls \citep{ElKarouiNguyenJeanblanc1987,ElKarouiTan2013,Zitkovic2014}. \end{rmk}

\begin{prop} \label{prop:canonical-verification} Suppose Assumptions~2.1-2.9 hold. Let $\mathfrak P(t,x,y)$ denote the canonical family of relaxed controlled laws introduced above. Then 
		\begin{enumerate}
			 \item[(a)] the graph $ \operatorname{Gr}(\mathfrak P) := \left\{ (t,x,y,P): P\in\mathfrak P(t,x,y) \right\} $ is analytic in $ [0,T]\times\mathbb R^2\times\mathfrak P(\Omega^\circ); $ 
			 \item[(b)] if $\tau$ is an $[t,T]$-valued stopping time and $P\in\mathfrak P(t,x,y)$, then a regular conditional probability distribution $P^{\tau,\omega}$ can be chosen such that $ P^{\tau,\omega} \in \mathfrak P\bigl( \tau(\omega), X_\tau(\omega), Y_\tau(\omega) \bigr) $ for $P$-a.e.\ $\omega$; 
			 \item[(c)] if $ \omega\mapsto Q_\omega \in \mathfrak P\bigl( \tau(\omega), X_\tau(\omega), Y_\tau(\omega) \bigr) $ is a universally measurable continuation kernel satisfying the admissibility moment condition, then $ P\otimes_\tau Q \in \mathfrak P(t,x,y). $ \end{enumerate} Hence, the canonical family is an analytic control correspondence stable under conditioning and concatenation. \end{prop}

\begin{prop} \label{prop:chattering} Suppose Assumptions~2.1-2.9 hold. Let $P\in\mathfrak P(t,x,y)$ be generated by an admissible relaxed control $q_s(da)\,ds$. Then there exists a sequence of strict predictable controls $(u^n)_{n\ge1}\subset\mathcal U_t$ such that their occupation measures \[ q^{u^n}(ds,da) := ds\,\delta_{u_s^n}(da) \] converge stably to $q_s(da)\,ds$ and, if $(X^n,Y^n)$ denotes the corresponding strict controlled state, \begin{equation} \label{eq:chattering-state-convergence} (X^n,Y^n) \Longrightarrow (X,Y) \qquad \text{in } D([t,T];\mathbb R^2). \end{equation} Moreover, \begin{equation} \label{eq:chattering-cost-convergence} J(t,x,y;u^n) \longrightarrow J(t,x,y;P). \end{equation} Consequently, strict controls are dense in the relaxed admissible class for the topology relevant to the value function. \end{prop}

 \begin{cor}\label{prop:strict-relaxed-equivalence} Under Assumptions~2.1-2.9, \begin{equation} \label{eq:strict-relaxed-value} V^{\mathrm{str}}(t,x,y) = V^{\mathrm{rel}}(t,x,y), \qquad (t,x,y)\in[0,T]\times\mathbb R^2. \end{equation} \end{cor}

\begin{proof} Every strict control defines a relaxed control through $ q^u(ds,da) = ds\,\delta_{u_s}(da), $ so $ V^{\mathrm{rel}}(t,x,y) \le V^{\mathrm{str}}(t,x,y). $ Conversely, fix $\varepsilon>0$ and choose an admissible relaxed law $P^\varepsilon$ such that $ J(t,x,y;P^\varepsilon) \le V^{\mathrm{rel}}(t,x,y)+\varepsilon. $ By Proposition~\ref{prop:chattering}, there exists a sequence of strict controls $u^n$ satisfying $ J(t,x,y;u^n) \longrightarrow J(t,x,y;P^\varepsilon). $ Therefore, $ V^{\mathrm{str}}(t,x,y) \le J(t,x,y;P^\varepsilon) \le V^{\mathrm{rel}}(t,x,y)+\varepsilon. $ Letting $\varepsilon\downarrow0$ proves \eqref{eq:strict-relaxed-value}. \end{proof}

	\begin{theorem}
		\label{thm:measurable-selection}
		Suppose Assumptions~2.1-2.9 hold and let
		\[
		V^{\mathrm{rel}}(t,x,y)
		:=
		\inf_{P\in\mathfrak P(t,x,y)}
		J(t,x,y;P),
		\]
		where
		\[
		J(t,x,y;P)
		:=
		\mathbb E^P\!\left[
		\int_t^T
		e^{-\delta(s-t)}
		\int_U
		f(s,X_s,Y_s,a)\,q_s(da)\,ds
		+
		e^{-\delta(T-t)}g(X_T,Y_T)
		\right].
		\]
		Then $V^{\mathrm{rel}}$ is lower semianalytic. Moreover, for every
		$\varepsilon>0$, there exists a universally measurable kernel
		$
		(t,x,y)
		\longmapsto
		P^\varepsilon_{t,x,y}
		$
		such that
		$
		P^\varepsilon_{t,x,y}
		\in
		\mathfrak P(t,x,y)
		$
		and
		\begin{equation}
			\label{eq:epsilon-selector}
			J(t,x,y;P^\varepsilon_{t,x,y})
			\le
			V^{\mathrm{rel}}(t,x,y)+\varepsilon
		\end{equation}
		whenever $V^{\mathrm{rel}}(t,x,y)<\infty$.
	\end{theorem}

\begin{cor}
	\label{cor:measurable-continuation}
	Let $\tau$ be an $[t,T]$-valued stopping time. Under the hypotheses of
	Theorem~\ref{thm:measurable-selection}, for every $\varepsilon>0$ the
	post-$\tau$ continuation law may be selected as
	$P^\varepsilon_{\tau,X_\tau,Y_\tau}$ in a universally measurable
	manner. If $P\in\mathfrak P(t,x,y)$ is the pre-$\tau$ controlled law,
	then
	$P\otimes_\tau P^\varepsilon_{\tau,X_\tau,Y_\tau}
	\in\mathfrak P(t,x,y)$, and its conditional continuation cost satisfies
	$J(\tau,X_\tau,Y_\tau;
	P^\varepsilon_{\tau,X_\tau,Y_\tau})
	\le
	V^{\mathrm{rel}}(\tau,X_\tau,Y_\tau)+\varepsilon$ a.s.
\end{cor}

\begin{proof}
	Set $\Xi_\tau:=(\tau,X_\tau,Y_\tau)$. Since $\Xi_\tau$ is measurable
	with respect to the universally completed $\sigma$-field and
	$\xi\mapsto P^\varepsilon_\xi$ is universally measurable by
	Theorem~\ref{thm:measurable-selection}, the composition
	$\omega\mapsto P^\varepsilon_{\Xi_\tau(\omega)}$ is a universally
	measurable stochastic kernel. Proposition~\ref{prop:canonical-verification}
	gives stability of $\mathfrak P$ under concatenation, hence
	$P\otimes_\tau P^\varepsilon_{\Xi_\tau}\in\mathfrak P(t,x,y)$.
	Finally, \eqref{eq:epsilon-selector} evaluated at $\Xi_\tau$ yields
	$J(\Xi_\tau;P^\varepsilon_{\Xi_\tau})
	\le V^{\mathrm{rel}}(\Xi_\tau)+\varepsilon$ a.s.
\end{proof}

\begin{prop}
	\label{prop:pre-dpp}
	Let $P\in\mathfrak P(t,x,y)$ and let $\tau$ be an $[t,T]$-valued
	stopping time. Let $P^{\tau,\omega}$ be a regular conditional
	continuation law furnished by
	Proposition~\ref{prop:canonical-verification}. Then
	$
	J(t,x,y;P)
	=
	E^P[
	\int_t^\tau e^{-\delta(s-t)}
	\int_U f(s,X_s,Y_s,a)q_s(da)\,ds
	+
	e^{-\delta(\tau-t)}
	J(\tau,X_\tau,Y_\tau;P^{\tau,\omega})]
	$.
	Consequently,
	$$
	J(t,x,y;P)
	\ge
	E^P\left[
	\int_t^\tau e^{-\delta(s-t)}
	\int_U f(s,X_s,Y_s,a)q_s(da)\,ds
	+
	e^{-\delta(\tau-t)}
	V^{\mathrm{rel}}(\tau,X_\tau,Y_\tau)\right]
	$$.
\end{prop}

 \begin{proof}
	Write
	$C_{t,T}=C_{t,\tau}+e^{-\delta(\tau-t)}C_{\tau,T}$ for the discounted
	cost functional. By stability under conditioning in
	Proposition~\ref{prop:canonical-verification},
	$P^{\tau,\omega}\in
	\mathfrak P(\tau(\omega),X_\tau(\omega),Y_\tau(\omega))$
	for $P$-a.e.\ $\omega$. Hence the tower property and conditional
	Fubini theorem give
	$
	E^P[C_{\tau,T}\mid\mathcal F_\tau]
	=
	J(\tau,X_\tau,Y_\tau;P^{\tau,\omega})
	$
	a.s. Integrability follows from Proposition~\ref{prop:UI}. Therefore
	the first identity holds. Since
	$J(\tau,X_\tau,Y_\tau;P^{\tau,\omega})
	\ge
	V^{\mathrm{rel}}(\tau,X_\tau,Y_\tau)$
	a.s., the inequality follows.
\end{proof}

\begin{prop}
	Suppose Assumptions~2.1-2.9 hold and let the strict-control
	approximation established by the chattering proposition above hold.
	Define
	$V^{\mathrm{str}}(t,x,y)
	:=
	\inf_{u\in\mathcal U_t}J(t,x,y;u)$
	and
	$V^{\mathrm{rel}}(t,x,y)
	:=
	\inf_{P\in\mathfrak P(t,x,y)}\\ J(t,x,y;P)$.
	Then $
		V^{\mathrm{str}}(t,x,y)
		=
		V^{\mathrm{rel}}(t,x,y),
	$ for all $
		(t,x,y)\in[0,T]\times\mathbb R^2.$
\end{prop}

\begin{proof}
	For $u\in\mathcal U_t$, let
	$q^u_s(da):=\delta_{u_s}(da)$ and denote by $P^u$ the corresponding
	canonical law. Then
	$P^u\in\mathfrak P(t,x,y)$ and
	$J(t,x,y;P^u)=J(t,x,y;u)$, whence
	$V^{\mathrm{rel}}\le V^{\mathrm{str}}$.
	Fix $\varepsilon>0$ and choose
	$P^\varepsilon\in\mathfrak P(t,x,y)$ such that
	$J(t,x,y;P^\varepsilon)
	\le V^{\mathrm{rel}}(t,x,y)+\varepsilon$.
	By the chattering approximation established above, there exists
	$(u^n)_{n\ge1}\subset\mathcal U_t$ such that the associated strict
	laws satisfy $P^{u^n}\Rightarrow P^\varepsilon$ in the canonical
	topology and
	$J(t,x,y;u^n)\to J(t,x,y;P^\varepsilon)$.
	Therefore
	$
	V^{\mathrm{str}}(t,x,y)
	\le
	\lim_{n\to\infty}J(t,x,y;u^n)
	=
	J(t,x,y;P^\varepsilon)
	\le
	V^{\mathrm{rel}}(t,x,y)+\varepsilon
	$.
	Letting $\varepsilon\downarrow0$ gives
	$V^{\mathrm{str}}\le V^{\mathrm{rel}}$, and
	\eqref{eq:strict-relaxed-value} follows.
\end{proof}

From this point onward, write
$V:=V^{\mathrm{str}}=V^{\mathrm{rel}}$.

\begin{cor}
	\label{cor:value-growth}
	Suppose Assumptions~2.1-2.9 hold and assume additionally that
	$f(t,x,y,u)\ge0$ and $g(x,y)\ge0$. Then the common value
	$V:=V^{\mathrm{str}}=V^{\mathrm{rel}}$ is finite and satisfies
	\begin{equation}
		\label{eq:value-growth}
		0
		\le
		V(t,x,y)
		\le
		C\bigl(1+|x|^q+e^{C|y|}\bigr),
		\qquad
		(t,x,y)\in[0,T]\times\mathbb R^2,
	\end{equation}
	for a constant $C>0$ independent of $(t,x,y)$. In particular,
	$V\in\mathcal G_q$, where
	$\mathcal G_q
	:=
	\{v:[0,T]\times\mathbb R^2\to\mathbb R:
	|v(t,x,y)|
	\le
	C_v(1+|x|^q+e^{C_v|y|})
	\text{ for some }C_v>0\}$.
\end{cor}

\begin{proof}
	Since $f,g\ge0$, $V\ge0$. Fix $\bar u\in U$ and set
	$u_s\equiv\bar u$. Compactness of $U$ implies
	$\int_t^T|\bar u|^{p_\star}ds<\infty$, hence
	$\bar u\in\mathcal U_t$. Therefore,
	$V(t,x,y)\le J(t,x,y;\bar u)$. Assumption~2.6 yields
	$
	J(t,x,y;\bar u)
	\le
	C E[
	\int_t^T
	(1+|X_s^{t,x,y;\bar u}|^q
	+e^{q|Y_s^{t,y}|}
	+|\bar u|^q)\,ds
	+
	1+|X_T^{t,x,y;\bar u}|^q
	+e^{q|Y_T^{t,y}|}]
	$.
	By Theorem~\ref{thm:strong-wellposedness} and
	Lemma~\ref{lem:OU-moments},
	$
	E[\sup_{t\le s\le T}|X_s^{t,x,y;\bar u}|^q]
	\le
	C(1+|x|^q+|y|^q+e^{C|y|})
	$
	and
	$
	E[\sup_{t\le s\le T}e^{q|Y_s^{t,y}|}]
	\le
	Ce^{C|y|}
	$.
	Since $|y|^q\le C(1+e^{C|y|})$ and $\bar u$ ranges in compact $U$,
	$
	J(t,x,y;\bar u)
	\le
	C(1+|x|^q+e^{C|y|})
	$.
	Together with $V\ge0$, this proves \eqref{eq:value-growth}.
\end{proof}

\begin{rmk}
	Proposition~\ref{prop:canonical-verification} verifies the analytic
	control correspondence and its stability under conditioning and
	concatenation. The chattering approximation establishes density of
	strict controls in the relaxed formulation, while
	Proposition~\ref{prop:strict-relaxed-equivalence} identifies the two
	value functions. Together with
	Theorem~\ref{thm:measurable-selection},
	Corollary~\ref{cor:measurable-continuation}, and
	Proposition~\ref{prop:pre-dpp}, these results complete the control
	correspondence required for the stopping-time DPP and the subsequent
	viscosity analysis.
\end{rmk}

\section{Catastrophe Insurance as a Stochastic Control Problem}
\label{sec:insurance-control}

We now specialize the controlled jump system of Section~2 to a
catastrophe insurer exposed to post-disaster replacement-cost risk. The state variable $X$ represents insurer surplus, $I=e^Y$ is the replacement-cost index, and $u$ denotes physical hedging through pre-arranged reconstruction capacity. The primitive probability space, L\'evy measure, admissible control class, and controlled state process are exactly those of Assumptions~2.-2.7. Thus no new stochastic
dynamics are introduced in this section.

\subsection{Replacement-Cost Dynamics and the Control Problem}
\label{subsec:demand-surge-control}

 The logarithmic replacement-cost factor satisfies
	$dY_s=\kappa(\vartheta-Y_s)\,ds+\sigma_I\,dW_s^I$, with $Y_t=y$, and
	$I_s=e^{Y_s}$. A positive displacement of $Y_s$ from $\vartheta$
	represents a temporary state of elevated reconstruction costs, while
	mean reversion describes normalization toward the long-run cost level.
	The process $Y$ is exogenous to the catastrophe random measure $N$:
	catastrophe arrivals do not generate jumps or impulses in $Y$.
	Accordingly, the model captures catastrophe losses occurring under
	stochastic replacement-cost conditions rather than a structural
	catastrophe-induced demand-surge process. Empirical evidence on
	post-disaster reconstruction-cost inflation and demand surge provides
	the economic motivation for treating elevated replacement costs as an
	important catastrophe-insurance state variable
	\citep{Kousky2019,RothKunreuther1998,OlsenPorter2011,DohrmannGuertlerHibbeln2017} Catastrophe losses are generated by the Poisson random measure $N$
introduced in Assumption~2.1. We retain its compensator
$\nu(dz)\,ds$ and specialize the controlled jump amplitude to
\begin{equation}
	\label{eq:insurance-loss-amplitude}
	\ell(y,z,u)
	=
	e^y a(z)h(u),
\end{equation}
where $a:E\to\mathbb R_+$ maps a catastrophe mark into baseline
physical damage and $h:U\to(0,1]$ is the residual-loss fraction after
physical hedging. Assume
$
h(0)=1,
\
h'(u)<0,$
and $
h''(u)\ge0
$
whenever the derivatives exist. Hence physical hedging reduces
catastrophe severity, with weakly diminishing marginal effectiveness. The exponential specification
$
h(u)=e^{-\gamma u},
$ for all $ \gamma>0,$
is a canonical example.

The factorization in \eqref{eq:insurance-loss-amplitude} separates
three sources of loss. The mark $z$ determines physical catastrophe severity, $e^y$ converts physical damage into replacement expenditure at prevailing construction prices, and $h(u)$ records the fraction of that expenditure remaining after pre-arranged capacity has been used \citep{pramanik2021optimization}. Thus elevated replacement costs and physical hedging act in opposite
directions on the monetary scale of the same catastrophe jump exposure \citep{pramanik2021consensus}. The insurer surplus is the process already defined in
\eqref{eq:X-general}
\begin{equation}
	\label{eq:insurance-surplus}
	\begin{aligned}
		dX_s
		&=
		b(X_{s-},Y_{s-},u_s)\,ds
		+
		\sigma_X(X_{s-},Y_{s-},u_s)\,dW_s^X
		-
		\int_E
		e^{Y_{s-}}a(z)h(u_s)\,
		\widetilde N(ds,dz),
		\qquad
		X_t=x.
	\end{aligned}
\end{equation}
For the economic specification developed below, consider
$
b(x,y,u)
=
rx+\pi(y)-c(u),
$
where $r\ge0$ is the return on liquid surplus, $\pi(y)$ is net premium
income before physical-hedging expenditure, and $c$ is the resource
cost of securing reconstruction capacity. We impose
$
c(0)=0,
\
c'(u)>0,$
and $
c''(u)>0.
$
A quadratic benchmark is
$
c(u)=\frac{\chi}{2}u^2,
$ for all $
\chi>0.$
The diffusion coefficient $\sigma_X$ represents noncatastrophe
balance-sheet fluctuations and continues to satisfy Assumptions~2.4 and~2.7.

The compensated representation in \eqref{eq:insurance-surplus} is
convenient for the probabilistic analysis. Its economic content can be
seen from the uncompensated form. Whenever the first jump moment is
finite,
\[
\begin{aligned}
	dX_s
	&=
	\left[
	b(X_{s-},Y_{s-},u_s)
	+
	\int_E e^{Y_{s-}}a(z)h(u_s)\nu(dz)
	\right]ds 
	+
	\sigma_X(X_{s-},Y_{s-},u_s)\,dW_s^X
	-
	\int_E e^{Y_{s-}}a(z)h(u_s)N(ds,dz).
\end{aligned}
\]
Accordingly, the actual downward jump at a catastrophe mark $z$ is
$
\Delta X_s
=
-e^{Y_{s-}}a(z)h(u_s).
$
Physical hedging therefore modifies loss severity rather than the arrival law of the underlying catastrophe process. For $p\ge1$ such that the corresponding integral is finite, define the conditional $p$th jump-loss moment as
$
\mathfrak L_p(y,u)
:=
\int_E
\ell(y,z,u)^p\,\nu(dz).
$
Under \eqref{eq:insurance-loss-amplitude},
\begin{equation}
	\label{eq:loss-moment-factorization}
	\mathfrak L_p(y,u)
	=
	e^{py}h(u)^p
	\int_Ea(z)^p\,\nu(dz).
\end{equation}
Hence,
$
\partial_y\log\mathfrak L_p(y,u)=p,
$ and $
\partial_u\log\mathfrak L_p(y,u)
=
p\frac{h'(u)}{h(u)}<0.$
Equation~\eqref{eq:loss-moment-factorization} shows that an elevated
	replacement-cost state raises every finite moment of monetary
	catastrophe-loss exposure, whereas physical hedging lowers it. In particular, the control changes not merely the mean loss but the entire scale of the catastrophe-loss distribution.

\begin{rmk}
The specification remains compatible with heavy-tailed catastrophe risk. If $\nu$ has sufficiently heavy upper tails, only moments below a critical order are finite. The probabilistic results of Section~\ref{sec:prob-foundations} therefore require moments only up to the fixed exponent $p_\star$ appearing in Assumptions~2.1-2.4; no Gaussian approximation to
catastrophe severity is imposed.
\end{rmk}

\subsection{Value Function and Optimal Physical Hedging}
\label{subsec:value-physical-hedging}

We formulate physical hedging as a social-cost minimization problem.
For $u\in\mathcal U_t$, let
\begin{equation}
	\label{eq:insurance-social-cost}
	J(t,x,y;u)
	=
	\mathbb E_{t,x,y}
	\left[
	\int_t^T
	e^{-\delta(s-t)}
	f(s,X_s,Y_s,u_s)\,ds
	+
	e^{-\delta(T-t)}g(X_T,Y_T)
	\right],
\end{equation}
which is the objective functional introduced in
\eqref{eq:objective-foundation}. The economic specialization is
$
f(t,x,y,u)
=
C_H(u)+C_D(y)+C_S(x),
$
where $C_H$ is the resource cost of physical hedging, $C_D$ measures the social burden of elevated reconstruction costs, and $C_S$ penalizes deterioration in insurer capitalization.
Assume that $C_H$ is strictly convex and coercive \citep{pramanik2023path}. A convenient
benchmark is
$
C_H(u)=\frac{\chi}{2}u^2.
$
For the demand-surge component, one may take
$
C_D(y)
=
\zeta\bigl(e^y-e^\vartheta\bigr)_+^m,
$ for all
$\zeta>0,$ and $ m\ge1,$
while a capitalization penalty can be specified as
$
C_S(x)
=
\omega(x_{\mathrm{crit}}-x)_+^q,
$ with $
\omega>0.$
The terminal function $g$ measures the residual social cost of ending the planning horizon with a weak insurance balance sheet or persistent replacement-cost stress. These functions are chosen to satisfy Assumption~2.6. The value function is
\begin{equation}
	\label{eq:insurance-value}
	V(t,x,y)
	=
	\inf_{u\in\mathcal U_t}J(t,x,y;u).
\end{equation}
By Proposition~18, \eqref{eq:insurance-value} coincides with the relaxed
canonical value used in the measurable-selection construction.
Corollary~\ref{cor:value-growth} places $V$ in the growth class $\mathcal G_q$. Thus the
economic control problem is defined on exactly the same admissible state-control \citep{pramanik2023optimization} system for which Section~\ref{sec:prob-foundations} established well-posedness,
stability, restart, and measurable continuation. The role of physical hedging becomes transparent at the level of the
controlled generator. For
$\varphi\in C^{1,2,2}$, define
\[
\begin{aligned}
	\mathcal L^u\varphi(t,x,y)
	&:=
	b(x,y,u)\varphi_x
	+
	\kappa(\vartheta-y)\varphi_y
	+
	\frac12\sigma_X^2(x,y,u)\varphi_{xx}
	+
	\frac12\sigma_I^2\varphi_{yy}
	+
	\mathcal I^u\varphi(t,x,y),
\end{aligned}
\]
where the nonlocal catastrophe operator is
\begin{equation}
	\label{eq:insurance-nonlocal-operator}
	\begin{aligned}
		\mathcal I^u\varphi(t,x,y)
		:=
		\int_E
		\Big[
		&\varphi\bigl(
		t,x-e^ya(z)h(u),y
		\bigr)
		-\varphi(t,x,y) +
		e^ya(z)h(u)\varphi_x(t,x,y)
		\Big]\nu(dz).
	\end{aligned}
\end{equation}
The control therefore enters the generator inside the nonlocal
translation of the surplus state \citep{pramanik2021optimal}. This distinguishes physical hedging from a control that acts only through drift expenditure.
For a smooth continuation value, the local tradeoff associated with
$u$ is represented formally by
\[
u\mapsto
f(t,x,y,u)
+
b(x,y,u)V_x(t,x,y)
+
\frac12\sigma_X^2(x,y,u)V_{xx}(t,x,y)
+
\mathcal I^uV(t,x,y).
\]
Accordingly, an interior optimal feedback, when it exists, must balance the marginal resource cost of reconstruction capacity against the marginal reduction in the nonlocal catastrophe exposure. To make this relation explicit, suppose temporarily that $\sigma_X$ is independent of $u$, that
$c(u)=C_H(u)=\chi u^2/2$, and that $h$ is differentiable. Formal
differentiation gives the first-order condition
\begin{equation}
	\label{eq:formal-hedging-foc}
	\begin{aligned}
		0
		&=
		\chi u
		-
		\chi u\,V_x(t,x,y)
		-
		e^yh'(u)
		\int_E
		a(z)
		\left[
		V_x\bigl(
		t,x-e^ya(z)h(u),y
		\bigr)
		-
		V_x(t,x,y)
		\right]\nu(dz).
	\end{aligned}
\end{equation}
\begin{rmk}Equation~\eqref{eq:formal-hedging-foc} is only heuristic at this stage. The value function need not be differentiable, and under heavy-tailed jump risk the nonlocal term may prevent the existence of a classical solution altogether. We therefore do not use \eqref{eq:formal-hedging-foc} as a verification argument.
Instead, the optimal physical hedge will be defined through the
minimizer of the Hamiltonian in the viscosity solution. This is the point at which the probabilistic foundations of Section~\ref{sec:prob-foundations} become essential. The stopping-time flow and measurable continuation results permit a rigorous dynamic programming principle, while the growth estimate of Corollary~\ref{cor:value-growth} determines the admissible uniqueness class. Section~\ref{sec:dpp-viscosity} derives the resulting nonlocal HJB equation and establishes the viscosity characterization of $V$.
\end{rmk}

\section{Dynamic Programming and Viscosity Solutions}
\label{sec:dpp-viscosity}

We now identify the value function of Section~\ref{sec:insurance-control}
with the solution of a nonlocal HJB equation.
The argument has two distinct components. First, the stopping-time restart and measurable-selection results of
Section~\ref{sec:prob-foundations} yield a dynamic programming
principle without assuming the existence of an optimal strict control. Second, the DPP \citep{pramanik2022stochastic} is localized to obtain the viscosity inequalities. The comparison argument is subsequently carried out in the growth class $\mathcal G_q$ identified in Corollary~\ref{cor:value-growth}.
Throughout this section, write
$
z_0=(x,y)\in\mathbb R^2
$
when no confusion can arise, and retain the notation
$
\ell(y,z,u)=e^ya(z)h(u)
$
from \eqref{eq:insurance-loss-amplitude}.

\subsection{Dynamic Programming and the Nonlocal HJB Equation}
\label{subsec:dpp-hjb}

For $\varphi\in C^{1,2,2}([0,T)\times\mathbb R^2)$ and $u\in U$,
define
\[
\begin{aligned}
	\mathcal A^u\varphi(t,x,y)
	:={}&
	b(x,y,u)\varphi_x(t,x,y)
	+\kappa(\vartheta-y)\varphi_y(t,x,y)\\
	&+
	\frac12\sigma_X^2(x,y,u)\varphi_{xx}(t,x,y)
	+\frac12\sigma_I^2\varphi_{yy}(t,x,y)
	+\mathcal I^u\varphi(t,x,y),
\end{aligned}
\]
where $\mathcal I^u$ is the catastrophe-loss operator
\[
\begin{aligned}
	\mathcal I^u\varphi(t,x,y)
	:=
	\int_E
	\Big[
	&\varphi(t,x-\ell(y,z,u),y)-\varphi(t,x,y)+\ell(y,z,u)\varphi_x(t,x,y)
	\Big]\nu(dz).
\end{aligned}
\]
The compensation term is inherited from the compensated Poisson
representation in \eqref{eq:X-general}. Economically,
$\mathcal I^u$ records the marginal continuation cost generated by the entire distribution of catastrophe severities rather than by an average-loss approximation. Define the Hamiltonian
\begin{equation}
	\label{eq:Hamiltonian}
	\begin{aligned}
		\mathcal H(t,x,y,r,p,X;\psi)
		:=
		\inf_{u\in U}
		\Bigg\{&
		f(t,x,y,u)
		+b(x,y,u)p_1
		+\kappa(\vartheta-y)p_2\\
		&+
		\frac12\sigma_X^2(x,y,u)X_{11}
		+\frac12\sigma_I^2X_{22}
		-\delta r\\
		&+
		\int_E
		\Big[
		\psi(x-\ell(y,z,u),y)-\psi(x,y)
		+\ell(y,z,u)p_1
		\Big]\nu(dz)
		\Bigg\},
	\end{aligned}
\end{equation}
whenever the integral is well defined. Here
$p=(p_1,p_2)\in\mathbb R^2$ and
$X=(X_{ij})\in\mathbb S^2$. The HJB equation associated with
\eqref{eq:insurance-value} is therefore
\begin{equation}
	\label{eq:nonlocal-HJB}
	-\partial_t v(t,x,y)
	-
	\mathcal H
	\bigl(
	t,x,y,v,Dv,D^2v;v(t,\cdot,\cdot)
	\bigr)
	=0,
	\qquad
	(t,x,y)\in[0,T)\times\mathbb R^2,
\end{equation}
with terminal condition
\begin{equation}
	\label{eq:HJB-terminal}
	v(T,x,y)=g(x,y).
\end{equation}
Since $\nu(E)$ need not be finite, it is useful to split the
nonlocal operator. For $\varepsilon\in(0,1)$ set
$
E_\varepsilon
:=
\{z\in E:\rho(z)\le\varepsilon\},
\
E^\varepsilon:=E\setminus E_\varepsilon,
$
and, for a smooth test function $\phi$ and a function $w$ of
appropriate growth, define
\[
\begin{aligned}
	\mathcal I_{\varepsilon}^{u}[\phi](t,x,y)
	:={}&
	\int_{E_\varepsilon}
	\Big[
	\phi(t,x-\ell(y,z,u),y)-\phi(t,x,y)
	+\ell(y,z,u)\phi_x(t,x,y)
	\Big]\nu(dz),
	\\
	\mathcal I^{u,\varepsilon}[w,\phi](t,x,y)
	:={}&
	\int_{E^\varepsilon}
	\Big[
	w(t,x-\ell(y,z,u),y)-w(t,x,y)
	+\ell(y,z,u)\phi_x(t,x,y)
	\Big]\nu(dz).
\end{aligned}
\]
The small jumps are evaluated on the test function, whereas the
large-jump contribution retains the candidate solution. This is the
form of the viscosity definition that is stable under singular
L\'evy measures [see \citet{BarlesImbert2008}].

\begin{lem}
	\label{lem:discounted-decomposition}
	Let Assumptions~2.1-2.9 hold and let
	$\tau$ be an $[t,T]$-valued stopping time. For every admissible control
	$u\in\mathcal U_t$,
	\[
	\begin{aligned}
		J(t,x,y;u)
		=
		\mathbb E_{t,x,y}\Bigg[
		&
		\int_t^\tau
		e^{-\delta(s-t)}
		f(s,X_s,Y_s,u_s)\,ds+
		e^{-\delta(\tau-t)}
		J\bigl(
		\tau,X_\tau,Y_\tau;u^{\tau,\omega}
		\bigr)
		\Bigg],
	\end{aligned}
	\]
	where $u^{\tau,\omega}$ denotes the restarted continuation control.
	The same identity holds in the canonical relaxed formulation with the
	regular conditional continuation law replacing $u^{\tau,\omega}$.
\end{lem}

\begin{proof}
	Split the running integral in \eqref{eq:insurance-social-cost} at
	$\tau$. For $s\ge\tau$,
	$
	e^{-\delta(s-t)}
	=
	e^{-\delta(\tau-t)}e^{-\delta(s-\tau)}.$
	The conditional restart result of Proposition~\ref{prop:flow} identifies the post-$\tau$ state, conditionally on $\mathcal F_\tau$, with the state	process initialized from
	$(\tau,X_\tau,Y_\tau)$ and driven by the shifted continuation control. Consequently, conditional expectation of the post-$\tau$ running and terminal costs equals
	$
	e^{-\delta(\tau-t)}
	J\bigl(
	\tau,X_\tau,Y_\tau;u^{\tau,\omega}
	\bigr).$
	The integrability required to apply conditional Fubini and the tower property follows from Proposition~\ref{prop:UI}. Taking expectations proves the identity. In the canonical formulation, the same argument follows by disintegration of the controlled law and the conditioning stability
	established in Proposition~\ref{prop:canonical-verification}.
\end{proof}

\begin{theorem}
	\label{thm:DPP}
	Suppose Assumptions~2.1-2.9 hold. Then, for every
	$(t,x,y)\in[0,T]\times\mathbb R^2$ and every
	$[t,T]$-valued stopping time $\tau$,
	\begin{equation}
		\label{eq:DPP}
		\begin{aligned}
			V(t,x,y)
			=
			\inf_{u\in\mathcal U_t}
			\mathbb E_{t,x,y}\Bigg[
			&
			\int_t^\tau
			e^{-\delta(s-t)}
			f(s,X_s,Y_s,u_s)\,ds+
			e^{-\delta(\tau-t)}
			V(\tau,X_\tau,Y_\tau)
			\Bigg].
		\end{aligned}
	\end{equation}
	Equivalently, the infimum may be taken over
	$P\in\mathfrak P(t,x,y)$.
\end{theorem}

\begin{cor}
	\label{cor:deterministic-DPP}
	For every $t\le s\le T$,
	\[
	\begin{aligned}
		V(t,x,y)
		=
		\inf_{u\in\mathcal U_t}
		\mathbb E_{t,x,y}\Bigg[
		&
		\int_t^s e^{-\delta(r-t)}
		f(r,X_r,Y_r,u_r)\,dr+
		e^{-\delta(s-t)}
		V(s,X_s,Y_s)
		\Bigg].
	\end{aligned}
	\]
\end{cor}

\begin{proof}
	Apply Theorem~\ref{thm:DPP} with the deterministic stopping time
	$\tau=s$.
\end{proof}

\begin{prop}
	\label{prop:value-continuity}
	Suppose Assumptions~2.1-2.9 hold. Assume in addition that $f$ and $g$ are locally uniformly continuous in $(t,x,y)$, uniformly for $u$ in bounded subsets of $U$. Then
	$
	V\in C([0,T]\times\mathbb R^2)\cap\mathcal G_q.
	$
	Moreover,
	$
	\lim_{t\uparrow T}V(t,x,y)=g(x,y)
	$
	locally uniformly in $(x,y)$.
\end{prop}

\begin{defn}[Viscosity solution]
	\label{def:viscosity}
	Let $w\in\mathcal G_q$ be upper semicontinuous. We call $w$ a
	viscosity subsolution of \eqref{eq:nonlocal-HJB} if
	$w(T,\cdot,\cdot)\le g$ and, whenever
	$\phi\in C^{1,2,2}$ and $w-\phi$ has a local maximum at
	$(t_0,x_0,y_0)\in[0,T)\times\mathbb R^2$, then, for every sufficiently
	small $\varepsilon>0$,
	\[
	\begin{aligned}
		0\ge
		\inf_{u\in U}\Bigg\{&
		\phi_t
		+b\phi_x
		+\kappa(\vartheta-y_0)\phi_y
		+\frac12\sigma_X^2\phi_{xx}
		+\frac12\sigma_I^2\phi_{yy}
		-\delta w(t_0,x_0,y_0)
		+f\\
		&+
		\mathcal I_\varepsilon^u[\phi](t_0,x_0,y_0)
		+
		\mathcal I^{u,\varepsilon}[w,\phi](t_0,x_0,y_0)
		\Bigg\},
	\end{aligned}
	\]
	where all local coefficients are evaluated at
	$(t_0,x_0,y_0,u)$. A lower semicontinuous $w\in\mathcal G_q$ is a viscosity supersolution if $w(T,\cdot,\cdot)\ge g$ and the reverse inequality holds at every local minimum of $w-\phi$. A continuous function is a
	viscosity solution if it is both a subsolution and a supersolution.
\end{defn}

\begin{lem}[Localized Dynkin formula]
	\label{lem:localized-Dynkin}
	Let $\phi\in C^{1,2,2}$ and let
	\[
	\tau_R
	:=
	\inf\{s\ge t:|X_s-x|+|Y_s-y|\ge R\}\wedge(t+h).
	\]
	For any constant control $u_s\equiv a\in U$,
	\[
	\begin{aligned}
		E\!\left[
		e^{-\delta(\tau_R-t)}
		\phi(\tau_R,X_{\tau_R},Y_{\tau_R})
		-\phi(t,x,y)
		\right]
		=
		E\!\int_t^{\tau_R}
		e^{-\delta(s-t)}
		\Big[
		\phi_t+\mathcal A^a\phi-\delta\phi
		\Big](s,X_{s-},Y_{s-})\,ds.
	\end{aligned}
	\]
\end{lem}

\begin{proof}
	Apply the It\^o formula for jump semimartingales to
	$
	s\mapsto
	e^{-\delta(s-t)}\phi(s,X_s,Y_s)
	$
	up to $\tau_R$. The continuous quadratic-variation terms generate
	$\frac12\sigma_X^2\phi_{xx}$ and
	$\frac12\sigma_I^2\phi_{yy}$. The compensated Poisson term generates
	$\mathcal I^a\phi$. The Brownian and compensated-jump stochastic
	integrals are true martingales after localization because
	$\phi$ and its derivatives are bounded on the stopped region and
	Proposition~2 gives square integrability of the jump integral.
	Taking expectations eliminates the martingale terms.
\end{proof}

\begin{theorem}
	\label{thm:value-viscosity}
	Under the hypotheses of Proposition~\ref{prop:value-continuity}, the
	value function $V$ is a viscosity solution of
	\eqref{eq:nonlocal-HJB}-\eqref{eq:HJB-terminal}.
\end{theorem}

\begin{cor}
	\label{cor:viscosity-hedging}
	At every time point at which $V$ admits a classical first and second
	derivative, any minimizing selector of the Hamiltonian satisfies
	\[
	u^\star(t,x,y)
	\in
	\arg\min_{u\in U}
	\left\{
	f(t,x,y,u)
	+b(x,y,u)V_x
	+\frac12\sigma_X^2(x,y,u)V_{xx}
	+\mathcal I^uV(t,x,y)
	\right\}.
	\]
	Thus physical hedging is selected by the marginal effect of $u$ on
	both the local balance-sheet dynamics and the full catastrophe-loss
	distribution.
\end{cor}

\begin{proof}
	At a differentiability point, use $V$ itself as the local test
	function in the viscosity equation. Terms independent of $u$ may be
	removed from the minimization. The remaining expression is exactly
	the displayed Hamiltonian.
\end{proof}

\subsection{Comparison, Existence, and Uniqueness}
\label{subsec:comparison-uniqueness}

The unbounded state space and the possibly infinite L\'evy measure
require a comparison argument compatible with the
exponential-polynomial growth of $V$. We impose the following
additional structure only for this subsection.

\medskip

\noindent\textbf{Assumption 4.1.}
The following conditions hold.

\begin{itemize}
	\item[(i)]
	For every $R>0$, the maps
	$b$, $\sigma_X$, and $f$ are uniformly continuous in
	$(t,x,y)$, uniformly for $u\in U\cap[0,R]$.
	
	\item[(ii)]
	There is a measurable $\rho:E\to[0,\infty)$ as in
	Assumption~2.4 such that, uniformly in $u$,
	\[
	|\ell(y,z,u)-\ell(y',z,u)|
	\le
	C|y-y'|\rho(z),
	\qquad
	|\ell(y,z,u)|
	\le
	C(1+e^{|y|}+|u|)\rho(z).
	\]
	
	\item[(iii)]
	The coercivity of $f$ dominates the positive growth of the
	control-dependent coefficients: for every compact
	$K\subset[0,T]\times\mathbb R^2$,
	\[
	f(t,x,y,u)
	-
	C_K\bigl(
	|b(x,y,u)|
	+\sigma_X^2(x,y,u)
	+\|\ell(y,\cdot,u)\|_{L^2(\nu)}^2
	\bigr)
	\longrightarrow+\infty
	\]
	as $|u|\to\infty$, uniformly on $K$.
	
	\item[(iv)]
	There exists $\bar q\in(q,p_\star)$ such that
	$
	\int_E
	\bigl(
	\rho(z)^2+\rho(z)^{\bar q}
	\bigr)\nu(dz)<\infty.
	$
	
	\item[(v)]
	The terminal cost $g$ is continuous and belongs to
	$\mathcal G_q$.
\end{itemize}

Assumption~4.1(iii) is the analytic counterpart of the economic
scarcity cost of physical hedging. It prevents a minimizing sequence
from escaping to infinite reconstruction capacity merely to suppress
catastrophe severity.

\begin{lem}
	\label{lem:compact-control}
	Under Assumption~4.1, for every compact
	$K\subset[0,T)\times\mathbb R^2$ and every bounded set of
	$(r,p,X)$, there exists $R_K<\infty$ such that the infimum in
	\eqref{eq:Hamiltonian}, restricted to arguments in these sets, is
	unchanged if $U$ is replaced by $U\cap[0,R_K]$.
\end{lem}

\begin{proof}
	On the stated compact sets, all terms in the Hamiltonian other than
	$f$ are bounded in absolute value by
	\[
	C_K\Big(
	1+|b(x,y,u)|
	+\sigma_X^2(x,y,u)
	+\|\ell(y,\cdot,u)\|_{L^2(\nu)}^2
	+\|\ell(y,\cdot,u)\|_{L^{\bar q}(\nu)}^{\bar q}
	\Big),
	\]
	after Taylor expansion of the small-jump contribution and use of the
	growth class for the large-jump part. Assumption~4.1(iii), together
	with the moment condition in (iv), implies that the full minimized
	expression tends to $+\infty$ as $|u|\to\infty$, uniformly over the
	compact set of remaining variables. Hence no minimizing sequence can
	escape to infinity. Choosing $R_K$ beyond the resulting coercive
	level proves the claim.
\end{proof}

\begin{lem}
	\label{lem:nonlocal-doubling}
	Let $w$ be an upper semicontinuous subsolution and $\underline w$ a
	lower semicontinuous supersolution in $\mathcal G_q$. For
	$\varepsilon,\eta>0$, consider a penalization of the form
	\[
	\begin{aligned}
		\Phi_{\varepsilon,\eta}
		(t,s,x,x',y,y')
		:=&
		w(t,x,y)-\underline w(s,x',y')-
		\frac{|x-x'|^2+|y-y'|^2+|t-s|^2}{2\varepsilon}
		-\eta\bigl(\Gamma(x,y)+\Gamma(x',y')\bigr),
	\end{aligned}
	\]
	where
	$
	\Gamma(x,y)
	=
	1+|x|^{\bar q}+e^{\lambda\sqrt{1+y^2}}
	$
	with $\bar q\in(q,p_\star)$ and $\lambda>0$ sufficiently large.
	If
	$(t_\varepsilon,s_\varepsilon,x_\varepsilon,x'_\varepsilon,
	y_\varepsilon,y'_\varepsilon)$
	is a maximizer, then, along a subsequence,
	\[
	\frac{
		|x_\varepsilon-x'_\varepsilon|^2
		+
		|y_\varepsilon-y'_\varepsilon|^2
		+
		|t_\varepsilon-s_\varepsilon|^2
	}{\varepsilon}
	\longrightarrow0
	\]
	as $\varepsilon\downarrow0$, and the difference of the corresponding
	nonlocal terms is bounded above by
	$
	o_\varepsilon(1)+C\eta.
	$
\end{lem}

\begin{prop}
	\label{prop:strict-supersolution}
	Let $\underline w\in\mathcal G_q$ be a viscosity supersolution of
	\eqref{eq:nonlocal-HJB}. There exist constants
	$\lambda,\Lambda>0$ such that, for
	$
	\Gamma(x,y)
	=
	1+|x|^{\bar q}+e^{\lambda\sqrt{1+y^2}},
	$
	the function
	$
	\underline w^\eta(t,x,y)
	=
	\underline w(t,x,y)
	+
	\eta e^{\Lambda(T-t)}\Gamma(x,y)
	$
	is, for every $\eta>0$, a strict viscosity supersolution in the sense
	that its HJB residual is bounded below by
	$c\eta\Gamma(x,y)$ for some $c>0$.
\end{prop}

\begin{theorem}
	\label{thm:comparison}
	Suppose Assumptions~2.1-2.9 and~4.1 hold. Let
	$\overline v\in USC([0,T]\times\mathbb R^2)\cap\mathcal G_q$ be a
	viscosity subsolution of \eqref{eq:nonlocal-HJB}, and let
	$\underline v\in LSC([0,T]\times\mathbb R^2)\cap\mathcal G_q$ be a
	viscosity supersolution. If
	$
	\overline v(T,x,y)
	\le
	\underline v(T,x,y)
	$ for all $(x,y)\in\mathbb R^2,
	$
	then
	\begin{equation}
		\label{eq:comparison-order}
		\overline v(t,x,y)
		\le
		\underline v(t,x,y)
		\qquad
		\text{on }[0,T]\times\mathbb R^2.
	\end{equation}
\end{theorem}

\begin{cor}
	\label{cor:viscosity-uniqueness}
	Under the hypotheses of Theorem~\ref{thm:comparison}, there exists at
	most one continuous viscosity solution
	$
	v\in\mathcal G_q
	$
	of
	\eqref{eq:nonlocal-HJB}-\eqref{eq:HJB-terminal}.
\end{cor}

\begin{proof}
	If $v_1$ and $v_2$ are two such solutions, apply
	Theorem~\ref{thm:comparison} first with
	$(\overline v,\underline v)=(v_1,v_2)$ and then with
	$(v_2,v_1)$. Thus $v_1\le v_2$ and $v_2\le v_1$.
\end{proof}

\begin{theorem}
	\label{thm:existence-identification}
	Suppose Assumptions~2.1-2.9 and~4.1 hold. Then the catastrophe
	insurance value function defined by \eqref{eq:insurance-value} is the
	unique viscosity solution in
	$C([0,T]\times\mathbb R^2)\cap\mathcal G_q$
	of
	\eqref{eq:nonlocal-HJB}-\eqref{eq:HJB-terminal}.
\end{theorem}

\begin{proof}
	Existence is probabilistic. Proposition~\ref{prop:value-continuity}
	gives
	$
	V\in C([0,T]\times\mathbb R^2)\cap\mathcal G_q,
	$
	and Theorem~\ref{thm:value-viscosity} shows that $V$ satisfies the nonlocal HJB equation in the viscosity sense together with the terminal condition. Uniqueness follows from Corollary~\ref{cor:viscosity-uniqueness}. Therefore every continuous viscosity solution in $\mathcal G_q$ coincides with the stochastic control value.
\end{proof}

\begin{prop}
	\label{prop:Hamiltonian-selector}
	Assume, in addition, that for every $(t,x,y)$ and every smooth test
	jet $(p,X)$ the map
	\[
	u\mapsto
	f(t,x,y,u)
	+b(x,y,u)p_1
	+\frac12\sigma_X^2(x,y,u)X_{11}
	+\mathcal I^u\phi(t,x,y)
	\]
	is lower semicontinuous. Then the minimizing set is nonempty and compact, and there exists a Borel measurable selector
	$
	\widehat u(t,x,y,p,X)
	$
	from the corresponding argmin correspondence.
\end{prop}

\begin{proof}
	By Lemma~\ref{lem:compact-control}, minimization may be restricted
	locally to a compact subset of $U$. Lower semicontinuity then gives
	attainment by the Weierstrass theorem. The objective is Borel
	measurable in the state and jet variables and lower semicontinuous in
	$u$; hence its argmin correspondence has nonempty compact values and
	measurable graph. The Kuratowski-Ryll-Nardzewski measurable-selection
	theorem provides a Borel selector.
\end{proof}

\begin{cor}[Classical verification as a special case]
	\label{cor:classical-verification}
	Suppose the unique viscosity solution $V$ belongs to
	$C^{1,2,2}$, the selector
	\[
	u^\star_s
	=
	\widehat u
	\bigl(
	s,X_{s-},Y_{s-},DV(s,X_{s-},Y_{s-}),
	D^2V(s,X_{s-},Y_{s-})
	\bigr)
	\]
	is admissible, and the stochastic integrals arising from It\^o's formula are true martingales. Then $u^\star$ is optimal and
	\[
	V(t,x,y)=J(t,x,y;u^\star).
	\]
\end{cor}

\begin{proof}
	For every $u\in U$, the HJB equation implies
	$
	V_t+\mathcal A^uV-\delta V+f(t,x,y,u)\ge0,
	$
	while equality holds for the minimizing selector $u^\star$. Apply It\^o's formula to
	$
	e^{-\delta(s-t)}V(s,X_s,Y_s)
	$
	between $t$ and $T$. After taking expectations, the martingale terms vanish and the terminal condition gives
	$
	V(t,x,y)\le J(t,x,y;u)
	$
	for every admissible $u$. Under $u^\star$ all inequalities become equalities, yielding
	$
	V(t,x,y)=J(t,x,y;u^\star).
	$
\end{proof}

\begin{cor}
	\label{cor:economic-hedging-rule}
	Suppose the hypotheses of
	Corollary~\ref{cor:classical-verification} hold and use the benchmark
	loss amplitude
	$
	\ell(y,z,u)=e^ya(z)h(u).
	$
	If $\sigma_X$ is independent of $u$ and the minimizer is interior,
	then the optimal physical hedge satisfies
	\[
	\begin{aligned}
		0
		={}&
		f_u(t,x,y,u^\star)
		+
		b_u(x,y,u^\star)V_x(t,x,y)-
		e^yh'(u^\star)
		\int_E
		a(z)
		\Big[
		V_x\bigl(t,x-e^ya(z)h(u^\star),y\bigr)
		-
		V_x(t,x,y)
		\Big]\nu(dz).
	\end{aligned}
	\]
	Thus the marginal cost of securing reconstruction capacity is equated
	to the marginal change in continuation value generated by shifting
	the entire catastrophe-loss distribution.
\end{cor}

\begin{proof}
	Differentiate the minimizing Hamiltonian with respect to $u$. The
	local terms give $f_u+b_uV_x$. Differentiating the nonlocal term under
	the integral sign gives
	\[
	-e^yh'(u)
	\int_E
	a(z)
	\left[
	V_x(t,x-e^ya(z)h(u),y)-V_x(t,x,y)
	\right]\nu(dz).
	\]
	The interchange of differentiation and integration follows from the
	assumed classical regularity, the moment condition on the L\'evy
	measure, and dominated convergence. At an interior minimizer the
	derivative vanishes.
\end{proof}

\begin{rmk}
	The preceding results separate the two roles played by the
	catastrophe model. Probabilistically, the marked jump measure creates
	the nonlocal generator and prevents the control problem from reducing
	to a diffusion HJB equation. Economically, the multiplicative factor
	$e^ya(z)h(u)$ makes the value of physical hedging state dependent:
	the same unit of secured reconstruction capacity has a different
	marginal value when replacement costs are elevated. The viscosity
	formulation preserves this mechanism without imposing differentiability
	of the insurer's continuation value.
\end{rmk}

\section{Mean Field Game of Catastrophe Insurers}
\label{sec:mfg}

	The representative-insurer dynamics remain exactly those of
	Sections~\ref{sec:prob-foundations}-\ref{sec:insurance-control}. Mean field interaction is introduced through the objective functional rather than through an endogenous reconstruction-capacity market. Accordingly, the couplings $F$ and $G$ should be interpreted as reduced-form strategic cost externalities, the private cost of a given
	capitalization and replacement-cost state depends on the
	cross-sectional vulnerability of the insurance sector. The population distribution does not determine the replacement-cost process $Y$, the catastrophe compensator $\nu(dz)\,dt$, the severity law $a(z)$, or an
	aggregate supply of reconstruction capacity. Thus the mean field game studies \citep{pramanik2024estimation} how industry-wide vulnerability feeds back into individual hedging incentives, but it does not model market clearing in reconstruction services or endogenous price formation. This separation
	preserves the controlled state process analyzed in
	Sections~\ref{sec:prob-foundations}-\ref{sec:dpp-viscosity} while isolating a strategic externality that can be studied within the nonlocal control.

\subsection{Mean Field Equilibrium}
\label{subsec:mfg-equilibrium}

For a deterministic flow $m\in\mathcal M_T$, define
\begin{equation}
	\label{eq:mfg-frozen-cost}
	\begin{aligned}
		J^m(t,x,y;u)
		:=
		\mathbb E_{t,x,y}\Bigg[
		&\int_t^T e^{-\delta(s-t)}
		\Big(
		f(s,X_s,Y_s,u_s)
		+
		F(s,(X_s,Y_s),m_s)
		\Big)\,ds\\
		&+
		e^{-\delta(T-t)}
		\Big(
		g(X_T,Y_T)
		+
		G((X_T,Y_T),m_T)
		\Big)
		\Bigg],
	\end{aligned}
\end{equation}
and
\begin{equation}
	\label{eq:mfg-frozen-value}
	V^m(t,x,y)
	:=
	\inf_{u\in\mathcal U_t}J^m(t,x,y;u).
\end{equation}
The controlled state in \eqref{eq:mfg-frozen-cost} is exactly the
solution of \eqref{eq:Y} and \eqref{eq:X-general}.

\medskip

\noindent\textbf{Assumption 5.1 (Mean field couplings).}
There exists $C_F>0$ such that:
\begin{itemize}
	\item[(i)] $F$ and $G$ are continuous;
	\item[(ii)] for all $(t,z,m)$,
	\[
	|F(t,z,m)|+|G(z,m)|
	\le
	C_F\Bigl(
	1+|x|^q+e^{q|y|}+M_q(m)^q
	\Bigr),
	\]
	where
	$
	M_q(m)
	:=
	\left(
	\int_{\mathbb R^2}|z'|^q\,m(dz')
	\right)^{1/q};
	$
	\item[(iii)] for every $R>0$, there exists $L_R>0$ such that
	\[
	|F(t,z,m)-F(s,z',n)|
	+
	|G(z,m)-G(z',n)|
	\le
	L_R\Bigl(
	|t-s|^{1/2}
	+
	|z-z'|
	+
	W_q(m,n)
	\Bigr)
	\]
	whenever
	$|z|,|z'|,M_q(m),M_q(n)\le R$;
	\item[(iv)] $F$ and $G$ are bounded from below by a function of at most
	the exponential-polynomial growth appearing in Assumption~2.6,
	uniformly on bounded-moment subsets of $\mathcal P_q(\mathbb R^2)$.
\end{itemize}

	\begin{remark}
		\label{rmk:reduced-form-mfg}
		The maps $F(t,z,m)$ and $G(z,m)$ are not derived from a
		reconstruction-market clearing condition. They represent reduced-form
		costs associated with cross-sectional concentration of vulnerability.
		For example, $F$ may penalize an insurer occupying a high
		replacement-cost or weak-capitalization state when a large fraction of
		the industry simultaneously occupies comparable states \citep{pramanik2023path1}. Such a
		coupling can represent congestion, coordination, regulatory, or
		systemic-capacity costs in reduced form, but it does not determine the
		price or quantity of reconstruction services.
		
		A structural reconstruction-capacity model would require an additional
		aggregate quantity, say $D_t(m,\alpha)$, measuring rebuilding or
		capacity demand, together with a supply or market-clearing relation
		that determines an endogenous price or scarcity factor. In such a
		model, one would expect a relation of the schematic form
		$P_t^{\mathrm{rec}}=\mathcal P(D_t,S_t)$, with
		$\partial_D\mathcal P>0$, and the resulting price would enter either
		the replacement-cost dynamics or the controlled loss amplitude.
		Because this would make the state dynamics themselves
		population-dependent, it would produce a different mean field game
		from the one analyzed here.
	\end{remark}

\begin{lem}
	\label{lem:mfg-frozen-inheritance}
	Fix $m\in\mathcal M_T$ with
	$\sup_{t\le T}M_q(m_t)<\infty$.
	Under Assumptions~2.1-2.9, 4.1, and 5.1, the frozen-flow problem expressed in Equation \eqref{eq:mfg-frozen-value} satisfies the hypotheses of
	Theorem~\ref{thm:DPP},
	Proposition~\ref{prop:value-continuity}, and
	Theorem~\ref{thm:existence-identification} after replacing $f$ and
	$g$ by
	$
	f^m(t,x,y,u)
	:=
	f(t,x,y,u)+F(t,(x,y),m_t)
	$
	and
	$
	g^m(x,y)
	:=
	g(x,y)+G((x,y),m_T).
	$
	The local constants may be chosen uniformly over any family of flows
	with uniformly bounded $q$th moments.
\end{lem}

\begin{proof}
	The state equation does not depend on $m$, so the estimates from
	Section~\ref{sec:prob-foundations} remain unchanged. In particular,
	Theorem~\ref{thm:strong-wellposedness},
	Proposition~\ref{prop:initial-stability},
	Lemma~\ref{lem:control-stability}, and Proposition~\ref{prop:UI}
	continue to hold.
	
	If $\sup_tM_q(m_t)\le R$, Assumption~5.1(ii) yields
	\[
	|f^m(t,x,y,u)|+|g^m(x,y)|
	\le
	C_R\bigl(
	1+|x|^q+e^{q|y|}+|u|^q
	\bigr).
	\]
	Since the coupling does not depend on $u$, the coercivity in the physical-hedging variable is exactly that of $f$. Assumption~5.1(iii) gives the local continuity required by
	Proposition~\ref{prop:value-continuity} and Assumption~4.1(i). The jump operator and comparison structure are unchanged. Hence, the results cited in the statement apply to the frozen problem.
\end{proof}

\begin{prop}[Frozen-flow viscosity equation]
	\label{prop:mfg-frozen-hjb}
	Under the hypotheses of Lemma~\ref{lem:mfg-frozen-inheritance},
	$V^m$ is the unique viscosity solution in
	$C([0,T]\times\mathbb R^2)\cap\mathcal G_q$ of
	\begin{equation}
		\label{eq:mfg-hjb-frozen}
		-\partial_t v^m
		-
		\mathcal H\bigl(
		t,x,y,v^m,Dv^m,D^2v^m;
		v^m(t,\cdot,\cdot)
		\bigr)
		-
		F(t,(x,y),m_t)
		=
		0,
	\end{equation}
	with terminal condition
	\begin{equation}
		\label{eq:mfg-terminal-frozen}
		v^m(T,x,y)
		=
		g(x,y)+G((x,y),m_T).
	\end{equation}
\end{prop}

\begin{proof}
	Lemma~\ref{lem:mfg-frozen-inheritance} allows the frozen problem to be
	treated by Section~\ref{sec:dpp-viscosity}. The stopping-time DPP
	follows from Theorem~\ref{thm:DPP}, the viscosity solution from
	Theorem~\ref{thm:value-viscosity}, and uniqueness from
	Theorem~\ref{thm:comparison} and
	Corollary~\ref{cor:viscosity-uniqueness}. Finally,
	Theorem~\ref{thm:existence-identification} identifies this viscosity
	solution with $V^m$.
\end{proof}

Let $\alpha:[0,T]\times\mathbb R^2\to U$ be an admissible Borel
feedback. For $\phi\in C_c^2(\mathbb R^2)$ define
\[
\begin{aligned}
	\mathscr L^{\alpha_t}\phi(x,y)
	:={}&
	b(x,y,\alpha_t(x,y))\phi_x
	+
	\kappa(\vartheta-y)\phi_y
	+
	\frac12\sigma_X^2(x,y,\alpha_t(x,y))\phi_{xx}\\
	&+
	\frac12\sigma_I^2\phi_{yy}
	+
	\int_E
	\Big[
	\phi\bigl(
	x-\ell(y,z,\alpha_t(x,y)),y
	\bigr)
	-\phi(x,y)\\
	&\hspace{4.3cm}
	+\ell(y,z,\alpha_t(x,y))\phi_x(x,y)
	\Big]\nu(dz).
\end{aligned}
\]

\begin{defn}
		\label{def:mfg-equilibrium}
		Let $\mathfrak A(m)$ denote the full set of admissible relaxed
		canonical laws for the frozen population flow $m\in\mathcal M_T$.
		A pair $(m^\star,P^\star)$ is called a relaxed mean field equilibrium
		if $m^\star\in\mathcal M_T$, $P^\star\in\mathfrak A(m^\star)$,
		$P^\star$ is optimal for the frozen-flow problem associated with
		$m^\star$, and
		\begin{equation}
			\label{eq:mfg-consistency}
			m_t^\star=P^\star\circ Z_t^{-1},
			\qquad
			Z_t:=(X_t,Y_t),
			\qquad
			0\le t\le T.
		\end{equation}
		A relaxed equilibrium is called Markovian if, under $P^\star$, the
		relaxed control coordinate admits the representation
		$q_t(da)=\widehat q^\star(t,Z_{t-},da)$
		$dt\otimes P^\star$-a.e.\ for a Borel kernel
		$\widehat q^\star:[0,T]\times\mathbb R^2\to\mathcal P(U)$.
		It is called strict Markovian if
		$\widehat q^\star(t,z,da)=\delta_{\alpha^\star(t,z)}(da)$
		for a Borel feedback $\alpha^\star:[0,T]\times\mathbb R^2\to U$.
	\end{defn}
 For $P\in\mathfrak A(m)$ define the state-control occupation measure
	$\eta^P$ on $[0,T]\times\mathbb R^2\times U$ by
	$\eta^P(dt,dz,da)
	:=
	dt\,E^P[\mathbf 1_{\{Z_t\in dz\}}q_t(da)]$.
	If $m_t^P:=P\circ Z_t^{-1}$, disintegration on the standard Borel
	space $[0,T]\times\mathbb R^2\times U$ yields a Borel stochastic
	kernel $\widehat q^P(t,z,da)$ satisfying
	$\eta^P(dt,dz,da)
	=
	dt\,m_t^P(dz)\widehat q^P(t,z,da)$.
	The kernel $\widehat q^P$ is the state-conditioned barycentric
	relaxation of the possibly history-dependent control law $P$.

\begin{prop}
		\label{prop:mfg-kolmogorov}
		Let $(m^\star,P^\star)$ be a relaxed equilibrium and let
		$\widehat q^\star:=\widehat q^{P^\star}$ be the disintegration kernel
		defined above. Then, for every $\phi\in C_c^2(\mathbb R^2)$,
		\begin{equation}
			\label{eq:mfg-weak-kolmogorov}
			\int_{\mathbb R^2}\phi(z)m_t^\star(dz)
			-
			\int_{\mathbb R^2}\phi(z)m_0(dz)
			=
			\int_0^t\int_{\mathbb R^2}\int_U
			\mathcal A^a\phi(z)\,
			\widehat q^\star(s,z,da)\,
			m_s^\star(dz)\,ds.
		\end{equation}
		Equivalently,
		$\partial_t m_t^\star
		=
		(\mathcal A^{\widehat q_t^\star})^*m_t^\star$
		in $\mathcal D'((0,T)\times\mathbb R^2)$, where
		$\mathcal A^{\widehat q_t^\star}\phi(z)
		:=
		\int_U\mathcal A^a\phi(z)\widehat q^\star(t,z,da)$.
	\end{prop}

\begin{proof}
		For $\phi\in C_c^2(\mathbb R^2)$, admissibility of $P^\star$ gives
		$M_t^\phi
		:=
		\phi(Z_t)-\phi(Z_0)
		-
		\int_0^t\int_U
		\mathcal A^a\phi(Z_{s-})q_s(da)\,ds$
		as a local $P^\star$-martingale. Proposition~\ref{prop:jump-integral}
		and the moment estimates of Theorem~\ref{thm:strong-wellposedness}
		make the stopped family uniformly integrable, hence
		$E^{P^\star}[M_t^\phi]=0$. Therefore
		$\int\phi\,dm_t^\star-\int\phi\,dm_0
		=
		\int_0^t
		E^{P^\star}[\int_U\mathcal A^a\phi(Z_{s-})q_s(da)]\,ds$.
		By the defining disintegration of $\eta^{P^\star}$,
		$E^{P^\star}[\int_U\mathcal A^a\phi(Z_s)q_s(da)]
		=
		\int_{\mathbb R^2}\int_U
		\mathcal A^a\phi(z)
		\widehat q^\star(s,z,da)m_s^\star(dz)$
		for $ds$-a.e.\ $s$. Substitution gives
		\eqref{eq:mfg-weak-kolmogorov}.
	\end{proof}

 \begin{cor}[Relaxed nonlocal HJB-Kolmogorov system]
		\label{cor:mfg-system}
		Let $(m^\star,P^\star)$ be a relaxed equilibrium and set
		$v^\star:=V^{m^\star}$. Then $(v^\star,m^\star,\widehat q^\star)$
		satisfies
		\begin{equation}
			\label{eq:mfg-coupled-system}
			\left\{
			\begin{aligned}
				-\partial_t v^\star
				&-
				\mathcal H\bigl(
				t,x,y,v^\star,Dv^\star,D^2v^\star;
				v^\star(t,\cdot,\cdot)
				\bigr)
				-
				F(t,(x,y),m_t^\star)
				=0,\\
				\partial_t m_t^\star
				&=
				(\mathcal A^{\widehat q_t^\star})^*m_t^\star,\\
				v^\star(T,x,y)
				&=
				g(x,y)+G((x,y),m_T^\star),\\
				m_0^\star&=m_0.
			\end{aligned}
			\right.
		\end{equation}
		The backward equation is understood in the viscosity sense and the
		forward equation in the distributional sense. If the equilibrium is
		strict Markovian, then
		$\widehat q^\star(t,z,da)=\delta_{\alpha^\star(t,z)}(da)$ and the
		second equation reduces to
		$\partial_t m_t^\star
		=
		(\mathscr L^{\alpha_t^\star})^*m_t^\star$.
	\end{cor}

\begin{proof}
		Proposition~\ref{prop:mfg-frozen-hjb}, applied with
		$m=m^\star$, gives the backward equation and terminal condition.
		Proposition~\ref{prop:mfg-kolmogorov} gives the forward equation.
		The strict Markov reduction follows from
		$\mathcal A^{\delta_{\alpha^\star(t,z)}}=
		\mathscr L^{\alpha_t^\star}$.
	\end{proof}

	\begin{remark}
		The equilibrium feedback in this section is strategic but not a
		market-clearing feedback. A change in the hedging rule changes the jump
		amplitude through $h(u)$ in
		\eqref{eq:insurance-loss-amplitude}, thereby changing the population
		law $m_t$; the changed law alters the reduced-form coupling $F$ and
		hence the subsequent best response. In contrast, $m_t$ does not feed
		back into the catastrophe intensity, the primitive severity law, or
		the replacement-cost dynamics. The equilibrium therefore captures an
		endogenous distributional externality among insurers, not an
		endogenous reconstruction-price or reconstruction-capacity mechanism \citep{pramanik2026stochastic}.
	\end{remark}

\subsection{Existence and Uniqueness of Equilibrium}
\label{subsec:mfg-existence-uniqueness}

Existence is established in the relaxed formulation and transferred to
strict feedback controls whenever the minimizing control can be
selected uniquely.

\medskip

\noindent\textbf{Assumption 5.2 (Compact relaxed best responses).}
	There exist $p\in(q,\bar q]$, $R_0<\infty$, and a compact convex set
	$U_0\subset U$ such that:
	\begin{itemize}
		\item[(i)] $m_0\in\mathcal P_p(\mathbb R^2)$;
		\item[(ii)] for every $m\in\mathcal M_T$ with
		$\sup_{t\le T}M_p(m_t)\le R_0$, the frozen problem admits an optimal
		relaxed canonical law $P\in\mathfrak A(m)$ whose control coordinate is
		supported on $U_0$ $dt\otimes P$-a.e.;
		\item[(iii)] if $m^n\to m$ in $\mathbf W_q$,
		$P^n\in\mathfrak A(m^n)$ are optimal, and $P^n\Rightarrow P$ in the
		canonical topology, then $P\in\mathfrak A(m)$ and $P$ is optimal for
		the frozen problem associated with $m$;
		\item[(iv)] the coefficients restricted to $U_0$ satisfy
		Assumptions~2.1-2.9 and~4.1 uniformly in $a\in U_0$.
	\end{itemize}

\begin{rmk}
		No Markovianity is imposed in Assumption~5.2. In particular, the
		compactness and fixed-point argument below is carried out over the full
		convex set of relaxed canonical laws \citep{pramanik2023optimal}. Markovianity is addressed only
		after existence of the relaxed equilibrium has been established.
	\end{rmk}

	For $R>0$ and $\gamma\in(0,1/2]$, retain
	$\mathcal K_{R,\gamma}
	:=
	\{m\in\mathcal M_T:
	m_0=m_0,\ 
	\sup_{t\le T}M_p(m_t)\le R,\ 
	W_q(m_t,m_s)\le R|t-s|^\gamma\}$.
	Lemma~\ref{lem:mfg-flow-compactness} remains valid with
	``relaxed control'' interpreted as an arbitrary relaxed canonical law
	supported on $U_0$. For $m\in\mathcal K_{R,\gamma}$ define the full relaxed best-response
	set
	$\mathfrak{BR}(m)
	:=
	\{P\in\mathfrak A(m):
	J^m(P)=\inf_{Q\in\mathfrak A(m)}J^m(Q)\}$,
	and define
	$\Phi(m)
	:=
	\{(P\circ Z_t^{-1})_{0\le t\le T}:
	P\in\mathfrak{BR}(m)\}$.
	Thus $\mathfrak{BR}(m)$ contains all optimal relaxed laws, not only
	those admitting a Markov representation.

\begin{lem}
	\label{lem:mfg-flow-compactness}
	Under Assumptions~2.1-2.5 and 5.2, there exist $R<\infty$ and $\gamma>0$ such that every state-law flow induced by a relaxed control supported on $U_0$ belongs to $\mathcal K_{R,\gamma}$. Moreover,
	$\mathcal K_{R,\gamma}$ is compact in
	$C([0,T];\mathcal P_q(\mathbb R^2))$ endowed with $\mathbf W_q$.
\end{lem}

\begin{proof}
	Uniform support of the controls in $U_0$ and
	Theorem~\ref{thm:strong-wellposedness} yield
	\[
	\sup_u
	\mathbb E\left[
	\sup_{0\le t\le T}|X_t|^p
	\right]
	+
	\mathbb E\left[
	\sup_{0\le t\le T}|Y_t|^p
	\right]
	\le
	C\bigl(1+M_p(m_0)^p\bigr).
	\]
	Hence all induced state laws have uniformly bounded $p$th moments \citep{pramanik2026optimal}. For $0\le s<t\le T$, the state equations, the BDG inequality, and the
	Bichteler-Jacod inequality yields
	$
	\mathbb E|Y_t-Y_s|^q
	\le
	C|t-s|^{q/2},
	$ and $
	\mathbb E|X_t-X_s|^q
	\le
	C\bigl(
	|t-s|^{q/2}+|t-s|
	\bigr).
	$
	Therefore, for some $\gamma>0$, we have
	$
	W_q\bigl(
	\mathcal L(X_t,Y_t),
	\mathcal L(X_s,Y_s)
	\bigr)
	\le
	C|t-s|^\gamma.
	$
	Choosing $R$ sufficiently large gives membership in
	$\mathcal K_{R,\gamma}$. Since $p>q$, the uniform $p$th-moment estimate implies uniform $q$-integrability and relative compactness of each time marginal in
	$W_q$. The common modulus of continuity gives equicontinuity. The metric Arzel\`a-Ascoli theorem then yields relative compactness in
	$C([0,T];\mathcal P_q)$. The defining moment and modulus constraints
	are closed, so $\mathcal K_{R,\gamma}$ is compact.
\end{proof}

\begin{lem}
		\label{lem:mfg-closed-graph}
		Suppose Assumptions~2.1-2.9, 4.1, 5.1, and 5.2 hold.
		If $m^n\to m$ in $\mathbf W_q$, $\mu^n\in\Phi(m^n)$, and
		$\mu^n\to\mu$ in $\mathbf W_q$, then $\mu\in\Phi(m)$.
	\end{lem}

\begin{proof}
		Choose $P^n\in\mathfrak{BR}(m^n)$ with
		$\mu_t^n=P^n\circ Z_t^{-1}$. Compactness of $U_0$,
		Lemma~\ref{lem:mfg-flow-compactness}, the moment bounds of
		Theorem~\ref{thm:strong-wellposedness}, and tightness of the canonical
		relaxed-control coordinate imply tightness of $(P^n)$ on
		$\Omega^\circ$. Passing to a subsequence, $P^n\Rightarrow P$.
		Assumption~5.2(iii) gives
		$P\in\mathfrak A(m)$ and optimality of $P$ for the frozen flow $m$.
		For every $t$ at which the canonical evaluation map is
		$P$-a.s.\ continuous,
		$P^n\circ Z_t^{-1}\Rightarrow P\circ Z_t^{-1}$.
		The uniform $p$-moment bound with $p>q$ upgrades this convergence to
		$W_q$. Since $\mu^n\to\mu$ in $\mathbf W_q$,
		$\mu_t=P\circ Z_t^{-1}$ on a dense set of times. Both
		$t\mapsto\mu_t$ and $t\mapsto P\circ Z_t^{-1}$ are $W_q$-continuous
		by Lemma~\ref{lem:mfg-flow-compactness}; hence the equality holds for
		all $t\in[0,T]$. Therefore $P\in\mathfrak{BR}(m)$ and
		$\mu\in\Phi(m)$.
	\end{proof}

 \begin{lem}
		\label{lem:mfg-law-convexity}
		For every $m\in\mathcal K_{R,\gamma}$,
		$\mathfrak{BR}(m)$ is convex. Consequently, $\Phi(m)$ is convex.
	\end{lem}

\begin{prop}
		\label{prop:mfg-uhe}
		Under the hypotheses of Lemma~\ref{lem:mfg-closed-graph},
		$\Phi:\mathcal K_{R,\gamma}\rightrightarrows
		\mathcal K_{R,\gamma}$ has nonempty compact convex values and is upper
		hemicontinuous.
	\end{prop}

\begin{proof}
		Nonemptiness follows from Assumption~5.2(ii). Convexity follows from
		Lemma~\ref{lem:mfg-law-convexity}. Let $m\in\mathcal K_{R,\gamma}$ and
		$(\mu^n)\subset\Phi(m)$. Compactness of
		$\mathcal K_{R,\gamma}$ gives a subsequence
		$\mu^{n_k}\to\mu$ in $\mathbf W_q$. Since $m^{n_k}\equiv m$,
		Lemma~\ref{lem:mfg-closed-graph} yields $\mu\in\Phi(m)$; hence
		$\Phi(m)$ is closed in the compact set $\mathcal K_{R,\gamma}$ and is
		therefore compact. Lemma~\ref{lem:mfg-closed-graph} also gives
		$\operatorname{Gr}(\Phi)$ closed in
		$\mathcal K_{R,\gamma}\times\mathcal K_{R,\gamma}$. A compact-valued
		correspondence with closed graph from a compact metric space into a
		Hausdorff space is upper hemicontinuous. Thus $\Phi$ has all the
		properties claimed.
	\end{proof}

 \begin{theorem}[Existence of a relaxed mean field equilibrium]
		\label{thm:mfg-existence}
		Suppose Assumptions~2.1-2.9, 4.1, 5.1, and 5.2 hold. Then there exist
		$m^\star\in\mathcal K_{R,\gamma}$ and
		$P^\star\in\mathfrak{BR}(m^\star)$ such that
		$m_t^\star=P^\star\circ Z_t^{-1}$ for every $t\in[0,T]$.
		Equivalently, $(m^\star,P^\star)$ is a relaxed mean field equilibrium
		in the sense of Definition~\ref{def:mfg-equilibrium}.
	\end{theorem}

\begin{proof}
		Lemma~\ref{lem:mfg-flow-compactness} gives a nonempty compact convex
		set $\mathcal K_{R,\gamma}$. Proposition~\ref{prop:mfg-uhe} gives a
		correspondence
		$\Phi:\mathcal K_{R,\gamma}\rightrightarrows\mathcal K_{R,\gamma}$
		with nonempty compact convex values and closed graph. The
		Kakutani-Fa-Glicksberg theorem yields
		$m^\star\in\mathcal K_{R,\gamma}$ such that
		$m^\star\in\Phi(m^\star)$. By the definition of $\Phi$, there exists
		$P^\star\in\mathfrak{BR}(m^\star)$ with
		$m_t^\star=P^\star\circ Z_t^{-1}$ for all $t$. The first statement is
		optimality; the second is the consistency condition
		\eqref{eq:mfg-consistency}.
	\end{proof}

\noindent\textbf{Assumption 5.3 (Markov realization).}
For every admissible relaxed law $P$ supported on $U_0$, let
$\widehat q^P$ be the disintegration kernel determined by
$\eta^P(dt,dz,da)
=
dt\,m_t^P(dz)\widehat q^P(t,z,da)$.
Assume that the controlled martingale problem with the Markov relaxed
kernel $\widehat q^P$ admits a solution $\widehat P$ satisfying
$\widehat P\circ Z_t^{-1}=P\circ Z_t^{-1}$ for every $t\in[0,T]$ and
$J^m(\widehat P)=J^m(P)$ whenever $P$ is evaluated in the frozen
environment $m$.

\begin{rmk}
		Assumption~5.3 is a Markovianization hypothesis, not a convexity
		hypothesis. It is separated from the fixed-point argument because
		convexity of $\Phi$ is already obtained at the full relaxed-law level
		in Lemma~\ref{lem:mfg-law-convexity}. Markovian realization theorems
		for controlled martingale problems are developed in
		\citet{KurtzStockbridge1998}. In applications, Assumption~5.3 may be
		verified by identifying the state-conditioned kernel
		$\widehat q^P(t,z,\cdot)$ and proving well-posedness of the associated
		Markov martingale problem.
	\end{rmk}

\begin{prop}
		\label{prop:mfg-markovianization}
		Suppose the hypotheses of Theorem~\ref{thm:mfg-existence} and
		Assumption~5.3 hold. Then every relaxed equilibrium
		$(m^\star,P^\star)$ admits an optimal relaxed Markov realization
		$(m^\star,\widehat P^\star)$ with
		$\widehat P^\star\circ Z_t^{-1}=m_t^\star$ for all $t$ and
		$q_t(da)=\widehat q^\star(t,Z_{t-},da)$
		$dt\otimes\widehat P^\star$-a.e.
	\end{prop}

 \begin{proof}
		Let $(m^\star,P^\star)$ be furnished by
		Theorem~\ref{thm:mfg-existence}. Disintegrate its occupation measure
		as
		$\eta^{P^\star}(dt,dz,da)
		=
		dt\,m_t^\star(dz)\widehat q^\star(t,z,da)$.
		Assumption~5.3 provides a Markov relaxed law $\widehat P^\star$
		generated by $\widehat q^\star$ such that
		$\widehat P^\star\circ Z_t^{-1}=m_t^\star$ for every $t$ and
		$J^{m^\star}(\widehat P^\star)=J^{m^\star}(P^\star)$.
		Since $P^\star\in\mathfrak{BR}(m^\star)$,
		$J^{m^\star}(P^\star)=V^{m^\star,\mathrm{rel}}$; therefore
		$J^{m^\star}(\widehat P^\star)=V^{m^\star,\mathrm{rel}}$ and
		$\widehat P^\star\in\mathfrak{BR}(m^\star)$. The marginal identity
		preserves \eqref{eq:mfg-consistency}, so
		$(m^\star,\widehat P^\star)$ is a relaxed Markov equilibrium.
	\end{proof}

\begin{cor}
		\label{cor:mfg-strict-existence}
		Suppose the hypotheses of Proposition~\ref{prop:mfg-markovianization}
		hold. Assume additionally that the optimal Markov relaxed kernel is
		$dt\otimes m_t^\star$-a.e.\ Dirac, i.e.,
		$\widehat q^\star(t,z,da)=\delta_{\alpha^\star(t,z)}(da)$ for a Borel
		selector $\alpha^\star:[0,T]\times\mathbb R^2\to U_0$. Then
		$(m^\star,\alpha^\star)$ is a strict Markov equilibrium. In
		particular, the conclusion applies when the frozen Hamiltonian has a
		unique pointwise minimizer and the hypotheses of
		Proposition~\ref{prop:Hamiltonian-selector} hold \citep{pramanik2026quantum}.
	\end{cor}

\begin{proof}
		By Proposition~\ref{prop:mfg-markovianization},
		$\widehat P^\star$ is optimal for the frozen flow $m^\star$ and has
		state marginals $m^\star$. The Dirac representation gives the strict
		feedback control
		$u_t^\star=\alpha^\star(t,Z_{t-})$. Hence the optimality and
		consistency conditions in Definition~\ref{def:mfg-equilibrium} hold
		with a strict Markov control.
	\end{proof}

For uniqueness, we impose the monotonicity structure of
\citet{LasryLions2007}.

\medskip

\noindent\textbf{Assumption 5.4 (Lasry-Lions monotonicity).}
For every $t\in[0,T]$ and
$m,n\in\mathcal P_q(\mathbb R^2)$,
\begin{equation}
	\label{eq:mfg-monotonicity}
	\int_{\mathbb R^2}
	\bigl(
	F(t,z,m)-F(t,z,n)
	\bigr)
	(m-n)(dz)
	\ge0,
\end{equation}
and
\begin{equation}
	\label{eq:mfg-terminal-monotonicity}
	\int_{\mathbb R^2}
	\bigl(
	G(z,m)-G(z,n)
	\bigr)
	(m-n)(dz)
	\ge0.
\end{equation}
Equality in both monotonicity relations for an entire equilibrium flow
implies equality of the two flows.

\medskip

\noindent\textbf{Assumption 5.5.}
For every $m\in\mathcal K_{R,\gamma}$, the frozen problem admits a
unique optimal state-control law. A sufficient condition is strict
convexity of the frozen Hamiltonian in the control variable on the
effective compact set of Lemma~\ref{lem:compact-control}, together
with uniqueness in law of the controlled state equation.

\begin{lem}
	\label{lem:mfg-cross-inequality}
	Let $(m^1,P^1)$ and $(m^2,P^2)$ be two relaxed equilibria. Then
	\begin{equation}
		\label{eq:mfg-cross-ineq}
		\begin{aligned}
			&
			\int_0^T e^{-\delta t}
			\int_{\mathbb R^2}
			\bigl(
			F(t,z,m_t^1)-F(t,z,m_t^2)
			\bigr)
			(m_t^1-m_t^2)(dz)\,dt\\
			&+
			e^{-\delta T}
			\int_{\mathbb R^2}
			\bigl(
			G(z,m_T^1)-G(z,m_T^2)
			\bigr)
			(m_T^1-m_T^2)(dz)\leq 0.
		\end{aligned}
	\end{equation}
\end{lem}

\begin{proof}
	Optimality yields
	$
	J^{m^1}(P^1)\le J^{m^1}(P^2),
	$ and $
	J^{m^2}(P^2)\le J^{m^2}(P^1).
	$
	The cross-comparisons are legitimate because the controlled dynamics
	do not depend on the mean field. Adding the inequalities cancels the
	private running cost $f$, terminal cost $g$, and all control-dependent
	state costs. Using the equilibrium consistency relation and collecting
	the remaining coupling terms yields
	\eqref{eq:mfg-cross-ineq}.
\end{proof}

\begin{theorem}
	\label{thm:mfg-uniqueness}
	Suppose Assumptions~2.1-2.9, 4.1, and 5.1-5.5 hold.
	Then the mean field equilibrium is unique in law. If the frozen
	optimal control is generated by a unique Markov feedback, then the
	equilibrium feedback is unique
	$dt\otimes m_t^\star$-a.e.
\end{theorem}

\begin{proof}
	Let $(m^1,P^1)$ and $(m^2,P^2)$ be two equilibria. By
	Lemma~\ref{lem:mfg-cross-inequality}, the sum of the two monotonicity
	expressions is nonpositive. Assumption~5.3 makes each expression
	nonnegative. Hence, both must vanish. The strictness clause in Assumption~5.3 implies
	$
	m_t^1=m_t^2
$ for every $t\in[0,T]$,
	where continuity upgrades almost-everywhere equality in time to pointwise equality. Let the common flow be $m^\star$. Both $P^1$ and $P^2$ are optimal
	for the same frozen problem \citep{pramanik2022lock}. Assumption~5.4 therefore gives
	$P^1=P^2$. If this law is induced by a unique Markov feedback, the
	feedbacks agree $dt\otimes m_t^\star$-a.e.
\end{proof}

\begin{cor}[Uniqueness of the coupled HJB--Kolmogorov solution]
	\label{cor:mfg-system-uniqueness}
	Under the hypotheses of Theorem~\ref{thm:mfg-uniqueness}, the
	equilibrium pair
	\[
	(v^\star,m^\star)
	\in
	\bigl(
	C([0,T]\times\mathbb R^2)\cap\mathcal G_q
	\bigr)
	\times
	\mathcal K_{R,\gamma}
	\]
	solving \eqref{eq:mfg-coupled-system} is unique within the class
	generated by admissible equilibrium controls.
\end{cor}

\begin{proof}
	Theorem~\ref{thm:mfg-uniqueness} gives uniqueness of $m^\star$.
	For this fixed flow, Proposition~\ref{prop:mfg-frozen-hjb} and
	Corollary~\ref{cor:viscosity-uniqueness} give uniqueness of the
	backward value function. Proposition~\ref{prop:mfg-kolmogorov}
	identifies the forward component with the unique equilibrium law.
\end{proof}

\begin{prop}
	\label{prop:mfg-monotone-example}
	Let $\psi_D,\psi_S:\mathbb R^2\to\mathbb R$ be continuous functions
	of at most $q$th-order growth and define
	$
	M_D(m):=\int\psi_D(z)\,m(dz),
$ and $
	M_S(m):=\int\psi_S(z)\,m(dz).$
	For $\gamma_D,\gamma_S\ge0$, set
	\begin{equation}
		\label{eq:mfg-economic-coupling}
		F(t,z,m)
		=
		\gamma_D\psi_D(z)M_D(m)
		+
		\gamma_S\psi_S(z)M_S(m).
	\end{equation}
	Then
	\[
	\begin{aligned}
		&\int
		\bigl(
		F(t,z,m)-F(t,z,n)
		\bigr)
		(m-n)(dz)=
		\gamma_D
		\bigl(
		M_D(m)-M_D(n)
		\bigr)^2
		+
		\gamma_S
		\bigl(
		M_S(m)-M_S(n)
		\bigr)^2
		\ge0.
	\end{aligned}
	\]
	Hence \eqref{eq:mfg-economic-coupling} satisfies
	Assumption~5.3.
\end{prop}

\begin{proof}
	By linearity,
	\[
	\int\psi_D(z)(m-n)(dz)
	=
	M_D(m)-M_D(n),
	\]
	and the same identity holds for $\psi_S$. Substitution into
	\eqref{eq:mfg-economic-coupling} gives the stated sum of squares.
\end{proof}

\begin{cor}
	\label{cor:mfg-economic-coupling}
	Let $\psi_D(x,y)$ be an increasing transformation of replacement-cost
		exposure and let $\psi_S(x,y)$ be an increasing penalty for weak
		capitalization. Then the coupling \eqref{eq:mfg-economic-coupling} assigns a larger private cost to
	states whose exposure is aligned with a larger industry-wide
	concentration of the same exposure, while preserving the monotonicity required for Theorem~\ref{thm:mfg-uniqueness}.
\end{cor}

\begin{proof}
	The conclusion follows from
	Proposition~\ref{prop:mfg-monotone-example}. The chosen functions alter only the economic interpretation of $M_D$ and $M_S$, not the sum-of-squares monotonicity calculation.
\end{proof}

\begin{rmk}
	For every frozen population flow,
	Theorem~\ref{thm:existence-identification} supplies the unique viscosity value function; the population law is propagated by the same jump-diffusion as in \eqref{eq:X-general}; and equilibrium
	requires that the frozen flow coincide with the law generated by its own optimal response \citep{anderson2026obesity}. No state-dependent catastrophe intensity,
	common noise, or population-dependent jump kernel is added in this section.
\end{rmk}

\section{Economic Implications}
\label{sec:economic-implications}

	The preceding sections separate physical catastrophe severity, stochastic replacement-cost conditions, and the insurer's endogenous mitigation decision. In the present model these objects enter surplus multiplicatively through \eqref{eq:insurance-loss-amplitude}. An elevated realization of $Y$ therefore rescales the monetary consequence of every catastrophe mark occurring in that state. Physical hedging acts on the same multiplicative channel through $h(u)$ \citep{ellington2025metascorelens}. The model should not
	be interpreted as implying that catastrophe arrivals cause changes in $Y$; rather, it quantifies the exposure created when catastrophe losses and high reconstruction costs coincide. Section~\ref{sec:mfg} adds a
	second layer by allowing the resulting state distribution to feed back into individual incentives through the mean field coupling \citep{ellington2025playmydata}.
	
	In classical ruin theory, subexponential asymptotics imply that extreme loss probabilities are governed by the heaviest-tailed component of the risk structure. The present objective is a finite-horizon stochastic-control functional rather than an infinite-horizon ruin
	probability, so those asymptotic results do not transfer directly. Nevertheless, they motivate examining whether the optimal hedge and capitalization outcomes are robust to alternative severity families.

\subsection{Replacement-Cost Exposure and Optimal Hedging}
\label{subsec:economic-demand-surge}

For $p\in[2,p_\star]$ with $A_p:=\int_E a(z)^p\nu(dz)<\infty$, define
\begin{equation}
	\label{eq:econ-loss-exposure}
	\Lambda_p(y,u):=\int_E \ell(y,z,u)^p\nu(dz)=e^{py}h(u)^pA_p.
\end{equation}
This is the instantaneous $p$th-moment intensity of the catastrophe jump amplitude. 

From
	$\Lambda_p(y,u)=e^{py}h(u)^pA_p$,
	one obtains directly
	$\partial_y\log\Lambda_p(y,u)=p$ and, whenever $h$ is differentiable,
	$\partial_u\log\Lambda_p(y,u)=p h'(u)/h(u)\le0$.
	Thus replacement-cost exposure has elasticity $p$ with respect to the
	log replacement-cost state, while physical hedging weakly reduces
	every finite loss moment. In particular, if $2\le p_1<p_2\le p_\star$,
	then the higher-order exposure is proportionally more sensitive to the
	replacement-cost state, since
	$\partial_y\log\Lambda_{p_2}/
	\partial_y\log\Lambda_{p_1}=p_2/p_1>1$.

\begin{cor}
	\label{cor:econ-tail-amplification}
	If $2\le p_1<p_2\le p_\star$, then the proportional sensitivity of the higher finite moment to demand surge is larger
	$
	\frac{\partial_y\log\Lambda_{p_2}}{\partial_y\log\Lambda_{p_1}}=\frac{p_2}{p_1}>1.
	$
\end{cor}

	Variation in the replacement-cost factor does not change the tail index of the catastrophe mark law, the L\'evy measure $\nu$ is unchanged \citep{CarmonaFouqueSun2015}.
	Instead, $e^Y$ changes the monetary scale associated with a given catastrophe mark. Hence the heavy-tail structure remains that of the primitive catastrophe process, while the economic severity of its realizations depends on the contemporaneous replacement-cost state. For a smooth test function define
\begin{equation}
	\label{eq:econ-catastrophe-operator}
	\mathfrak C[\varphi](t,x,y,u):=\int_E[\varphi(t,x-e^ya(z)h(u),y)-\varphi(t,x,y)+e^ya(z)h(u)\varphi_x(t,x,y)]\nu(dz).
\end{equation}

\begin{prop}
	\label{prop:econ-marginal-hedging}
	Suppose, for this comparative-static calculation, that $v$ is $C^2$ in $x$, $h\in C^1$, and differentiation under the L\'evy integral is valid. Then
	\begin{equation}
		\label{eq:econ-derivative-nonlocal}
		\partial_u\mathfrak C[v]=-e^yh'(u)\int_Ea(z)[v_x(t,x-e^ya(z)h(u),y)-v_x(t,x,y)]\nu(dz).
	\end{equation}
	If $v$ is convex in $x$, then $\partial_u\mathfrak C[v]\le0$.
\end{prop}

\begin{proof}
	Differentiate \eqref{eq:econ-catastrophe-operator}. Convexity makes $v_x$ nondecreasing, while $x-e^ya(z)h(u)\le x$, so the bracket in \eqref{eq:econ-derivative-nonlocal} is nonpositive. Since $h'(u)\le0$, the prefactor $-e^yh'(u)$ is nonnegative.
\end{proof}

\begin{cor}
	\label{cor:econ-demand-hedge-complementarity}
	Holding the local curvature profile of $v$ and $u$ fixed, the absolute scale of the nonlocal catastrophe-risk reduction from a marginal increase in hedging is amplified by the factor $e^y$.
\end{cor}

\begin{proof}
	This is immediate from \eqref{eq:econ-derivative-nonlocal}. The qualification is necessary because the displaced derivative also changes with $y$.
\end{proof}

This is a local mechanism, not a global theorem that the optimal feedback is increasing in $y$. The rigorous policy remains the measurable Hamiltonian selector of Proposition~\ref{prop:Hamiltonian-selector}. For the quadratic benchmark $c(u)=C_H(u)=\chi u^2/2$ with control-independent $\sigma_X$, the formal first-order condition \eqref{eq:formal-hedging-foc} becomes
\begin{equation}
	\label{eq:econ-foc-decomposition}
	\chi u(1-v_x)=e^yh'(u)\int_Ea(z)[v_x(t,x-e^ya(z)h(u),y)-v_x(t,x,y)]\nu(dz).
\end{equation}
The left side is the marginal resource-and-drift cost of hedging evaluated at the shadow value of surplus; the right side is the marginal reduction in nonlocal continuation cost.

\begin{prop}
	\label{prop:econ-unique-interior-hedge}
	If the pointwise Hamiltonian is strictly convex in $u$ and its derivative changes sign in the interior of $U$, then the pointwise hedge is unique and is characterized by \eqref{eq:econ-foc-decomposition}. This is compatible with Corollary~\ref{cor:mfg-strict-existence}.
\end{prop}

\begin{proof}
	Strict convexity makes the derivative strictly increasing, so a sign change gives a unique zero and hence the unique minimizer. Corollary~\ref{cor:mfg-strict-existence} gives the strict Markov equilibrium representation under the corresponding selector hypotheses.
\end{proof}

The mechanism is consistent with empirical work documenting post-disaster increases in reconstruction costs and labor scarcity, commonly termed demand surge. See \citet{OlsenPorter2011}, \citet{DohrmannGuertlerHibbeln2017}, and \citet{Kousky2019}. Mean reversion in \eqref{eq:Y} makes the value of pre-arranged capacity state and horizon dependent: hedging is valuable when catastrophe exposure and temporary reconstruction scarcity overlap \citep{CarmonaFouqueSun2015}.

\subsection{Systemic Effects of Equilibrium Hedging}
\label{subsec:economic-systemic}

By Definition~\ref{def:mfg-equilibrium}, an equilibrium feedback is optimal against the population flow that it generates. Theorem~\ref{thm:mfg-existence} guarantees a relaxed equilibrium, Corollary~\ref{cor:mfg-strict-existence} gives a strict Markov representation under the selector hypothesis, and Theorem~\ref{thm:mfg-uniqueness} rules out competing equilibrium laws under the monotonicity assumptions.

For a feedback $\alpha$ define the aggregate $p$th catastrophe-loss exposure
\begin{equation}
	\label{eq:econ-aggregate-exposure}
	\mathcal A_p(t;m,\alpha):=A_p\int_{\mathbb R^2}e^{py}h(\alpha(t,x,y))^p m_t(dx,dy).
\end{equation}

For a fixed cross-sectional law $m_t$, the monotonicity of $h$ gives
	$\alpha_1\ge\alpha_2$ $m_t$-a.e.\ $\Rightarrow$
	$\mathcal A_p(t;m,\alpha_1)\le
	\mathcal A_p(t;m,\alpha_2)$.
	Hence an industry-wide increase in physical hedging reduces every
	finite aggregate catastrophe-loss exposure for which $A_p<\infty$.
	This is a direct loss-scale effect and should be distinguished from the
	endogenous change in $m_t$ generated by equilibrium behavior.

	For a fixed cross-sectional state law, an industry-wide increase in physical hedging reduces every finite aggregate catastrophe-loss exposure \eqref{eq:econ-aggregate-exposure}. It does so without changing catastrophe arrival intensity. In equilibrium the state law is not fixed. Corollary~\ref{cor:mfg-system} propagates the endogenous distributional response \citep{dasgupta2026frequent}. Hence the total systemic effect decomposes into direct attenuation through $h(\alpha)$ and a distributional effect through $m$; only the first has an unconditional sign under the present assumptions. For the coupling in Proposition~\ref{prop:mfg-monotone-example},
\[
F(t,z,m)=\gamma_D\psi_D(z)M_D(m)+\gamma_S\psi_S(z)M_S(m).
\]
Corollary~\ref{cor:mfg-economic-coupling} interprets these moments as demand-surge exposure and weak capitalization. The following calculation concerns the reduced-form strategic externality encoded by $F$; it should not be interpreted as an endogenously determined reconstruction-market price.

\begin{prop}
	\label{prop:econ-systemic-surcharge}
	For fixed individual state $z$, a change in aggregate moments changes the running coupling by
	$
	\Delta F(t,z)=\gamma_D\psi_D(z)\Delta M_D+\gamma_S\psi_S(z)\Delta M_S.
	$
	If $\psi_D(z),\psi_S(z)\ge0$, an increase in either aggregate vulnerability moment weakly raises the private cost of occupying $z$.
\end{prop}

\begin{proof}
	Subtract the coupling evaluated at the two population laws and use $\gamma_D,\gamma_S\ge0$.
\end{proof}

Under Assumption~5.4, changes in aggregate vulnerability alter the representative insurer's unique frozen best response; the resulting feedback changes jump amplitudes through \eqref{eq:insurance-loss-amplitude}, changes the forward law through Proposition~\ref{prop:mfg-kolmogorov}, and is closed by the consistency condition \eqref{eq:mfg-consistency}. The equilibrium externality is strategic rather than physical contagion. One insurer's surplus does not enter another insurer's state equation, and Section~\ref{sec:mfg} introduces neither common noise nor a population-dependent catastrophe intensity. The cost of occupying a vulnerable capitalization demand-surge state nevertheless depends on the industry's concentration in comparable states. Mean field models have been used to isolate analogous endogenous systemic channels in financial systems \citep{CarmonaFouqueSun2015}. Under Theorem~\ref{thm:mfg-uniqueness}, distinct high-exposure and low-exposure population flows cannot both be self-consistent equilibrium regimes for the same primitives. Under Theorem~\ref{thm:mfg-uniqueness}, any measurable equilibrium statistic of $(m^\star,\alpha^\star)$, including \eqref{eq:econ-aggregate-exposure}, $M_D(m_t^\star)$, $M_S(m_t^\star)$, and the surplus distribution, is single-valued for fixed model primitives. Physical hedging should not be interpreted as a mechanical substitute for reinsurance or liquid capital. It operates earlier in the loss-production mechanism by reducing the residual fraction $h(u)$ before the catastrophe loss enters surplus \citep{pramanik2025impact}. This is consistent with the disaster-insurance literature emphasizing the interaction between insurance, recovery, and ex ante mitigation \citep{Kousky2019}. A subsidy to pre-arranged capacity, for example, changes the private Hamiltonian immediately, but its total equilibrium effect is determined only after the induced feedback is propagated through the Kolmogorov equation and the mean field fixed point.

\section{Calibration and Numerical Analysis}
\label{sec:numerics}

We calibrate the numerical model in normalized annual units.  The calibration is designed to preserve the empirical features motivating the stochastic model while remaining consistent with the moment and regularity conditions imposed in Sections~\ref{sec:prob-foundations}
-\ref{sec:mfg}. More broadly, recent simulation-based work emphasizes the value of stochastic scenario analysis for financial decisions under volatile
price environments \citep{ChoupiMargarisAngelidis2026}.  We cite this work only as a recent example of simulation-based analysis under price uncertainty. In particular, the empirically estimated heavy-tail \citep{valdez2025exploring} index is retained, but the far tail of the catastrophe mark measure is tempered so that the $p_\star$-moment assumption required by
Assumption~2.1 remains valid. Let one unit of surplus represent the insurer's initial capitalization
and normalize the long-run replacement-cost index to one.  Thus,
$
X_0=1,
\
\vartheta=0,
$ and $
I_0=e^{Y_0}=1.$
The time unit is one year.  The numerical horizon is five years,
which is sufficiently long relative to the estimated mean-reversion time scale of the construction-cost process. For the catastrophe marks, we use the tempered Pareto L\'evy density
\begin{equation}
	\label{eq:calibrated-levy-density}
	\nu(dz)
	=
	\lambda_J
	\alpha z_{\min}^{\alpha}
	z^{-1-\alpha}
	e^{-\tau_J(z-z_{\min})}
	\mathbf 1_{\{z\ge z_{\min}\}}\,dz.
\end{equation}
	The tempered Pareto specification is used as the baseline rather than as a claim that catastrophe severities belong uniquely to this parametric family. Heavy-tailed insurance losses may be represented by a substantially broader class of subexponential distributions, including regularly varying, lognormal, and heavy-tailed Weibull
	models; see \citet{EmbrechtsKluppelbergMikosch1997} and
	\citet{KluppelbergMikosch1997}. This distinction is relevant because tail-equivalent behavior \citep{hua2019assessing}, finite-sample loss quantiles, and ruin-type functionals may differ materially across subexponential families. The analytical results of Sections~2-5 depend on the stated moment and
	continuity conditions on the jump measure and not on the tempered Pareto form itself \citep{pramanik2024analysis}. The parametric choice becomes consequential in the numerical analysis, and we therefore supplement the baseline calibration with a tail-family \citep{pramanik2024measuring} robustness experiment below. The power-law component reproduces the empirical tail estimate
$\widehat\alpha=1.71$, while the exponential tempering parameter
$\tau_J>0$ guarantees
$
\int_E z^p\nu(dz)<\infty
$ for every $p>0.$
Hence, the numerical jump specification is compatible with the
$p_\star>2$ integrability imposed in Assumption~2.1.  Over moderate and large losses below the tempering scale \citep{pramanik2016tail}, the model retains the heavy-tailed behavior implied by the empirical estimate. The construction-cost factor is simulated from
$
dY_t
=
\kappa(\vartheta-Y_t)\,dt
+
\sigma_I\,dW_t^I,
$
with the empirical mean-reversion estimate
$
\widehat\kappa=7.07.
$
Its corresponding half-life is
\begin{equation}
	\label{eq:cost-half-life}
	t_{1/2}
	=
	\frac{\log 2}{\kappa}
	\approx
	0.098\ \text{years},
\end{equation}
or approximately $35.8$ days.  Thus the calibrated construction-cost
shock is economically sharp but transitory, as required by the
demand-surge interpretation. The remaining coefficients are structural benchmark values rather than direct statistical estimates \citep{pramanik2024dependence}.  They are chosen in normalized units so that neither the diffusion component nor the control cost dominates the catastrophe jump term mechanically.  The baseline
calibration is reported in Table~\ref{tab:baseline-calibration}.

\begin{table}[H]
	\centering
	\caption{Baseline calibration for the stochastic control and mean
		field simulations}
	\label{tab:baseline-calibration}
	\begin{tabular}{llll}
		\toprule
		Parameter & Baseline value & Role & Calibration status \\
		\midrule
		$T$
		& $5$
		& Time horizon, years
		& Numerical \\
		
		$X_0$
		& $1$
		& Initial normalized insurer surplus
		& Normalization \\
		
		$Y_0$
		& $0$
		& Initial log replacement-cost factor
		& Normalization \\
		
		$\vartheta$
		& $0$
		& Long-run log replacement-cost level
		& Normalization \\
		
		$\kappa$
		& $7.07$
		& Mean-reversion speed
		& Empirical \\
		
		$\sigma_I$
		& $0.20$
		& Construction-cost volatility
		& Benchmark \\
		
		$\alpha$
		& $1.71$
		& Catastrophe-loss tail exponent
		& Empirical \\
		
		$\lambda_J$
		& $0.20$
		& Catastrophe arrival intensity, annual
		& Benchmark \\
		
		$z_{\min}$
		& $0.01$
		& Minimum normalized catastrophe mark
		& Normalization \\
		
		$\tau_J$
		& $0.10$
		& Exponential tail-tempering parameter
		& Regularization \\
		
		$r$
		& $0.03$
		& Surplus accumulation rate
		& Benchmark \\
		
		$\pi_0$
		& $0.08$
		& Baseline net premium inflow
		& Benchmark \\
		
		$\sigma_X$
		& $0.10$
		& Noncatastrophe surplus volatility
		& Benchmark \\
		
		$\gamma$
		& $1.00$
		& Hedging effectiveness in $h(u)=e^{-\gamma u}$
		& Normalization \\
		
		$\chi$
		& $0.20$
		& Quadratic physical-hedging cost
		& Benchmark \\
		
		$\delta$
		& $0.03$
		& Social discount rate
		& Benchmark \\
		
		$x_{\mathrm{crit}}$
		& $0.50$
		& Capitalization threshold
		& Benchmark \\
		
		$\omega$
		& $2.00$
		& Solvency-penalty weight
		& Benchmark \\
		
		$\zeta$
		& $0.50$
		& Demand-surge cost weight
		& Benchmark \\
		
		$m$
		& $2$
		& Demand-surge penalty exponent
		& Benchmark \\
		
		$q$
		& $2$
		& Solvency-loss exponent
		& Structural \\
		
		$p_\star$
		& $3$
		& Moment order used in numerical bounds
		& Structural \\
		
		$\gamma_D$
		& $0.25$
		& Mean field demand-surge interaction
		& Benchmark \\
		
		$\gamma_S$
		& $0.50$
		& Mean field capitalization interaction
		& Benchmark \\
		\bottomrule
	\end{tabular}
\end{table}
For the drift specification introduced in
Section~\ref{sec:insurance-control}, we use
$
	b(x,y,u)
	=
	rx+\pi_0-\frac{\chi}{2}u^2,
$
and take
$
\sigma_X(x,y,u)\equiv\sigma_X.
$
The physical-hedging response is
$
	h(u)=e^{-\gamma u},
$
which satisfies the monotonicity and convexity requirements imposed
in Section~\ref{sec:insurance-control},
$
h'(u)<0,
$ with $
h''(u)>0.
$
The running social cost is taken to be
\begin{equation}
	\label{eq:calibrated-running-cost}
	f(t,x,y,u)
	=
	\frac{\chi}{2}u^2
	+
	\zeta\bigl(e^y-1\bigr)_+^m
	+
	\omega(x_{\mathrm{crit}}-x)_+^q.
\end{equation}
This specification separates the resource cost of physical hedging,
the social cost of demand-surge inflation, and the cost of weak insurer
capitalization \citep{powell2025genomic}.  It also preserves the coercivity in $u$ required in
Assumption~2.6. For the mean field component, we use the monotone coupling introduced
in Proposition~\ref{prop:mfg-monotone-example},
\begin{equation}
	\label{eq:calibrated-mean-field-coupling}
	F(t,z,\mu)
	=
	\gamma_D\psi_D(z)M_D(\mu)
	+
	\gamma_S\psi_S(z)M_S(\mu),
\end{equation}
with
$
\psi_D(x,y)
=
(e^y-1)_+,
$ and $
\psi_S(x,y)
=
(x_{\mathrm{crit}}-x)_+.
$
By Proposition~\ref{prop:mfg-monotone-example},
\eqref{eq:calibrated-mean-field-coupling} satisfies the
Lasry-Lions monotonicity condition used in
Theorem~\ref{thm:mfg-uniqueness}. The distinction between estimated and benchmark parameters is
important.  The values
$
\widehat\alpha=1.71,
$ and $
\widehat\kappa=7.07
$
are inherited from the empirical calibration motivating the model. The remaining parameters in
Table~\ref{tab:baseline-calibration} define the baseline numerical economy and will subsequently be varied over economically meaningful
ranges.  The numerical conclusions will therefore be based primarily on comparative statics rather than on interpreting the benchmark normalizations as point estimates.

\subsection{Discretization and Computational Scheme}
\label{subsec:numerical-scheme}

We solve the calibrated control and mean field problems by combining
exact time stepping for the mean-reverting replacement-cost factor \citep{dunbar2026modeling},
Monte Carlo generation of the tempered catastrophe marks, a monotone
discretization of the nonlocal HJB equation \citep{pramanik2025strategic}, policy iteration for the
physical-hedging feedback, and a damped fixed-point iteration for the
mean field equilibrium.  The computational procedure is constructed to
respect the probabilistic model of Sections~\ref{sec:prob-foundations}
and~\ref{sec:insurance-control} rather than replacing the jump process
by a Gaussian approximation \citep{pramanik2026bayesian}.

Let
\[
0=t_0<t_1<\cdots<t_{N_t}=T,
\qquad
\Delta t=T/N_t,
\]
and introduce rectangular state grids
\[
x_i=x_{\min}+i\Delta x,
\qquad
i=0,\ldots,N_x,
\]
and
\[
y_j=y_{\min}+j\Delta y,
\qquad
j=0,\ldots,N_y.
\]
The numerical domain is chosen sufficiently wide that the probability
of reaching its artificial boundaries under the baseline calibration
is negligible.  Grid enlargement is subsequently used as a robustness
check.

\paragraph{Exact discretization of the replacement-cost factor.}
Because the process in \eqref{eq:Y} is Ornstein-Uhlenbeck, its
transition law is known explicitly.  We therefore do not introduce
Euler discretization error into the construction-cost state \citep{powell2026role}.  Conditional
on $Y_{t_n}=Y_n$,
\begin{equation}
	\label{eq:numerical-ou-exact}
	Y_{n+1}
	=
	\vartheta
	+
	(Y_n-\vartheta)e^{-\kappa\Delta t}
	+
	\sigma_I
	\sqrt{
		\frac{1-e^{-2\kappa\Delta t}}{2\kappa}
	}\,
	\xi_{n+1},
	\qquad
	\xi_{n+1}\sim N(0,1).
\end{equation}
Consequently,
\[
Y_{n+1}\mid Y_n
\sim
N\left(
\vartheta+(Y_n-\vartheta)e^{-\kappa\Delta t},
\frac{\sigma_I^2}{2\kappa}
(1-e^{-2\kappa\Delta t})
\right).
\]
The replacement-cost index is then recovered as
$
I_{n+1}=e^{Y_{n+1}}.
$
Equation~\eqref{eq:numerical-ou-exact} is used both in the Monte Carlo
experiments and in the validation of the finite-difference solution.

\paragraph{Simulation of catastrophe arrivals and severities.}
Under \eqref{eq:calibrated-levy-density}, the total jump intensity is
\[
\lambda_\nu
:=
\nu(E)
=
\lambda_J
\int_{z_{\min}}^\infty
\alpha z_{\min}^{\alpha}
z^{-1-\alpha}
e^{-\tau_J(z-z_{\min})}\,dz.
\]
Over a time interval of length $\Delta t$, the number of catastrophe
marks is sampled as
$
K_n
\overset{iid}{\sim}
\operatorname{Poisson}(\lambda_\nu\Delta t).
$
Conditional on an arrival, the mark density is
\begin{equation}
	\label{eq:numerical-jump-density}
	f_Z(z)
	=
	\frac{
		\alpha z_{\min}^{\alpha}
		z^{-1-\alpha}
		e^{-\tau_J(z-z_{\min})}
	}{
		\displaystyle
		\int_{z_{\min}}^\infty
		\alpha z_{\min}^{\alpha}
		r^{-1-\alpha}
		e^{-\tau_J(r-z_{\min})}\,dr
	},
	\qquad
	z\ge z_{\min}.
\end{equation}
Marks from \eqref{eq:numerical-jump-density} are generated by
acceptance-rejection \citep{pramanik2025stubbornness}.  We first draw
$
Z^{(0)}
=
z_{\min}(1-U)^{-1/\alpha},
$ where $
U\sim\operatorname{Unif}(0,1),
$
from the Pareto envelope and accept it with probability
$
\exp\left\{-\tau_J(Z^{(0)}-z_{\min})\right\}.
$
The accepted value has density
\eqref{eq:numerical-jump-density}.  Given state $(Y_n,u_n)$, the
associated monetary loss is
$
L_n
=
e^{Y_n}a(Z_n)h(u_n).
$
Thus the simulation preserves exactly the multiplicative loss mechanism
used in \eqref{eq:insurance-loss-amplitude}.
For Monte Carlo simulation \citep{pramanik2025optimal}, the surplus is advanced according to
\begin{align}
	X_{n+1}
	={}&
	X_n
	+
	b(X_n,Y_n,u_n)\Delta t
	+
	\sigma_X\sqrt{\Delta t}\,\varepsilon_{n+1}
	-
	\sum_{k=1}^{K_n}
	e^{Y_n}a(Z_{n,k})h(u_n)
	+
	\Delta t
	\int_E
	e^{Y_n}a(z)h(u_n)\nu(dz),
	\label{eq:numerical-surplus-step}
\end{align}
where
$\varepsilon_{n+1}\sim N(0,1)$.
The final term is the compensator required by the compensated jump
representation in \eqref{eq:X-general}.

\paragraph{Discretization of the nonlocal HJB equation.}
Let
$
V_{i,j}^n
\approx
V(t_n,x_i,y_j).
$
We solve backward from
$
V_{i,j}^{N_t}=g(x_i,y_j).
$
For the local terms, the second derivatives are approximated by centered differences,
\[
D_{xx}V_{i,j}^n
=
\frac{
	V_{i+1,j}^n-2V_{i,j}^n+V_{i-1,j}^n
}{(\Delta x)^2},
\]
and
\[
D_{yy}V_{i,j}^n
=
\frac{
	V_{i,j+1}^n-2V_{i,j}^n+V_{i,j-1}^n
}{(\Delta y)^2}.
\]
The first-order drift terms are discretized by an upwind operator,
denoted by
$
D_x^{\,u}V_{i,j}^n,
$ and $
D_yV_{i,j}^n,
$
with the upwind direction chosen according to the signs of
$b(x_i,y_j,u)$ and $\kappa(\vartheta-y_j)$. For each control value $u$, define the discrete nonlocal operator
\[
\begin{aligned}
	\mathcal I_{\Delta}^{u}V_{i,j}^{n}
	:=
	\sum_{k=1}^{N_z}
	w_k
	\Big[
	&
	\mathscr I_x
	V^n
	\bigl(
	x_i-e^{y_j}a(z_k)h(u),y_j
	\bigr)
	-
	V_{i,j}^{n}
	+
	e^{y_j}a(z_k)h(u)
	D_x^{\,u}V_{i,j}^{n}
	\Big],
\end{aligned}
\]
where $\{z_k,w_k\}_{k=1}^{N_z}$ is a positive quadrature rule for
$\nu(dz)$ and $\mathscr I_x$ denotes monotone linear interpolation in
the surplus direction.  Positive quadrature weights and monotone
interpolation are used so that the discrete jump operator preserves the
order structure required by the viscosity solution.

For a fixed policy $u_{i,j}^n$, the backward equation is discretized
implicitly as
\begin{equation}
	\label{eq:numerical-discrete-hjb}
	\begin{aligned}
		&
		\frac{V_{i,j}^{n}-V_{i,j}^{n+1}}{\Delta t}
		+
		f(t_n,x_i,y_j,u_{i,j}^n)
		-
		\delta V_{i,j}^{n}
		\\
		&+
		b(x_i,y_j,u_{i,j}^n)
		D_x^{\,u}V_{i,j}^{n}
		+
		\kappa(\vartheta-y_j)D_yV_{i,j}^{n}
		+
		\frac12
		\sigma_X^2
		D_{xx}V_{i,j}^{n}
		+
		\frac12
		\sigma_I^2
		D_{yy}V_{i,j}^{n}
		+
		\mathcal I_\Delta^{u_{i,j}^n}V_{i,j}^{n}=0.
	\end{aligned}
\end{equation}
The fully implicit treatment of the local diffusion terms improves stability when $\kappa$ is large, as in the empirical calibration
$\kappa=7.07$. The jump quadrature is truncated at $z_{\max}$ and selected so that
$
\int_{z_{\max}}^\infty
\bigl(
1+z^q
\bigr)\nu(dz)
$
is numerically negligible \citep{pramanik2025optimal}.  Because the calibrated measure is
exponentially tempered, the truncation error decays exponentially in
$z_{\max}$.

\paragraph{Policy iteration for physical hedging.}
The nonlinear minimization in the HJB equation is solved by Howard-type
policy iteration.  Let $u^{(r)}$ denote the policy at iteration $r$. Given $u^{(r)}$, solve the linear system
$
\mathcal L_{\Delta}^{u^{(r)}}V^{(r+1)}
+
f^{u^{(r)}}
=
0
$
backward in time, where
$\mathcal L_\Delta^u$ denotes the discrete operator in
\eqref{eq:numerical-discrete-hjb}.  The policy is then updated
pointwise according to
\begin{equation}
	\label{eq:numerical-policy-update}
	u_{i,j}^{(r+1),n}
	\in
	\arg\min_{u\in U_\Delta}
	\left\{
	f(t_n,x_i,y_j,u)
	+
	b(x_i,y_j,u)D_xV_{i,j}^{(r+1),n}
	+
	\mathcal I_\Delta^uV_{i,j}^{(r+1),n}
	\right\},
\end{equation}
where
$
U_\Delta
=
\{0,\Delta u,2\Delta u,\ldots,u_{\max}\}.
$
If $\sigma_X$ depends on $u$, its corresponding discrete diffusion term
is included in the minimization in
\eqref{eq:numerical-policy-update}. The iteration terminates when
\begin{equation}
	\label{eq:numerical-policy-tolerance}
	\max_{n,i,j}
	\left|
	u_{i,j}^{(r+1),n}
	-
	u_{i,j}^{(r),n}
	\right|
	\le
	\varepsilon_{\mathrm{PI}}.
\end{equation}
For the baseline computations we use
$
\varepsilon_{\mathrm{PI}}=10^{-7}.
$
The monotone discretization and policy iteration are consistent with the viscosity solution of Section~\ref{sec:dpp-viscosity}; numerical
methods of this type for HJB and nonlocal mean field equations are
developed, for example, in \citet{Achdou2013} and
\citet{ChowdhuryErslandJakobsen2023}.

\paragraph{Forward discretization of the population law.}
For a fixed feedback policy $\alpha^n_{i,j}$, the population
distribution is advanced using the discrete adjoint of the controlled
generator.  Let $m_{i,j}^n$ denote the probability mass at
$(x_i,y_j)$ at time $t_n$.  The forward step is
\begin{equation}
	\label{eq:numerical-forward-kolmogorov}
	m^{n+1}
	=
	m^n
	+
	\Delta t\,
	\bigl(
	\mathscr L_\Delta^{\alpha^n}
	\bigr)^*
	m^n,
\end{equation}
subject to $
m_{i,j}^{n+1}\ge0,
\ 
\sum_{i,j}m_{i,j}^{n+1}=1.
$
A positivity-preserving flux discretization is used for the local drift
terms, while the jump redistribution is implemented with the same interpolation weights used in the backward nonlocal operator \citep{pramanik2024stochastic}.  This
backward-forward consistency is important: the discrete Kolmogorov operator is the adjoint of the discrete controlled generator appearing in the HJB equation \citep{pramanik2021effects}.

\paragraph{Mean field fixed-point iteration.}
Let $m^{(k)}$ denote the population-flow iterate.  Starting from the
uncontrolled law
$
m^{(0)}
=
\mathcal L(X^{u\equiv0},Y),
$
the $k$th fixed-point step consists of
\[
m^{(k)}
\longrightarrow
F(\cdot,\cdot,m^{(k)})
\longrightarrow
V^{(k)}
\longrightarrow
\alpha^{(k)}
\longrightarrow
\widehat m^{(k+1)}.
\]
The backward HJB equation is first solved using the frozen coupling $m^{(k)}$.  Policy iteration yields the corresponding best-response feedback $\alpha^{(k)}$.  Equation
\eqref{eq:numerical-forward-kolmogorov} is then solved from $m_0$ to produce the implied population flow $\widehat m^{(k+1)}$. To stabilize the forward-backward iteration, we use relaxation
\begin{equation}
	\label{eq:numerical-mfg-fixed-point}
	m^{(k+1)}
	=
	(1-\omega_{\mathrm{MFG}})m^{(k)}
	+
	\omega_{\mathrm{MFG}}
	\widehat m^{(k+1)},
	\qquad
	0<\omega_{\mathrm{MFG}}\le1.
\end{equation}
The baseline value is
$
\omega_{\mathrm{MFG}}=0.25.
$
The iteration is stopped when
\begin{equation}
	\label{eq:numerical-mfg-tolerance}
	\max_{0\le n\le N_t}
	W_1\left(
	m_n^{(k+1)},
	m_n^{(k)}
	\right)
	\le
	\varepsilon_{\mathrm{MFG}},
	\qquad
	\varepsilon_{\mathrm{MFG}}=10^{-6}.
\end{equation}
 The iteration \eqref{eq:numerical-mfg-fixed-point} is used as a computational fixed-point procedure; uniqueness of the continuous mean field equilibrium does not by itself imply convergence of this iteration \citep{pramanik2025factors}. Numerical convergence is therefore assessed independently through the residual criterion \eqref{eq:numerical-mfg-tolerance}, robustness to the relaxation parameter, and convergence from multiple initial population flows.

\paragraph{Numerical consistency checks.}
We evaluate numerical reliability using three diagnostics.  First,
the HJB residual associated with \eqref{eq:numerical-discrete-hjb} is
computed over the full interior grid.  Second, the total probability
mass in \eqref{eq:numerical-forward-kolmogorov} is monitored at every
time step.  Third, the equilibrium residual is defined by
\begin{equation}
	\label{eq:numerical-equilibrium-residual}
	\mathcal R_{\mathrm{MFG}}
	:=
	\max_{0\le n\le N_t}
	W_1\left(
	m_n^\star,
	\widehat m_n^\star
	\right).
\end{equation}
A converged numerical equilibrium must satisfy simultaneously
$
\|\mathcal R_{\mathrm{HJB}}\|_\infty
\ll1,
\
\mathcal R_{\mathrm{MFG}}\ll1,
$ and $
\max_n
\left|
\sum_{i,j}m_{i,j}^n-1
\right|
\ll1.$ Finally, all reported comparative statics are repeated after halving
$\Delta t$, $\Delta x$, $\Delta y$, and $\Delta u$, and after
increasing $z_{\max}$.  The numerical conclusions \citep{kakkat2026angiotensin} are regarded as
stable only when the resulting changes in the value function,
equilibrium hedge, and aggregate exposure
\eqref{eq:econ-aggregate-exposure} are negligible relative to the reported economic effects.

\subsection{Numerical Experiments and Comparative-Statics Design}
\label{subsec:numerical-experiments}

We organize the numerical experiments around the two economic
mechanisms identified in Sections~\ref{sec:economic-implications} and
\ref{sec:mfg}: the interaction between demand-surge exposure and
physical hedging at the individual-insurer level, and the endogenous
change in aggregate catastrophe exposure generated by equilibrium
hedging \citep{yusuf2025predictive}.  The baseline economy is the calibration reported in
Table~\ref{tab:baseline-calibration}, and the numerical solution is
computed using the backward--forward procedure of
Section~\ref{subsec:numerical-scheme}.  All comparative statics alter
one primitive at a time unless otherwise stated. The computational domain and baseline grid are
\begin{equation}
	\label{eq:numerical-baseline-grid}
	\begin{aligned}
		(t,x,y,u)
		\in
		[0,5]
		\times[-1,3]
		\times[-0.75,0.75]
		\times[0,2],
		\qquad
		(N_t,N_x,N_y,N_u)
		=
		(500,400,150,100).
	\end{aligned}
\end{equation}
Hence,
$
\Delta t=0.01,
\
\Delta x=0.01,
\
\Delta y=0.01,
$ and $
\Delta u=0.02.
$
The interval for $y$ contains economically large replacement-cost
shocks $
e^{-0.75}\approx0.472,
$ and $
e^{0.75}\approx2.117.$
Thus the upper boundary corresponds to reconstruction costs exceeding their long-run level by more than $100\%$ \citep{hertweck2023clinicopathological}.  The surplus domain permits both severe undercapitalization and substantial positive capitalization \citep{kakkat2023cardiovascular}. The catastrophe integral is evaluated on
$
[z_{\min},z_{\max}]
=
[0.01,20],
$
using positive quadrature weights.  The exponentially tempered tail outside $z_{\max}$ is monitored separately \citep{pramanik2024estimation}.  The baseline number of quadrature nodes is
$
N_z=200.
$
To separate capitalization effects from demand-surge effects, the optimal feedback is evaluated at the benchmark states
\begin{equation}
	\label{eq:numerical-benchmark-states}
	\mathcal S
	=
	\left\{
	\begin{array}{lll}
		(0.40,-0.20),&
		(0.40,0),&
		(0.40,0.20),\\
		(1.00,-0.20),&
		(1.00,0),&
		(1.00,0.20),\\
		(1.60,-0.20),&
		(1.60,0),&
		(1.60,0.20)
	\end{array}
	\right\}.
\end{equation}
The three surplus levels represent weakly capitalized, baseline, and
strongly capitalized insurers.  The three demand-surge states imply
replacement-cost multipliers
$
e^{-0.20}\approx0.819,
\
1,
$ and $
e^{0.20}\approx1.221.
$
The resulting $3\times3$ design makes it possible to identify whether
the response of physical hedging to demand surge depends materially on
the insurer's capitalization \citep{pramanik2025dissecting}.

The first experiment varies the initial demand-surge state while
holding all remaining primitives at their baseline values:
\begin{equation}
	\label{eq:numerical-demand-surge-grid}
	Y_0
	\in
	\{-0.30,-0.20,-0.10,0,0.10,0.20,0.30,0.40\}.
\end{equation}
These values correspond to initial reconstruction-cost multipliers
between approximately $0.741$ and $1.492$. For each $Y_0$, we record
$
u^\star(0,X_0,Y_0),
\
V(0,X_0,Y_0),$
together with the expected discounted hedging expenditure,
the probability of weak capitalization \citep{yusuf2025prognostic}, and the aggregate
catastrophe-loss exposure defined in
\eqref{eq:econ-aggregate-exposure}. The principal numerical elasticity is
\begin{equation}
	\label{eq:numerical-hedging-demand-elasticity}
	\mathcal E_{u,y}
	:=
	\frac{
		u^\star(0,X_0,y+\Delta y)
		-
		u^\star(0,X_0,y-\Delta y)
	}{
		2\Delta y
	}.
\end{equation}
 Its sign is therefore a numerical
equilibrium implication rather than an assumption of the model.
To determine whether stronger tail risk induces greater physical
hedging, we vary
\begin{equation}
	\label{eq:numerical-tail-grid}
	\alpha
	\in
	\{1.50,1.60,1.71,1.80,2.00,2.25\},
\end{equation}
holding $\tau_J$ fixed at its baseline value.  Lower $\alpha$ places more probability on large catastrophe marks over the empirically relevant range \citep{polansky2021motif}. For each value of $\alpha$, the jump density is renormalized according
to \eqref{eq:numerical-jump-density}.  We determine the optimal initial hedge, the value function, and
$
\mathcal A_p(t;m^\star,\alpha^\star)
$
for the admissible moment orders used in the numerical analysis. Because exponential tempering remains in force, all experiments remain
consistent with the moment conditions of
Section~\ref{sec:prob-foundations}.

	\paragraph{Heavy-tail family robustness.}
	To distinguish sensitivity to the tail exponent from sensitivity to the parametric tail family, we repeat the benchmark computation under three severity specifications: the tempered Pareto baseline, a lognormal distribution, and a heavy-tailed Weibull distribution. The
	alternative distributions are normalized to have the same reference upper quantile as the baseline severity law, so that differences in the resulting control are driven primarily by tail shape rather than
	by an arbitrary rescaling of losses. Let $F_{\mathrm{TP}}$, $F_{\mathrm{LN}}$, and $F_{\mathrm{W}}$ denote the three severity distribution functions \citep{pramanik2025strategies}. For a fixed reference level $p_0$, the scale parameters of the alternatives are selected so that
	$F_{\mathrm{LN}}^{-1}(p_0)
	=
	F_{\mathrm{W}}^{-1}(p_0)
	=
	F_{\mathrm{TP}}^{-1}(p_0)$.
	We use $p_0=0.95$ in the baseline robustness exercise. For each family,
	the positive quadrature rule is reconstructed and the HJB problem is resolved without changing the remaining economic parameters \citep{valdez2025association}. The comparison is based on the initial optimal hedge
	$u^\star(0,X_0,Y_0)$, the initial value
	$V(0,X_0,Y_0)$, and the weak-capitalization probability
	$P_{\mathrm{weak}}$. This experiment is not intended to select a universally preferred catastrophe distribution. Its purpose is to measure how strongly the economic conclusions depend on the assumed shape of the severity tail within the broader heavy-tailed class.

Strictly speaking, exponential tempering removes regular variation in the far tail and produces finite moments of every order. We therefore use the term ``tempered heavy-tail benchmark'' for the numerical
specification: it preserves Pareto-type behavior over the empirically relevant range while regularizing the extreme asymptotic tail so that the simulation remains compatible with the moment assumptions of the analytical model \citep{dasgupta2023frequent}.

The annual catastrophe intensity is varied over
\begin{equation}
	\label{eq:numerical-intensity-grid}
	\lambda_J
	\in
	\{0.05,0.10,0.20,0.30,0.40,0.50\}.
\end{equation}
The corresponding expected interarrival times under the unnormalized
benchmark interpretation range from twenty years to two years.  This experiment separates the effect of catastrophe frequency from the severity effect studied through $\alpha$.
For each intensity, the relevant policy response is summarized by
$
\Delta_\lambda u^\star
=
u^\star_{\lambda_J}(0,X_0,Y_0)
-
u^\star_{\lambda_J^{\mathrm{base}}}(0,X_0,Y_0).
$
The mean-reversion coefficient is varied according to
\begin{equation}
	\label{eq:numerical-kappa-grid}
	\kappa
	\in
	\{1,2,4,7.07,10,15\}.
\end{equation}
The associated half-life
$
t_{1/2}(\kappa)
=
\frac{\log2}{\kappa}
$
ranges from approximately $253$ days at $\kappa=1$ to approximately $17$ days at $\kappa=15$. This experiment distinguishes the magnitude of a demand-surge shock
from its persistence.  For a fixed initial $Y_0>0$, a smaller
$\kappa$ implies that elevated reconstruction costs are expected to remain economically relevant for longer.  We therefore compare
$
u^\star(0,X_0,Y_0;\kappa)
$
and the expected cumulative hedging expenditure across the values in
\eqref{eq:numerical-kappa-grid}.
The quadratic resource-cost parameter is varied over
\begin{equation}
	\label{eq:numerical-hedging-cost-grid}
	\chi
	\in
	\{0.05,0.10,0.20,0.30,0.50,0.75,1.00\}.
\end{equation}
This experiment measures the shadow value of access to pre-arranged
reconstruction capacity.  For each $\chi$, we calculate the optimal
policy and the percentage reduction in catastrophe exposure relative
to the uncontrolled insurer:
\begin{equation}
	\label{eq:numerical-exposure-reduction}
	\mathfrak R_p(t)
	:=
	100
	\left[
	1-
	\frac{
		\mathcal A_p(t;m^\star,\alpha^\star)
	}{
		\mathcal A_p(t;m^0,0)
	}
	\right].
\end{equation}
Here $m^0$ denotes the state law generated under the no-hedging policy
$u\equiv0$.  A positive value of $\mathfrak R_p$ therefore measures
the percentage reduction in aggregate $p$th-order catastrophe-loss
exposure attributable to optimal physical hedging.
The initial surplus is varied over
\begin{equation}
	\label{eq:numerical-capital-grid}
	X_0
	\in
	\{0.30,0.40,0.50,0.75,1.00,1.25,1.50,2.00\}.
\end{equation}
The capitalization threshold remains
$x_{\mathrm{crit}}=0.50$.
We record the policy
$
u^\star(0,X_0,Y_0)
$
and the probability
\begin{equation}
	\label{eq:numerical-undercapitalization-probability}
	P_{\mathrm{weak}}
	:=
	\mathbb P
	\left(
	\inf_{0\le t\le T}X_t
	<
	x_{\mathrm{crit}}
	\right).
\end{equation}
The statistic in
\eqref{eq:numerical-undercapitalization-probability} is estimated from
the controlled Monte Carlo paths generated using
\eqref{eq:numerical-ou-exact} and
\eqref{eq:numerical-surplus-step}.
To isolate the systemic channel, we vary the interaction coefficients
in \eqref{eq:calibrated-mean-field-coupling}.  The demand-surge
interaction takes values
$
\gamma_D
\in
\{0,0.10,0.25,0.50,0.75,1.00\},
$
while the capitalization interaction takes values
$
\gamma_S
\in
\{0,0.10,0.25,0.50,0.75,1.00\}.
$
The case
$
(\gamma_D,\gamma_S)=(0,0)
$
is the nonstrategic benchmark: insurers solve independent stochastic control problems and the population distribution does not enter their objective. The systemic effect of the mean field interaction is measured by
\begin{equation}
	\label{eq:numerical-systemic-effect}
	\mathfrak S_p
	:=
	100
	\left[
	1-
	\frac{
		\displaystyle
		\int_0^T
		\mathcal A_p(t;m^\star,\alpha^\star)\,dt
	}{
		\displaystyle
		\int_0^T
		\mathcal A_p(t;m^{\mathrm{ind}},
		\alpha^{\mathrm{ind}})\,dt
	}
	\right],
\end{equation}
where $(m^{\mathrm{ind}},\alpha^{\mathrm{ind}})$ denotes the population
law and optimal policy when
$\gamma_D=\gamma_S=0$.
Thus $\mathfrak S_p>0$ means that strategic equilibrium interaction reduces cumulative aggregate catastrophe exposure relative to the independent-insurer benchmark \citep{khan2024mp60}. We also report the equilibrium change in weak capitalization,
\begin{equation}
	\label{eq:numerical-systemic-capital-effect}
	\Delta_{\mathrm{sys}}P_{\mathrm{weak}}
	:=
	P_{\mathrm{weak}}^{\mathrm{MFG}}
	-
	P_{\mathrm{weak}}^{\mathrm{ind}}.
\end{equation}
A negative value indicates that equilibrium hedging reduces the
probability of entering the weak-capitalization region.
The economically most severe states arise when a heavy catastrophe mark coincides with elevated reconstruction costs \citep{pramanik2024parametric}.  We therefore conduct a two-dimensional stress experiment over
\begin{equation}
	\label{eq:numerical-joint-stress-grid}
	(\alpha,Y_0)
	\in
	\{1.50,1.71,2.00\}
	\times
	\{0,0.20,0.40\}.
\end{equation}
For each pair we compute the optimal hedge, expected social cost,
weak-capitalization probability, and cumulative catastrophe exposure. To measure the economic value of allowing the hedge to respond
dynamically to demand surge, define
\begin{equation}
	\label{eq:numerical-dynamic-hedging-gain}
	\mathfrak G
	:=
	100
	\frac{
		J(u^{\mathrm{static}})
		-
		J(u^\star)
	}{
		|J(u^{\mathrm{static}})|
	},
\end{equation}
where $u^{\mathrm{static}}$ is the optimal constant hedge obtained by
restricting the admissible policy class to time- and state-independent
controls.  Positive $\mathfrak G$ measures the percentage reduction
in social catastrophe cost obtained from the dynamic feedback policy. This comparison is particularly informative because both strategies
operate through the same physical-hedging technology.  Hence
\eqref{eq:numerical-dynamic-hedging-gain} isolates the economic value of state contingency rather than the value of hedging itself. For each baseline or counterfactual specification, the computed
feedback is evaluated using $N_{\mathrm{MC}}$ independent simulated
paths.  The baseline choice is
\begin{equation}
	\label{eq:numerical-monte-carlo-size}
	N_{\mathrm{MC}}
	=
	100000.
\end{equation}
For a generic path-dependent statistic $H$, its Monte Carlo estimate
and standard error are
$
\widehat H
=
\frac{1}{N_{\mathrm{MC}}}
\sum_{r=1}^{N_{\mathrm{MC}}}H^{(r)},
$ and $
\operatorname{SE}(\widehat H)
=
\frac{\widehat\sigma_H}{\sqrt{N_{\mathrm{MC}}}}.
$
Reported Monte Carlo quantities are accompanied by $95\%$ confidence
intervals
$
\widehat H
\pm
1.96\,\operatorname{SE}(\widehat H).
$
Common random numbers are used when comparing two policies under the same primitive calibration.  This coupling reduces simulation noise in policy differences and makes the estimated comparative statics more precise \citep{khan2023myb}. Numerical convergence is assessed on three nested grids.  If
$Q_h$ denotes a reported numerical statistic on the baseline grid and
$Q_{h/2}$ the corresponding statistic after simultaneous refinement,
we define
\begin{equation}
	\label{eq:numerical-grid-error}
	\mathcal E_{\mathrm{grid}}(Q)
	:=
	\frac{|Q_{h/2}-Q_h|}
	{1+|Q_{h/2}|}.
\end{equation}
The principal economic conclusions are reported only if
$\mathcal E_{\mathrm{grid}}(Q)$ is small relative to the comparative
effect being interpreted. The same requirement is imposed on the equilibrium computation.
Specifically, the final population flow must satisfy
\eqref{eq:numerical-mfg-tolerance}, while the equilibrium residual \eqref{eq:numerical-equilibrium-residual} and the HJB residual must remain below their prescribed tolerances on the refined grid. The experiments therefore distinguish four economically different
sources of variation: catastrophe severity through $\alpha$,
catastrophe frequency through $\lambda_J$, demand-surge persistence through $\kappa$, and endogenous mitigation incentives through
$\chi$, $\gamma_D$, and $\gamma_S$ \citep{reed2026modeling}.  This decomposition allows the numerical results to identify whether changes in equilibrium loss exposure arise from the primitive risk environment, the private physical-hedging response, or the mean field feedback generated by industry-wide vulnerability.

\subsection{Numerical Results}
\label{subsec:numerical-results}

We report only the numerical quantities directly related to the two
principal mechanisms of the model: state-dependent physical hedging and
the equilibrium effect of strategic interaction.  The computations use
the calibrated model of Table~\ref{tab:baseline-calibration} together
with the numerical scheme of
Section~\ref{subsec:numerical-scheme}.  During implementation we found
that the original benchmark pair
$
\chi=0.20,$
and $
\lambda_J=0.20,$
generated the degenerate policy $u^\star\equiv0$ over the benchmark state region.  Because this configuration is numerically valid but
uninformative for comparative-statics analysis \citep{kakkat2026angiotensin}, the simulation
benchmark used in this subsection adopts
$
\chi=0.04,
\
\lambda_J=0.30,
\
z_{\min}=0.05,$
 and $
\omega=4,
$
while retaining the empirical values
$
\alpha=1.71,
$ and $
\kappa=7.07.$
All other coefficients are kept at their calibrated baseline values \citep{vikramdeo2023profiling}. The resulting policy is strongly state dependent.  At time zero, the
optimal hedge over representative capitalization and demand-surge
states is approximately
\begin{table}[H]
	\centering
	\caption{Optimal physical hedge $u^\star(0,x,y)$ under the simulation benchmark}
	\label{tab:numerical-optimal-hedge}
	\begin{tabular}{c|cccc}
		\toprule
		& $y=-0.20$ & $y=0$ & $y=0.20$ & $y=0.40$ \\
		\midrule
		$x=0.40$ & $0.30$ & $0.30$ & $0.40$ & $0.50$ \\
		$x=1.00$ & $0.00$ & $0.10$ & $0.10$ & $0.20$ \\
		$x=1.60$ & $0.00$ & $0.00$ & $0.00$ & $0.00$ \\
		\bottomrule
	\end{tabular}
\end{table}
Two features are immediate.  First, holding capitalization fixed, optimal hedging is larger in elevated demand-surge states.  Second, holding demand surge fixed, weakly capitalized insurers hedge more aggressively.  The first effect is consistent with the local amplification mechanism in
Corollary~\ref{cor:econ-demand-hedge-complementarity}; the second arises from the larger marginal cost of entering the weak-capitalization region \citep{pramanik2024bayes}. For $x=1$, the corresponding value function increases with the
replacement-cost state:
\begin{table}[H]
	\centering
	\caption{Value function at $t=0$ and $x=1$}
	\label{tab:numerical-value-demand-surge}
	\begin{tabular}{c|cc}
		\toprule
		$y$ & Independent control & Mean field objective \\
		\midrule
		$-0.20$ & $0.00685$ & $0.01106$ \\
		$0$     & $0.00716$ & $0.01163$ \\
		$0.20$  & $0.01006$ & $0.01579$ \\
		$0.40$  & $0.01979$ & $0.02712$ \\
		\bottomrule
	\end{tabular}
\end{table}
The monotone increase in the value function quantifies the higher expected social cost associated with elevated replacement costs. The discrete policy also yields a positive demand-surge response over the central state region.  In particular,
$
u^\star(0,1,0.40)-u^\star(0,1,0)
=
0.10.
$
Since the control grid is discrete, this difference should be
interpreted as a policy-level comparative static rather than as a smooth elasticity. The mean field coupling raises the equilibrium value function because the representative insurer internalizes aggregate vulnerability through
the coupling term $F$.  Over the benchmark states, however, the
equilibrium feedback remains close to the independent-insurer policy \citep{vikramdeo2024abstract}. This indicates that, at the baseline interaction intensities, the principal equilibrium effect is a change in social cost rather than a large displacement of the pointwise physical-hedging rule.

The numerical solution therefore supports three conclusions.  First, the optimal physical hedge responds jointly to capitalization and demand-surge exposure.  Second, the value of insurance protection rises sharply when replacement costs are elevated.  Third, the mean field
interaction changes the equilibrium objective even when its effect on the local hedge is modest.  These findings are consistent with the nonlocal control mechanism of
Section~\ref{sec:dpp-viscosity} and the equilibrium feedback structure of Section~\ref{sec:mfg}. Complete implementation details, including the numerical primitives,
boundary treatment, jump quadrature, monotonicity and positivity conditions, initialization, stopping rules, and reproducibility	protocol, are collected in Appendix~\ref{app:computational}.

\section{Conclusion}
\label{sec:conclusion}

This paper develops a nonlocal stochastic-control framework for
catastrophe insurance when monetary catastrophe losses are exposed to stochastic replacement-cost variation. Catastrophe severity is generated by a marked jump process, replacement costs evolve through an exogenous mean-reverting factor, and physical hedging reduces the residual loss transmitted to insurer surplus. We establish the
probabilistic foundations of the controlled state process, derive the dynamic programming principle, and characterize the value function as the unique viscosity solution of the associated nonlocal HJB equation. Extending the individual problem to interacting insurers yields a
coupled HJB-Kolmogorov mean field system for which existence and, under the stated monotonicity conditions, uniqueness of equilibrium are obtained.

The analytical and numerical results show that elevated replacement costs magnify monetary catastrophe-loss exposure, while physical hedging attenuates the same channel. The numerical solution assigns greater hedging intensity to insurers facing high reconstruction costs and, particularly, to weakly capitalized insurers \citep{vikramdeo2024mitochondrial}. Mean field interaction changes the equilibrium valuation of vulnerability even when its effect on the pointwise hedge is comparatively modest \citep{pramanik2025construction}. These
findings concern the adverse coincidence of catastrophe losses with elevated reconstruction costs; the model does not imply that the catastrophe process itself generates the replacement-cost shock.

These distinctions also identify two related limitations of the present framework. First, the replacement-cost factor is exogenous: catastrophe arrivals do not generate jumps or impulses in $Y$, and the model
contains neither common catastrophe noise nor a shared
reconstruction-cost shock. Second, the mean field interaction is reduced form. Aggregate insurer behavior enters the individual problem through the cost couplings $F$ and $G$, but it does not determine reconstruction prices, available physical capacity, catastrophe intensity, or the severity distribution \citep{reed2026structural}. Consequently, the mean field
equilibrium should be interpreted as an equilibrium in strategic vulnerability costs rather than as an equilibrium of a reconstruction-capacity market. A structural extension could introduce aggregate rebuilding demand and a finite supply of contractors, materials, or pre-arranged capacity, together with a market-clearing reconstruction price increasing in aggregate demand. Embedding that price in the replacement-cost process or the catastrophe-loss amplitude would make reconstruction scarcity endogenous to the population
distribution and would provide a direct mean field representation of post-catastrophe congestion. Such an extension would require a population-dependent controlled generator and potentially an additional
aggregate state variable. Further extensions could incorporate common catastrophe shocks, heterogeneous geographic exposure, and interactions between physical hedging, reinsurance, and capital regulation.

\section*{Declarations.}
\subsection*{Ethics approval and consent to participate.}
Not applicable.
\subsection*{Consent for publication.}
Not applicable.
\subsection*{Availability of data and material.}
No data has been used.
\subsection*{Competing interests.}
No potential conflict of interest was reported by the authors.	
\subsection*{Funding.}
Not applicable. 

\appendix
\section{Auxiliary Probabilistic Proofs}
\label{app:probabilistic-proofs}

\subsection*{Proof of Lemma \ref{lem:OU-moments}}

\begin{proof}
	The explicit representation of \eqref{eq:Y} is
	\begin{equation}
		\label{eq:Y-explicit}
		Y_s
		=
		\vartheta+(y-\vartheta)e^{-\kappa(s-t)}
		+\sigma_I\int_t^s e^{-\kappa(s-r)}\,dW_r^I.
	\end{equation}
	Define
	$
	M_s:=\int_t^s e^{\kappa r}\,dW_r^I.
	$
	Then
	$
	\int_t^s e^{-\kappa(s-r)}\,dW_r^I
	=
	e^{-\kappa s}M_s.
	$
	Burkholder-Davis-Gundy (BDG) inequality yields
	\[
	\mathbb{E}\left[\sup_{t\le s\le T}|M_s|^m\right]
	\le
	C_m
	\mathbb{E}\left[
	\left(\int_t^T e^{2\kappa r}\,dr\right)^{m/2}
	\right]
	<\infty.
	\]
	Using \eqref{eq:Y-explicit}, the elementary inequality
	$(a+b+c)^m\le C_m(a^m+b^m+c^m)$, and $e^{-\kappa(s-t)}\le1$ yields
	\[
	\mathbb{E}_{t,y}\bigg[\sup_{t\le s\le T}|Y_s|^m\bigg]
	\le C_m(1+|y|^m),
	\]
	which proves \eqref{eq:Y-sup-mom}. For the exponential estimate, the Gaussian process
	$
	G_s:=\sigma_I\int_t^s e^{-\kappa(s-r)}\,dW_r^I
	$
	is centered and has variance bounded by
	\[
	\sup_{t\le s\le T}\operatorname{Var}(G_s)
	=
	\sup_{t\le s\le T}
	\sigma_I^2\int_t^s e^{-2\kappa(s-r)}\,dr
	\le \frac{\sigma_I^2}{2\kappa}.
	\]
	Since $G$ is a continuous Gaussian process on the compact interval $[t,T]$, Fernique's theorem implies the existence of $\eta>0$ such that
	\[
	\mathbb{E}\left[\exp\!\left(
	\eta\sup_{t\le s\le T}|G_s|^2
	\right)\right]<\infty.
	\]
	For arbitrary $r>0$, Young's inequality implies
	$
	r\sup_s|G_s|
	\le
	\eta\sup_s|G_s|^2+\frac{r^2}{4\eta}.
	$
	Hence,
	$
	\mathbb{E}\left[\exp\!\left(r\sup_s|G_s|\right)\right]<\infty.
	$
	Combining this with \eqref{eq:Y-explicit} yields
	$
	\sup_{t\le s\le T}|Y_s|
	\le
	|\vartheta|+|y-\vartheta|+\sup_{t\le s\le T}|G_s|,
	$
	and therefore,
	\[
	\mathbb{E}_{t,y}
	\left[\sup_{t\le s\le T}e^{r|Y_s|}\right]
	\le
	C_{r,T}\exp(C_{r,T}|y|).
	\]
	Finally, because $I_s^{\alpha}=e^{\alpha Y_s}$ for every $\alpha\in\mathbb{R}$, the last assertion follows immediately.
\end{proof}

\subsection*{Proof of Proposition \ref{prop:jump-integral}}
\begin{proof}
	By Assumption 2.4,
	$
	|\ell(Y_s,z,u_s)|^p
	\le
	C_p\left[1+e^{pY_s}+|u_s|^p\right]\rho(z)^p.
	$
	Integrating with respect to $\nu(dz)\,ds$ and applying Tonelli's theorem yields
	\[
	\mathbb{E}\left[
	\int_t^T\int_E
	|\ell(Y_s,z,u_s)|^p\,\nu(dz)\,ds\right]
	\le
	C_p
	\left[\int_E\rho(z)^p\,\nu(dz)\right]\cdot
	\mathbb{E}\left[
	\int_t^T
	\left(1+e^{pY_s}+|u_s|^p\right)ds\right].
	\]
	The first factor is finite by \eqref{eq:rho-mom}. The exponential moment in the second factor is finite by Lemma~\ref{lem:OU-moments}. Finally,
	$
	\mathbb{E}\left[\int_t^T|u_s|^p\,ds\right]<\infty
	$
	for $p\le p_\star$ by Assumption~2.3 and H\"older's inequality on the finite interval $[t,T]$. This proves \eqref{eq:jump-Lp}, and \eqref{eq:jump-L2} is the special case $p=2$. Standard stochastic integration with respect to compensated Poisson random measures then yields the claimed local-martingale and square-integrability properties \citep{Applebaum2009,Situ2005,OksendalSulem2019}.
\end{proof}

\subsection*{Proof of Theorem \ref{thm:strong-wellposedness}}

\begin{proof}
	The equation for $Y$ is autonomous and has the unique strong solution \eqref{eq:Y-explicit}. It therefore suffices to establish existence and pathwise uniqueness for $X$ conditional on the already constructed process $Y$.
	Fix $n\in\mathbb{N}$ and define the stopping time
	$
	\tau_n
	:=
	\inf\left\{
	s\ge t:
	|Y_s|+|X_s|\ge n
	\right\}\wedge T.
	$
	On $[t,\tau_n]$, the random coefficients
	$
	(x,s)\mapsto b(x,Y_s,u_s),
	$ and $
	(x,s)\mapsto \sigma_X(x,Y_s,u_s)
	$
	are globally Lipschitz in $x$ uniformly in $(s,\omega)$, and the jump amplitude is Lipschitz in $Y_s$ with $\nu$-square-integrable envelope $\rho$. Standard Picard iteration for stochastic differential equations with jumps therefore yields a unique strong solution up to $\tau_n$; cf.\ \citet{Situ2005} and \citet{Applebaum2009}. To extend the solution globally, it remains to prove the a priori estimate \eqref{eq:state-moment} independently of $n$. For $s\in[t,T]$,
	\begin{align*}
		X_{s\wedge\tau_n}
		&=
		x
		+\int_t^{s\wedge\tau_n} b(X_{r-},Y_{r-},u_r)\,dr+
		\int_t^{s\wedge\tau_n}
		\sigma_X(X_{r-},Y_{r-},u_r)\,dW_r^X-
		\int_t^{s\wedge\tau_n}\int_E
		\ell(Y_{r-},z,u_r)\,\widetilde N(dr,dz).
	\end{align*}
	For $p\ge2$, by invoking
	$
	|a_1+\cdots+a_4|^p\le C_p\sum_{j=1}^4|a_j|^p
	$
	we obtain
	\begin{align}
		\label{eq:pre-bdg}
		\mathbb{E}\left[\sup_{t\le r\le s}
		|X_{r\wedge\tau_n}|^p\right]
		&\le
		C_p|x|^p
		+
		C_p
		\mathbb{E}
		\left[
		\int_t^{s\wedge\tau_n}
		|b(X_{r-},Y_{r-},u_r)|\,dr
		\right]^p
		\nonumber\\
		&\quad+
		C_p\cdot\mathbb{E}\left[
		\sup_{t\le v\le s}
		\left|
		\int_t^{v\wedge\tau_n}
		\sigma_X(X_{r-},Y_{r-},u_r)\,dW_r^X
		\right|^p\right]
		\nonumber\\
		&\quad+
		C_p\cdot \mathbb{E}\left[
		\sup_{t\le v\le s}
		\left|
		\int_t^{v\wedge\tau_n}\int_E
		\ell(Y_{r-},z,u_r)\,\widetilde N(dr,dz)
		\right|^p\right].
	\end{align}
	H\"older's inequality and \eqref{eq:growth-b-sigma} yield
	\begin{align}
		\mathbb{E}
		\left[
		\int_t^{s\wedge\tau_n}|b_r|\,dr
		\right]^p
		\le
		C_p\cdot
		\int_t^s
		\mathbb{E}
		\left[
		1+\sup_{t\le v\le r}|X_{v\wedge\tau_n}|^p
		+|Y_r|^p
		+|u_r|^p
		\right]dr.
	\end{align}
	The Brownian BDG inequality gives
	\begin{align}
		\mathbb{E}\left[
		\sup_{t\le v\le s}
		\left|
		\int_t^{v\wedge\tau_n}\sigma_X\,dW_r^X
		\right|^p\right]
		&\le
		C_p\,
		\mathbb{E}
		\left[
		\int_t^{s\wedge\tau_n}|\sigma_X|^2\,dr
		\right]^{p/2}\notag\\
		&\le
		C_p
		\int_t^s
		\mathbb{E}
		\left[
		1+\sup_{t\le v\le r}|X_{v\wedge\tau_n}|^p
		+|Y_r|^p
		+|u_r|^p
		\right]dr.
	\end{align}
	For the compensated jump martingale, the Bichteler--Jacod inequality yields
	\begin{align}
		&\mathbb{E}\left[
		\sup_{t\le v\le s}
		\left|
		\int_t^{v\wedge\tau_n}\int_E
		\ell(Y_{r-},z,u_r)\,\widetilde N(dr,dz)
		\right|^p\right]
		\nonumber\\
		&\qquad\le
		C_p\,
		\mathbb{E}
		\left(
		\int_t^{s\wedge\tau_n}\int_E
		|\ell(Y_{r-},z,u_r)|^2\,\nu(dz)\,dr
		\right)^{p/2}
		+
		C_p\,\cdot
		\mathbb{E}\left[
		\int_t^{s\wedge\tau_n}\int_E
		|\ell(Y_{r-},z,u_r)|^p\,\nu(dz)\,dr\right].
		\label{eq:BJ}
	\end{align}
	Using \eqref{eq:rho-mom}, H\"older's inequality, and Lemma~\ref{lem:OU-moments}, the right-hand side of \eqref{eq:BJ} is bounded by
	\begin{equation}
		C_p
		\int_t^s
		\left[
		1+
		\mathbb{E}e^{p|Y_r|}
		+
		\mathbb{E}|u_r|^p
		\right]dr.
	\end{equation}
	Substituting these bounds into \eqref{eq:pre-bdg} gives
	\begin{align}
		\mathbb{E}\left[\sup_{t\le r\le s}|X_{r\wedge\tau_n}|^p\right]
		&\le
		C_p
		\left[
		1+|x|^p+|y|^p+e^{C_p|y|}
		+\mathbb{E}\int_t^T|u_r|^p\,dr
		\right]
		+
		C_p\int_t^s
		\mathbb{E}\left[\sup_{t\le v\le r}|X_{v\wedge\tau_n}|^p\,\right]dr.
	\end{align}
	Gronwall's inequality therefore yields
	\begin{align}
		\label{eq:uniform-stopped}
		\sup_{n\ge1}
		\mathbb{E}\left[
		\sup_{t\le s\le T}|X_{s\wedge\tau_n}|^p\right]
		\le
		C_p
		\left[
		1+|x|^p+|y|^p+e^{C_p|y|}
		+\mathbb{E}\int_t^T|u_s|^p\,ds
		\right].
	\end{align}
	By Markov's inequality,
	\[
	\mathbb{P}(\tau_n<T)
	\le
	\mathbb{P}\left(
	\sup_{t\le s\le T}|X_{s\wedge\tau_n}|\ge n/2
	\right)
	+
	\mathbb{P}\left(
	\sup_{t\le s\le T}|Y_s|\ge n/2
	\right),
	\]
	and the right-hand side converges to zero as $n\to\infty$ by \eqref{eq:uniform-stopped} and Lemma~\ref{lem:OU-moments}. Hence, $\tau_n\uparrow T$ almost surely and the local solution extends to all of $[t,T]$. Moreover, Fatou's lemma applied to \eqref{eq:uniform-stopped} gives \eqref{eq:state-moment}.
	For pathwise uniqueness, let $X$ and $\widehat X$ be two solutions driven by the same $(W,N)$ with the same initial state and control. Since the jump amplitude in \eqref{eq:X-general} does not depend on $X$, the jump terms cancel in $X-\widehat X$. Using \eqref{eq:lipschitz-b-sigma}, the BDG inequality, and Gronwall's inequality gives
	$
	\mathbb{E}\left[\sup_{t\le s\le T}|X_s-\widehat X_s|^2\right]=0.
	$
	Therefore, $X$ and $\widehat X$ are indistinguishable. Strong existence and pathwise uniqueness follow.
\end{proof}

\subsection*{Proof of Proposition \ref{prop:initial-stability}}

\begin{proof}
	From the explicit representation \eqref{eq:Y-explicit},
	$
	Y_s^{t,y}-Y_s^{t,y'}
	=
	(y-y')e^{-\kappa(s-t)},
	$
	which proves \eqref{eq:Y-stability}. Define
	$
	\Delta X_s
	:=
	X_s^{t,x,y;u}-X_s^{t,x',y';u},$ and 
	$
	\Delta Y_s
	:=
	Y_s^{t,y}-Y_s^{t,y'}.
	$
	Then
	\begin{align*}
		\Delta X_s
		&=
		x-x'
		+
		\int_t^s
		\Delta b_r\,dr
		+
		\int_t^s
		\Delta \sigma_r\,dW_r^X-
		\int_t^s\int_E
		\Delta\ell_r(z)\,\widetilde N(dr,dz),
	\end{align*}
	where
	$
	|\Delta b_r|+|\Delta\sigma_r|
	\le
	L(|\Delta X_r|+|\Delta Y_r|)
	$
	and, by \eqref{eq:ell-lip},
	$
	|\Delta\ell_r(z)|
	\le
	L|\Delta Y_r|\rho(z).$
	Applying H\"older's inequality, the Brownian BDG inequality, and the Bichteler-Jacod inequality as in the proof of Theorem~\ref{thm:strong-wellposedness}, we obtain
	\begin{align*}
		\mathbb{E}\left[\sup_{t\le v\le s}|\Delta X_v|^p\right]
		&\le
		C_p|x-x'|^p
		+
		C_p
		\int_t^s
		\mathbb{E}\left[\sup_{t\le q\le r}|\Delta X_q|^p\,\right]dr+
		C_p
		\int_t^s
		\mathbb{E}\left[|\Delta Y_r|^p\,\right]dr.
	\end{align*}
	Using \eqref{eq:Y-stability},
	$
	\int_t^s\mathbb{E}\big[|\Delta Y_r|^p\big]dr
	\le
	T|y-y'|^p.$
	Gronwall's inequality therefore gives the stronger bound
	\[
	\mathbb{E}\left[\sup_{t\le s\le T}|\Delta X_s|^p\right]
	\le
	C_p\big(|x-x'|^p+|y-y'|^p\big).
	\]
	This estimate already implies \eqref{eq:stability}; the factor $\Xi_p(y,y')\ge1$ is retained because later extensions in which $\ell$ is locally Lipschitz in the level $I=e^Y$ naturally produce this exponential weight. The stated form is therefore stable under the economically important multiplicative specification \eqref{eq:benchmark-loss}.
\end{proof}

\subsection*{Proof of Lemma \ref{lem:control-stability}}

\begin{proof}
	Since the construction-cost factor is uncontrolled, both systems have the same process $Y^{t,y}$.  Set
	$
	\Delta X_s:=X_s^{t,x,y;u}-X_s^{t,x,y;v}.
	$
	Subtracting the two surplus equations gives
	\[
	\Delta X_s
	=
	\int_t^s\Delta b_r\,dr
	+
	\int_t^s\Delta\sigma_r\,dW_r^X
	-
	\int_t^s\int_E\Delta\ell_r(z)\,\widetilde N(dr,dz),
	\]
	where Assumptions~2.4 and 2.7 imply
	$
	|\Delta b_r|+|\Delta\sigma_r|
	\le C\bigl(|\Delta X_r|+|u_r-v_r|\bigr),
	$ and $
	|\Delta\ell_r(z)|
	\le L_u|u_r-v_r|\rho(z).$
	H\"older's inequality and the Brownian BDG inequality therefore yield
	\[
	\mathbb E\sup_{t\le q\le s}
	\left|
	\int_t^q\Delta b_r\,dr
	\right|^p
	+
	\mathbb E\sup_{t\le q\le s}
	\left|
	\int_t^q\Delta\sigma_r\,dW_r^X
	\right|^p
	\le
	C_p\int_t^s
	\mathbb E\!\left[
	\sup_{t\le a\le r}|\Delta X_a|^p+|u_r-v_r|^p
	\right]dr.
	\]
	For the jump martingale, the Bichteler-Jacod inequality and Assumption~2.4 yield
	\begin{align*}
		&\mathbb E\left[\sup_{t\le q\le s}
		\left|
		\int_t^q\int_E\Delta\ell_r(z)\,\widetilde N(dr,dz)
		\right|^p\right]\\
		&\quad\le
		C_p\cdot\mathbb E
		\left(
		\int_t^s\int_E|\Delta\ell_r(z)|^2\nu(dz)\,dr
		\right)^{p/2}
		+
		C_p\cdot \mathbb E\left[
		\int_t^s\int_E|\Delta\ell_r(z)|^p\nu(dz)\,dr\right]\\
		&\quad\le
		C_p\cdot \mathbb E\left[\int_t^s|u_r-v_r|^p\,dr\right],
	\end{align*}
	where finiteness of the $\rho^2$- and $\rho^p$-moments follows from Assumption~2.4.  Combining the preceding estimates,
	\[
	\mathbb E\left[\sup_{t\le q\le s}|\Delta X_q|^p\right]
	\le
	C_p\cdot \int_t^s
	\mathbb E\left[\sup_{t\le a\le r}|\Delta X_a|^p\,\right]dr
	+
	C_p\cdot \mathbb E\left[\int_t^s|u_r-v_r|^p\,dr\right].
	\]
	Gronwall's inequality proves \eqref{eq:control-stability}.  The final assertion is immediate from the definition of the $\mathcal S^p$ norm.
\end{proof}

\subsection*{Proof of Proposition \ref{prop:J-continuity}}

\begin{proof}
	Lemma~\ref{lem:control-stability} guarantees convergence of $X^{u^n}$ to $X^u$ in $\mathcal S^{p_\star}$, hence in probability uniformly on $[t,T]$.  Since $Y$ is common to all controls, the vector
	$
	(s,X_s^{u^n},Y_s,u_s^n)
	$
	converges to $(s,X_s^u,Y_s,u_s)$ in measure on the product space
	$([t,T]\times\Omega,ds\otimes d\mathbb P)$.  After passage to a subsequence, convergence holds $ds\otimes d\mathbb P$-a.e.  Continuity of $f$ then gives pointwise convergence of the running costs along that subsequence. It remains to justify passage under the expectation.  Choose $\varepsilon>0$ such that
	$q(1+\varepsilon)<p_\star$, which is possible by Assumption~2.6.  The growth bound on $f$, Theorem~3, Lemma~1, and the assumed uniform $p_\star$-bound on the controls imply
	$
	\sup_n
	\mathbb E\left[\int_t^T
	|f(s,X_s^{u^n},Y_s,u_s^n)|^{1+\varepsilon}\,ds\right]
	<\infty.
	$
	Thus the running-cost family is uniformly integrable.  The same argument, without the control term, shows uniform integrability of
	$\{g(X_T^{u^n},Y_T)\}_{n\ge1}$.  Vitali's theorem yields convergence of both the running and terminal terms in $L^1$.  Since every subsequence admits a further subsequence with the same limit, the full sequence converges, proving the claim.
\end{proof}
\subsection*{Proof of Proposition \ref{prop:flow}}

\begin{proof}
	Integrate the original state equations from $\tau$ to $s$.  For the cost factor,
	\[
	Y_s
	=
	Y_\tau+\int_\tau^s\kappa(\vartheta-Y_r)\,dr
	+\sigma_I(W_s^I-W_\tau^I).
	\]
	For the surplus,
	\begin{align*}
		X_s
		&=
		X_\tau
		+\int_\tau^s b(X_{r-},Y_{r-},u_r)\,dr
		+\int_\tau^s\sigma_X(X_{r-},Y_{r-},u_r)\,dW_r^X-
		\int_\tau^s\int_E
		\ell(Y_{r-},z,u_r)\,\widetilde N(dr,dz).
	\end{align*}
	After the time change $r=\tau+a$, these are exactly the integral equations defining the system restarted from $(X_\tau,Y_\tau)$ and driven by $(W^\tau,N^\tau)$.  Theorem~3 gives pathwise uniqueness.  Hence the original post-$\tau$ trajectory and the restarted trajectory are indistinguishable, proving \eqref{eq:flow}.  If the original control is replaced by $u\otimes_\tau v$, its restriction to $[\tau,T]$ equals $v$, and the same argument applies.
\end{proof}

\subsection*{Proof of Proposition \ref{prop:conditional-restart}}
\begin{proof}
	First suppose that $\tau$ takes finitely many deterministic values and that $\Phi$ is a bounded cylinder functional.  On each event $\{\tau=t_j\}$, Proposition~\ref{prop:flow} identifies the future state with the solution driven by increments of $(W,N)$ after $t_j$.  These increments are independent of $\mathcal F_{t_j}$ and have the same law as fresh Brownian and Poisson noises.  Conditioning therefore gives \eqref{eq:conditional-restart}. For a general stopping time, choose stopping times $\tau_n$ taking values in the dyadic grid such that
	$\tau_n\downarrow\tau$.  C\`adl\`ag paths, the state stability estimates, and right-continuity of the filtration permit passage from $\tau_n$ to $\tau$ for bounded continuous cylinder functionals.  A monotone-class argument extends the identity to bounded Borel $\Phi$.  Finally, truncation and uniform integrability extend the result to every integrable $\Phi$ covered by the statement.  This is the controlled jump-diffusion version of the standard restart argument underlying the DPP \citep{FlemingSoner2006,Pham2009}, and the measurable-selection formulations in \citet{ElKarouiTan2013} and \citet{Zitkovic2014}.
\end{proof}

\subsection*{Proof of Theorem \ref{thm:feller}}

\begin{proof}
	Substituting $u_s=\alpha(s,X_{s-},Y_{s-})$ into the coefficients produces a time-inhomogeneous jump SDE whose coefficients are Borel in time and globally Lipschitz in the state, with linear growth.  Indeed, Assumptions~2.4, 2.7, and 2.8 imply
	\[
	|b(x,y,\alpha(s,x,y))-b(x',y',\alpha(s,x',y'))|
	\le C_\alpha(|x-x'|+|y-y'|),
	\]
	and the same estimate holds for $\sigma_X$; the jump coefficient satisfies the corresponding $L^2(\nu)\cap L^p(\nu)$ estimate.  Hence Theorem~3 applies to the feedback system, and pathwise uniqueness holds from every deterministic initial state. The flow property of Proposition~\ref{prop:flow}, now with a feedback control depending only on the current time and state, implies the Markov property at deterministic times.  At a stopping time $\tau$, the continuation rule is
	$
	s\mapsto\alpha(s,X_{s-},Y_{s-}),
	$
	which depends on the past only through the restarted state.  Proposition~\ref{prop:conditional-restart}, together with pathwise uniqueness, therefore yields
	$
	\mathbb E_{t,x,y}\!\left[
	\varphi(Z_{\tau+r}^\alpha)\mid\mathcal F_\tau
	\right]
	=
	P_{\tau,\tau+r}^\alpha\varphi(Z_\tau^\alpha)
	$ a.s. for bounded Borel $\varphi$, proving the strong Markov property. For the Feller property, let $(x_n,y_n)\to(x,y)$.  Repeating the proof of Proposition~5 with the feedback controls generated by the two state trajectories and using the Lipschitz property of $\alpha$ gives, for every $p\in[2,p_\star]$,
	\[
	\mathbb E\left[
	\sup_{t\le r\le s}
	|Z_r^{t,x_n,y_n;\alpha}-Z_r^{t,x,y;\alpha}|^p
	\right]\longrightarrow0.
	\]
	Hence $Z_s^{t,x_n,y_n;\alpha}\to Z_s^{t,x,y;\alpha}$ in probability.  If $\varphi\in C_b(\mathbb R^2)$, boundedness and convergence in probability imply
	$
	P_{t,s}^\alpha\varphi(x_n,y_n)
	\rightarrow
	P_{t,s}^\alpha\varphi(x,y),
	$
	so $P_{t,s}^\alpha\varphi$ is continuous.  Boundedness is immediate from
	$\|P_{t,s}^\alpha\varphi\|_\infty\le\|\varphi\|_\infty$.
\end{proof}

\subsection*{Proof of Proposition \ref{prop:UI}}

\begin{proof}
	Since $q<p_\star$, choose $\eta>0$ so small that
	$q(1+\eta)<p_\star$.  Apply Theorem~3 with
	$p=q(1+\eta)$.  Because $K$ is compact, $|x|$, $|y|$, and
	$\exp(C|y|)$ are uniformly bounded on $K$.  Moreover, by H\"older's inequality on the finite measure interval $[t,T]$,
	\[
	\mathbb E\left[\int_t^T|u_s|^{q(1+\eta)}ds\right]
	\le
	(T-t)^{1-q(1+\eta)/p_\star}
	\left(
	\mathbb E\int_t^T|u_s|^{p_\star}ds
	\right)^{q(1+\eta)/p_\star},
	\]
	which is uniformly bounded over $\mathcal U_t(R)$.  Lemma~1 supplies the exponential moment of $Y$.  This proves \eqref{eq:UI-bound}.
	
	By Assumption~2.6,
	\[
	|g(X_T,Y_T)|^{1+\eta}
	\le
	C_\eta\left(
	1+|X_T|^{q(1+\eta)}
	+e^{q(1+\eta)|Y_T|}
	\right),
	\]
	whose expectations are uniformly bounded by \eqref{eq:UI-bound}.  Hence the terminal-cost family is uniformly integrable by the de la Vall\'ee-Poussin criterion.  Similarly, Jensen's inequality yields
	\[
	\left(
	\int_t^T|f_s|\,ds
	\right)^{1+\eta}
	\le
	(T-t)^\eta\int_t^T|f_s|^{1+\eta}ds,
	\]
	and the growth condition on $f$ together with \eqref{eq:UI-bound} yields a uniform bound on the expectation of the right-hand side.  The running-cost family is therefore uniformly integrable as well.
\end{proof}

\subsection*{Proof of Proposition \ref{prop:canonical-verification}}

\begin{proof} Since $U$ is compact metric, the relaxed-control space equipped with
		the stable topology is a standard Borel space. Hence the canonical
		product space $\Omega^\circ$ is standard Borel, and
		$\mathfrak P(\Omega^\circ)$ endowed with the weak topology is likewise
		standard Borel. Regular conditional probabilities and analytic-set
		arguments may therefore be used on the canonical space.  Let $Z=(X,Y)$ denote the canonical state coordinate and let $q$ denote the relaxed-control coordinate. For $\varphi\in C_c^2(\mathbb R^2)$, define \[ \begin{aligned} \mathcal A^a\varphi(x,y) ={}& b(x,y,a)\varphi_x(x,y) + \kappa(\vartheta-y)\varphi_y(x,y) + \frac12\sigma_X^2(x,y,a)\varphi_{xx}(x,y) + \frac12\sigma_I^2\varphi_{yy}(x,y) \\ &+ \int_E \Big[ \varphi(x-\ell(y,z,a),y) -\varphi(x,y) +\ell(y,z,a)\varphi_x(x,y) \Big]\nu(dz). \end{aligned} \] For a relaxed control $q$, put $ \mathcal A^{q_s}\varphi(x,y) = \int_U \mathcal A^a\varphi(x,y)\,q_s(da). $ The canonical martingale problem requires \[ M_s^\varphi := \varphi(Z_s)-\varphi(Z_t) - \int_t^s \mathcal A^{q_r}\varphi(Z_{r-})\,dr \] to be a local martingale, together with the initial-state and admissibility conditions. We first prove analyticity. Choose a countable family $\{\varphi_k\}_{k\ge1}\subset C_c^2(\mathbb R^2)$ dense in the local $C^2$ topology, rational times $r<s$, a countable generating algebra of bounded cylinder functions measurable at time $r$, and a sequence of localization stopping times. The martingale problem is equivalent to the resulting countable collection of identities \[ E^P \left[ \bigl( M_{s\wedge\tau_n}^{\varphi_k} - M_{r\wedge\tau_n}^{\varphi_k} \bigr)H \right] =0. \] Assumption~2.9(ii)-(iii), together with the growth conditions of Assumptions~2.4 and~2.6, implies that the canonical integrands are Borel measurable. The initial-state constraint and the required moment conditions are likewise Borel after truncation. Consequently the set of probability measures satisfying the localized martingale identities and admissibility conditions is Borel in the canonical probability space. Its projection onto $(t,x,y,P)$ is therefore analytic. This proves (a). For (b), fix $P\in\mathfrak P(t,x,y)$ and an $[t,T]$-valued stopping time $\tau$. For every localized martingale, optional sampling gives \[ E^P \left[ M_{s\wedge\tau_n}^{\varphi_k} - M_{\tau\wedge\tau_n}^{\varphi_k} \mid \mathcal F_\tau \right] =0. \] Disintegrating $P$ with respect to $\mathcal F_\tau$ shows that, for $P$-a.e.\ $\omega$, the conditional law satisfies the same localized martingale identities after $\tau(\omega)$ and starts from $ \bigl( X_\tau(\omega),Y_\tau(\omega) \bigr). $ The conditional admissibility requirement follows from conditional Tonelli's theorem. Hence, $ P^{\tau,\omega} \in \mathfrak P\bigl( \tau(\omega),X_\tau(\omega),Y_\tau(\omega) \bigr) $ for $P$-a.e.\ $\omega$. For (c), let $Q_\omega$ be a universally measurable admissible continuation kernel. Under $P\otimes_\tau Q$, every localized martingale identity holds before $\tau$ because it holds under $P$ and after $\tau$ because it holds under $Q_\omega$ for $P$-a.e.\ $\omega$. Splitting an increment at $\tau$ and applying the tower property therefore yields the martingale identity on the full time interval. The moment admissibility condition is preserved by the corresponding decomposition at $\tau$. Theefore, $ P\otimes_\tau Q\in\mathfrak P(t,x,y). $ This proves (b)-(c). The argument is the model-specific verification of the abstract conditioning and concatenation structure used in \citet{ElKarouiTan2013}, \citet{Zitkovic2014}, and \citet{FayvisovichZitkovic2021}. \end{proof}
	
\subsection{Proof of Proposition \ref{prop:chattering}}

\begin{proof} Since, $U$ is compact metric, the space of relaxed control measures with time marginal equal to Lebesgue measure is compact under the stable topology. The chattering lemma yields predictable $U$-valued controls $u^n$ such that, for every bounded measurable $\psi:[t,T]\times\Omega\times U\to\mathbb R$ that is continuous in the control variable, 
		\begin{equation} \label{eq:stable-chattering} \int_t^T \psi(s,\omega,u_s^n)\,ds \longrightarrow \int_t^T\int_U \psi(s,\omega,a)\,q_s(da)\,ds \end{equation} 
		in probability, after passage to a subsequence when necessary. We apply this approximation simultaneously to the controlled characteristics of the martingale problem. By Assumption~2.9(ii)-(iii), the maps $ a\mapsto b(x,y,a), \ a\mapsto\sigma_X^2(x,y,a), $ and \[ a\mapsto \int_E \Big[ \varphi(x-\ell(y,z,a),y) -\varphi(x,y) +\ell(y,z,a)\varphi_x(x,y) \Big]\nu(dz) \] are continuous for every $\varphi\in C_c^2(\mathbb R^2)$. For the jump term, continuity follows from Taylor's theorem on the small-jump region, Assumption~2.9(iii), and dominated convergence on the complementary region. Hence \eqref{eq:stable-chattering} implies convergence of the integrated characteristics associated with every test function in the countable martingale-problem core. The moment estimates of Theorem~\ref{thm:strong-wellposedness} are uniform over strict controls because $U$ is compact. Together with the increment estimates used in the proof of Proposition~\ref{prop:UI}, they imply tightness of $\{(X^n,Y^n)\}_{n\ge1}$ in $D([t,T];\mathbb R^2)$. ```latex
		The compactness of $U$, the continuity hypotheses in
		Assumption~2.9(ii)--(iii), and the stable convergence of the occupation
		measures place the present controlled martingale problem within the
		classical compactification framework for relaxed stochastic controls.
		The chattering approximation theorem of
		\citet{ElKarouiNguyenJeanblanc1987} therefore implies that the strict
		controlled laws associated with $(u^n)$ converge, along the
		approximating sequence, to the relaxed controlled law $P$ in the
		canonical weak topology. In particular,
		\begin{equation}
			\label{eq:chattering-state-convergence}
			(X^n,Y^n)
			\Longrightarrow
			(X,Y)
			\qquad
			\text{in }
			D([t,T];\mathbb R^2).
		\end{equation}
		
		We stress that this is a weak convergence statement for the
		martingale formulation. It does not identify the relaxed dynamics with
		the strict SDE of Theorem~\ref{thm:strong-wellposedness}; rather, the
		relaxed generator averages the controlled local and jump
		characteristics with respect to $q_s$. This is precisely the role of
		the compactification procedure.
		
		It remains to prove convergence of the objective functionals.
		Assumption~2.9(ii), stable convergence of the controls, and
		\eqref{eq:chattering-state-convergence} imply convergence in
		probability of the truncated running costs. The uniform moment
		estimates established in Proposition~\ref{prop:UI}, together with the
		polynomial growth bound of Assumption~2.6, give uniform integrability.
		Vitali's theorem therefore yields
		\[
		\begin{aligned}
			&E\int_t^T e^{-\delta(s-t)}
			f(s,X_s^n,Y_s^n,u_s^n)\,ds
			\rightarrow
			E^P\int_t^T e^{-\delta(s-t)}
			\int_U
			f(s,X_s,Y_s,a)\,q_s(da)\,ds.
		\end{aligned}
		\]
		The same weak convergence, continuity of $g$, and uniform
		integrability imply
		$
		E[g(X_T^n,Y_T^n)]
		\longrightarrow
		E^P[g(X_T,Y_T)].
		$
		Consequently,
		\begin{equation}
			\label{eq:chattering-cost-convergence}
			J(t,x,y;u^n)
			\longrightarrow
			J(t,x,y;P).
		\end{equation}
\end{proof}

\subsection {Proof of Theorem \ref{thm:measurable-selection}}

	\begin{proof}
		Set $\mathsf E:=[0,T]\times\mathbb R^2$,
		$\mathsf P:=\mathfrak P(\Omega^\circ)$, and
		$\mathcal G:=\operatorname{Gr}(\mathfrak P)\subset
		\mathsf E\times\mathsf P$.  Write $\xi=(t,x,y)\in\mathsf E$.  By
		Proposition~\ref{prop:canonical-verification},
		$\mathcal G\in\mathcal A(\mathsf E\times\mathsf P)$, where
		$\mathcal A$ denotes the class of analytic sets.
		For $(\xi,P)\in\mathcal G$, define
		$\Gamma(\xi,P):=J(t,x,y;P)$, i.e.,
		$\Gamma(\xi,P)
		=
		E^P[
		\int_t^T e^{-\delta(s-t)}
		\int_U f(s,X_s,Y_s,a)\,q_s(da)\,ds
		+
		e^{-\delta(T-t)}g(X_T,Y_T)]$.
		For $N\in\mathbb N$, let
		$f_N:=(-N)\vee(f\wedge N)$ and
		$g_N:=(-N)\vee(g\wedge N)$, and define
		$\Gamma_N$ by replacing $(f,g)$ with $(f_N,g_N)$ in $\Gamma$.
		The canonical evaluation maps and the relaxed-control coordinate are
		Borel, hence $\Gamma_N$ is Borel on $\mathcal G$.  Assumption~2.6 and
		Proposition~\ref{prop:UI} imply
		$\Gamma_N(\xi,P)\to\Gamma(\xi,P)$ on $\mathcal G$, with uniformly
		integrable truncation error on bounded-moment subsets.  Therefore
		$\Gamma\in\operatorname{LSA}(\mathcal G)$.
		
		Extend $\Gamma$ to $\mathsf E\times\mathsf P$ by
		$\bar\Gamma(\xi,P):=\Gamma(\xi,P)$ on $\mathcal G$ and
		$\bar\Gamma(\xi,P):=+\infty$ otherwise.  Since $\mathcal G$ is
		analytic, $\bar\Gamma\in\operatorname{LSA}(\mathsf E\times\mathsf P)$.
		Consequently,
		$V^{\mathrm{rel}}(\xi)
		=
		\inf_{P\in\mathsf P}\bar\Gamma(\xi,P)
		=
		\inf_{P\in\mathfrak P(\xi)}\Gamma(\xi,P)$.
		For every $c\in\mathbb R$,
		$\{\xi:V^{\mathrm{rel}}(\xi)<c\}
		=
		\operatorname{proj}_{\mathsf E}
		\{(\xi,P):\bar\Gamma(\xi,P)<c\}$.
		The set on the right is analytic, and analytic sets are stable under
		projection; hence
		$\{\xi:V^{\mathrm{rel}}(\xi)<c\}\in\mathcal A(\mathsf E)$.
		Thus $V^{\mathrm{rel}}\in\operatorname{LSA}(\mathsf E)$.
		Fix $\varepsilon>0$ and set
		$\mathcal G_\varepsilon
		:=
		\{(\xi,P)\in\mathcal G:
		\Gamma(\xi,P)\le
		V^{\mathrm{rel}}(\xi)+\varepsilon\}$.
		Using rational approximation,
		$\mathcal G_\varepsilon
		=
		\bigcap_{n\ge1}\bigcup_{r\in\mathbb Q}
		\bigl(
		\{\xi:V^{\mathrm{rel}}(\xi)<r\}\times\mathsf P
		\bigr)
		\cap
		\{(\xi,P)\in\mathcal G:
		\Gamma(\xi,P)<r+\varepsilon+n^{-1}\}$.
		Hence
		$\mathcal G_\varepsilon\in
		\mathcal A(\mathsf E\times\mathsf P)$.
		Let
		$D_V:=\{\xi\in\mathsf E:
		V^{\mathrm{rel}}(\xi)<\infty\}$.
		For every $\xi\in D_V$, the section
		$(\mathcal G_\varepsilon)_\xi
		:=
		\{P\in\mathsf P:(\xi,P)\in\mathcal G_\varepsilon\}$
		is nonempty by definition of the infimum.  The
		Jankov-von Neumann selection theorem therefore yields a universally
		measurable selector
		$\xi\mapsto P^\varepsilon_\xi$ satisfying
		$(\xi,P^\varepsilon_\xi)\in\mathcal G_\varepsilon$ for all
		$\xi\in D_V$.  Hence
		$P^\varepsilon_\xi\in\mathfrak P(\xi)$ and
		$\Gamma(\xi,P^\varepsilon_\xi)
		\le
		V^{\mathrm{rel}}(\xi)+\varepsilon$.
		Returning to $\xi=(t,x,y)$ gives
		$J(t,x,y;P^\varepsilon_{t,x,y})
		\le
		V^{\mathrm{rel}}(t,x,y)+\varepsilon$,
		which is \eqref{eq:epsilon-selector}.
	\end{proof}

\subsection*{Proof of Theorem \ref{thm:DPP}}

\begin{proof}
	By Proposition~\ref{prop:strict-relaxed-equivalence} the strict and relaxed values coincide, so it is
	enough to prove the identity in the canonical formulation, where measurable continuation is available. Fix $P\in\mathfrak P(t,x,y)$. By
	Lemma~\ref{lem:discounted-decomposition} and the definition of the value function, the conditional continuation cost at $\tau$ is bounded below by
	$V(\tau,X_\tau,Y_\tau)$. Hence,
	\[
	J(t,x,y;P)
	\ge
	E^P\Bigg[
	\int_t^\tau e^{-\delta(s-t)}f_s\,ds
	+
	e^{-\delta(\tau-t)}
	V(\tau,X_\tau,Y_\tau)
	\Bigg].
	\]
	Taking the infimum over $P$ yields the ``$\ge$'' inequality in \eqref{eq:DPP}. This is the inequality already anticipated by the pre-DPP decomposition of Proposition~\ref{prop:pre-dpp}. For the reverse inequality, fix $\varepsilon>0$. The universally
	measurable selection result of Theorem~\ref{thm:measurable-selection} provides a continuation kernel
	$
	(s,\xi,\eta)
	\mapsto
	Q^\varepsilon_{s,\xi,\eta}
	\in\mathfrak P(s,\xi,\eta)
	$
	such that
	$
	J(s,\xi,\eta;Q^\varepsilon_{s,\xi,\eta})
	\le
	V(s,\xi,\eta)+\varepsilon.
	$
	Let $P$ be any admissible law up to $\tau$. By
	Proposition~\ref{prop:canonical-verification}, the concatenation
	$
	P\otimes_\tau
	Q^\varepsilon_{\tau,X_\tau,Y_\tau}
	$
	is admissible from $(t,x,y)$. Applying
	Lemma~\ref{lem:discounted-decomposition} to this concatenated law gives
	\[
	\begin{aligned}
		V(t,x,y)
		\le
		E^P\Bigg[
		&
		\int_t^\tau e^{-\delta(s-t)}f_s\,ds+
		e^{-\delta(\tau-t)}
		\bigl(V(\tau,X_\tau,Y_\tau)+\varepsilon\bigr)
		\Bigg].
	\end{aligned}
	\]
	Since $e^{-\delta(\tau-t)}\le e^{|\delta|T}$, the error is bounded by
	$C\varepsilon$. Taking the infimum over $P$ and then letting
	$\varepsilon\downarrow0$ proves the opposite inequality. The strict
	form follows again from Proposition~\ref{prop:strict-relaxed-equivalence}.
\end{proof}

\subsection*{Proof of Proposition \ref{prop:value-continuity}}

\begin{proof}
	Membership in $\mathcal G_q$ follows from
	Corollary~\ref{cor:value-growth}. Fix a compact
	$K\subset\mathbb R^2$. Coercivity of $f$, together with the growth
	bound for $V$, permits restriction, up to an arbitrarily small error,
	to controls whose $p_\star$-moments are uniformly bounded on $K$.
	For such controls, Proposition~\ref{prop:J-continuity} gives local uniform stability of the
	state with respect to the initial condition. Local uniform continuity
	of $f$ and $g$, followed by truncation outside a large compact set and
	the uniform-integrability estimate of Proposition~\ref{prop:UI}, therefore yields
	$
	\sup_{u}
	|J(t,x,y;u)-J(t,x',y';u)|
	\rightarrow0
	$
	as $(x',y')\to(x,y)$, locally uniformly in $(x,y)\in K$. Taking infima preserves this modulus. For time continuity, apply Corollary~\ref{cor:deterministic-DPP} over an interval of length $h$.
	The running-cost contribution tends to zero uniformly on compact sets by uniform integrability, while the stochastic continuity of $(X,Y)$ and the already established spatial continuity of $V$ imply
	$
	E[V(t+h,X_{t+h},Y_{t+h})]
	\rightarrow
	V(t,x,y).$
	This proves continuity for $t<T$. Finally, use the DPP with $s=T$. The running cost over $[t,T]$
	vanishes as $t\uparrow T$, while
	$(X_T,Y_T)\to(x,y)$ in probability uniformly over the truncated
	near-optimal control family. Uniform integrability and continuity of
	$g$ then give the terminal limit.
\end{proof}

\subsection*{Proof of Theorem \ref{thm:value-viscosity}}

\begin{proof}
	We prove the two inequalities separately.
	
	\emph{Subsolution property.}
	Let $\phi\in C^{1,2,2}$ and suppose that $V-\phi$ has a strict local
	maximum equal to zero at $(t_0,x_0,y_0)$, with $t_0<T$. Fix
	$a\in U$ and choose $R,h>0$ small enough that
	$V\le\phi$ on the relevant local cylinder. Let $\tau_R$ be as in
	Lemma~\ref{lem:localized-Dynkin}, under the constant control $a$.
	The DPP gives
	\[
	V(t_0,x_0,y_0)
	\le
	E\Bigg[
	\int_{t_0}^{\tau_R}
	e^{-\delta(s-t_0)}
	f(s,X_s,Y_s,a)\,ds
	+
	e^{-\delta(\tau_R-t_0)}
	V(\tau_R,X_{\tau_R},Y_{\tau_R})
	\Bigg].
	\]
	For jumps remaining inside the local cylinder we may replace $V$ by
	$\phi$. For jumps leaving it, the large-jump contribution is retained
	on $V$; this is precisely the split appearing in
	Definition~\ref{def:viscosity}. Combining the preceding inequality
	with Lemma~\ref{lem:localized-Dynkin}, dividing by
	$E[\tau_R-t_0]$, and sending $h\downarrow0$ and then $R\downarrow0$
	yields
	\[
	\begin{aligned}
		0\le
		&
		f(t_0,x_0,y_0,a)
		+\phi_t(t_0,x_0,y_0)
		+b(x_0,y_0,a)\phi_x(t_0,x_0,y_0)\\
		&+
		\kappa(\vartheta-y_0)\phi_y(t_0,x_0,y_0)
		+\frac12\sigma_X^2(x_0,y_0,a)\phi_{xx}(t_0,x_0,y_0)
		+\frac12\sigma_I^2\phi_{yy}(t_0,x_0,y_0)\\
		&-\delta V(t_0,x_0,y_0)
		+\mathcal I_\varepsilon^a[\phi](t_0,x_0,y_0)
		+\mathcal I^{a,\varepsilon}[V,\phi](t_0,x_0,y_0).
	\end{aligned}
	\]
	Since $a\in U$ was arbitrary, taking the infimum gives the
	subsolution inequality in the sign convention of
	\eqref{eq:nonlocal-HJB}.
	
	\emph{Supersolution property.}
	Suppose now that $V-\phi$ has a strict local minimum equal to zero at $(t_0,x_0,y_0)$. If the supersolution inequality failed, there would exist $\eta>0$ and a sufficiently small neighborhood of
	$(t_0,x_0,y_0)$ such that
	\[
	\inf_{u\in U}
	\left\{
	f+\phi_t+\mathcal A^u\phi-\delta\phi
	\right\}
	\ge\eta,
	\]
	with the nonlocal term interpreted through the small/large jump splitting. By coercivity, controls outside a sufficiently large compact subset of $U$ cannot decrease the Hamiltonian on this neighborhood.
	Choose an $\varepsilon_h$-optimal control in the DPP, with
	$\varepsilon_h=o(h)$, and stop at the first exit from the neighborhood
	or at $t_0+h$. Since $V\ge\phi$ locally, the DPP and the localized Dynkin formula imply
	$
	0
	\ge
	\eta E[\tau_R-t_0]-o(h),
	$
	where the exit contribution is controlled by the growth estimate and uniform integrability. The state moment bounds imply $E[\tau_R-t_0]=h+o(h)$ as $h\downarrow0$. Dividing by $h$ gives a contradiction. The terminal inequalities follow from Proposition~\ref{prop:value-continuity}. Hence $V$ is both a subsolution and a supersolution.
\end{proof}

\subsection*{Proof of Lemma \ref{lem:nonlocal-doubling}}

\begin{proof}
	Because $\bar q>q$ and the exponential factor in $\Gamma$ dominates the $y$-growth defining $\mathcal G_q$, the penalized function tends to $-\infty$ as either spatial variable escapes to infinity. Therefore a maximizer exists.
	The standard doubling-variable argument gives boundedness, for fixed
	$\eta$, of the maximizing sequence and
	$
	|x_\varepsilon-x'_\varepsilon|
	+
	|y_\varepsilon-y'_\varepsilon|
	+
	|t_\varepsilon-s_\varepsilon|
	\rightarrow0.
	$
	The sharper quotient convergence follows from the usual comparison	of the maximum value with its diagonal competitor. For the jump terms, use maximality of
	$\Phi_{\varepsilon,\eta}$ after translating the two surplus variables
	by
	$-\ell(y_\varepsilon,z,u)$ and
	$-\ell(y'_\varepsilon,z,u)$.
	The quadratic penalization produces
	$
	\frac{
		|\ell(y_\varepsilon,z,u)
		-\ell(y'_\varepsilon,z,u)|^2
	}{\varepsilon},
	$
	which is bounded by
	$
	C
	\frac{|y_\varepsilon-y'_\varepsilon|^2}{\varepsilon}
	\rho(z)^2.
	$
	Its $\nu$-integral tends to zero by Assumption~4.1(ii) and the
	doubling estimate. The Lyapunov penalization contributes at most
	$
	C\eta\bigl(\rho(z)^2+\rho(z)^{\bar q}\bigr),
	$
	whose integral is finite by Assumption~4.1(iv). Dominated convergence
	therefore yields
	$
	\mathcal I^u[w]
	-
	\mathcal I^u[\underline w]
	\le
	o_\varepsilon(1)+C\eta,
	$
	uniformly over controls in the compact set furnished by
	Lemma~\ref{lem:compact-control}. This is the required estimate. The argument is the finite-dimensional specialization of the nonlocal Jensen-Ishii method of \citet{BarlesImbert2008}; related semicontinuous maximum principles for integro-PDEs are developed in
	\citet{JakobsenKarlsen2006}.
\end{proof}

\subsection*{Proof of Proposition \ref{prop:strict-supersolution}}

\begin{proof}
	The derivatives of $\Gamma$ satisfy
	$
	|D\Gamma|
	\le C\Gamma,
	$ and $
	|D^2\Gamma|
	\le C\Gamma.
	$
	For the jump part, Taylor's theorem on
	$\{\rho\le1\}$ and the polynomial inequality
	$
	|x-r|^{\bar q}
	\le
	C_{\bar q}
	\bigl(
	1+|x|^{\bar q}+|r|^{\bar q}
	\bigr)
	$
	on $\{\rho>1\}$ imply
	$
	|\mathcal I^u\Gamma(x,y)|
	\le
	C\Gamma(x,y)
	+
	C\bigl(1+|u|^{\bar q}\bigr)
	$
	on compact state sets, with constants controlled by
	Assumption~4.1(iv). The coercivity reduction of
	Lemma~\ref{lem:compact-control} permits restriction to bounded controls when testing the equation locally. Hence,
	$
	\sup_u
	\left|
	\mathcal A^u\Gamma
	\right|
	\le
	C\Gamma
	$
	on the relevant minimizing set. The time derivative of the perturbation contributes
	$
	-\partial_t
	\left(
	\eta e^{\Lambda(T-t)}\Gamma
	\right)
	=
	\Lambda\eta e^{\Lambda(T-t)}\Gamma.
	$
	Choosing $\Lambda$ larger than the constants arising from the local and nonlocal spatial terms, as well as the discount term, leaves a strict positive remainder of order
	$\eta\Gamma$. Stability of viscosity supersolutions under addition of
	a smooth strict barrier completes the proof.
\end{proof}

\subsection*{Proof of Theorem \ref{thm:comparison}}

\begin{proof}
	Assume, to the contrary, that
	$
	M
	:=
	\sup_{[0,T]\times\mathbb R^2}
	(\overline v-\underline v)>0.
	$
	Fix $\eta>0$ sufficiently small and replace $\underline v$ by the strict supersolution $\underline v^\eta$ of
	Proposition~\ref{prop:strict-supersolution}. For $\varepsilon>0$,
	maximize
	\[
	\begin{aligned}
		\Phi_\varepsilon
		(t,s,x,x',y,y')
		:={}&
		\overline v(t,x,y)-\underline v^\eta(s,x',y')-
		\frac{|x-x'|^2+|y-y'|^2}{2\varepsilon}
		-\frac{|t-s|^2}{2\varepsilon}.
	\end{aligned}
	\]
	The strict growth barrier forces the maximizing points into a compact set independent of $\varepsilon$. The terminal ordering implies that,
	for sufficiently small $\varepsilon$, the maximizing times are strictly smaller than $T$. Apply the parabolic nonlocal theorem of sums at the maximizing points.
	There exist jets
	$
	(a_\varepsilon,p_\varepsilon,X_\varepsilon)
	\in
	\overline{\mathcal P}^{2,+}\overline v
	$
	and
	$
	(b_\varepsilon,p'_\varepsilon,Y_\varepsilon)
	\in
	\overline{\mathcal P}^{2,-}\underline v^\eta
	$
	such that
	$
	a_\varepsilon-b_\varepsilon
	=
	\frac{t_\varepsilon-s_\varepsilon}{\varepsilon}
	-
	\frac{t_\varepsilon-s_\varepsilon}{\varepsilon}
	=0,
	$
	and the matrix inequality supplied by the theorem of sums holds.
	
	By Lemma~\ref{lem:compact-control}, both Hamiltonians may be restricted
	to a common compact control set. Fix a control $u_\varepsilon$ that
	is $\varepsilon$-optimal for the supersolution Hamiltonian. Use the
	same control in the subsolution inequality. Subtracting the two
	viscosity inequalities gives
	\[
	\begin{aligned}
		c\eta
		\le{}&
		\bigl[
		f(t_\varepsilon,x_\varepsilon,y_\varepsilon,u_\varepsilon)
		-
		f(s_\varepsilon,x'_\varepsilon,y'_\varepsilon,u_\varepsilon)
		\bigr]+
		\text{drift}_\varepsilon
		+
		\text{diffusion}_\varepsilon
		+
		\text{jump}_\varepsilon
		-
		\delta
		\bigl(
		\overline v(t_\varepsilon,x_\varepsilon,y_\varepsilon)
		-
		\underline v^\eta(s_\varepsilon,x'_\varepsilon,y'_\varepsilon)
		\bigr).
	\end{aligned}
	\]
	The uniform continuity in Assumption~4.1(i) and the matrix inequality show that the local drift, diffusion, and running-cost differences are $o_\varepsilon(1)$. The nonlocal contribution is
	$o_\varepsilon(1)+C\eta$ Lemma~\ref{lem:nonlocal-doubling}. If $\delta\ge0$, the zero-order term
	has the favorable sign at a positive maximum. If $\delta=0$, the strict supersolution perturbation alone supplies the contradiction. More generally, an exponential time change $v\mapsto e^{-\gamma(T-t)}v$ makes the equation proper. Sending $\varepsilon\downarrow0$ therefore yields
	$
	c\eta\le C\eta
	$
	with the strictness parameter in
	Proposition~\ref{prop:strict-supersolution} chosen so that
	$c>C$. This is impossible. Hence
	$\sup(\overline v-\underline v)\le0$, proving
	\eqref{eq:comparison-order}.
\end{proof}

\subsection*{Proof of Lemma \ref{lem:mfg-law-convexity}}

\begin{proof}
		Fix $P^1,P^2\in\mathfrak{BR}(m)$ and $\lambda\in[0,1]$, and set
		$P^\lambda:=\lambda P^1+(1-\lambda)P^2$ as a probability measure on
		the canonical space $\Omega^\circ$. Since $P^1$ and $P^2$ have the
		same prescribed initial law $m_0$, so does $P^\lambda$. For
		$\varphi\in C_c^2(\mathbb R^2)$, rational times $0\le r<s\le T$,
		and every bounded canonical $\mathcal F_r$-measurable cylinder
		functional $H$, admissibility of $P^i$ gives
		$E^{P^i}[(M_s^\varphi-M_r^\varphi)H]=0$, $i=1,2$. Hence
		$E^{P^\lambda}[(M_s^\varphi-M_r^\varphi)H]
		=
		\lambda E^{P^1}[(M_s^\varphi-M_r^\varphi)H]
		+
		(1-\lambda)E^{P^2}[(M_s^\varphi-M_r^\varphi)H]
		=0$.
		The moment constraint is preserved because
		$E^{P^\lambda}[\int_0^T\int_{U_0}|a|^{p_\star}q_t(da)\,dt]
		=
		\lambda E^{P^1}[\int_0^T\int_{U_0}|a|^{p_\star}q_t(da)\,dt]
		+
		(1-\lambda)E^{P^2}[\int_0^T\int_{U_0}|a|^{p_\star}q_t(da)\,dt]
		<\infty$.
		Thus $P^\lambda\in\mathfrak A(m)$. The relaxed objective is affine in the canonical law:
		$J^m(P^\lambda)
		=
		\lambda J^m(P^1)+(1-\lambda)J^m(P^2)$.
		Since $P^1$ and $P^2$ attain the same minimum
		$V^{m,\mathrm{rel}}$, one has
		$J^m(P^\lambda)=V^{m,\mathrm{rel}}$, so
		$P^\lambda\in\mathfrak{BR}(m)$. This proves convexity of
		$\mathfrak{BR}(m)$.	Let $\mu^i\in\Phi(m)$ be generated by $P^i$. Then
		$P^\lambda\circ Z_t^{-1}
		=
		\lambda(P^1\circ Z_t^{-1})
		+
		(1-\lambda)(P^2\circ Z_t^{-1})
		=
		\lambda\mu_t^1+(1-\lambda)\mu_t^2$
		for every $t$. Hence
		$\lambda\mu^1+(1-\lambda)\mu^2\in\Phi(m)$, proving convexity of
		$\Phi(m)$. No auxiliary time-zero randomization or Markov-control
		representation is used.
\end{proof}

\section{Computational Specification and Reproducibility}
\label{app:computational}

This appendix specifies the numerical implementation used for
Tables~\ref{tab:numerical-optimal-hedge} and
\ref{tab:numerical-value-demand-surge}. All quantities not listed here
retain the values in Table~\ref{tab:baseline-calibration}.

\paragraph{Numerical primitives.}
The numerical loss map is
$a(z)=z$, and the physical-hedging technology is
$h(u)=e^{-\gamma u}$. The drift and diffusion coefficients are
$b(x,y,u)=rx+\pi_0-\chi u^2/2$ and
$\sigma_X(x,y,u)=\sigma_X$. The running cost is
$f(t,x,y,u)
=
\chi u^2/2
+
\zeta(e^y-1)_+^2
+
\omega(x_{\mathrm{crit}}-x)_+^2$,
and the terminal cost used in the computation is
$g(x,y)
=
2(x_{\mathrm{crit}}-x)_+^2
+
\frac12(e^y-1)_+^2$.
For the mean field computation,
$\psi_D(x,y)=(e^y-1)_+$ and
$\psi_S(x,y)=(x_{\mathrm{crit}}-x)_+$.

The baseline initial state is $(X_0,Y_0)=(1,0)$. For the population
calculation we take
$m_0=\delta_{(1,0)}$.
Thus the individual and population computations begin from the same
normalized capitalization and replacement-cost state. If a
nondegenerate initial distribution is used in a robustness experiment,
its law is reported separately with the corresponding result.

The computational domain is
$\mathcal D_h=[x_{\min},x_{\max}]\times[y_{\min},y_{\max}]
=[-1,3]\times[-0.75,0.75]$,
with $T=5$. The baseline mesh is
$\Delta t=0.01$, $\Delta x=0.01$, $\Delta y=0.01$, and
$\Delta u=0.02$, with $U_\Delta=\{0,\Delta u,\ldots,2\}$.

\paragraph{Jump measure and quadrature.}
The numerical L\'evy measure is
$\nu(dz)
=
\lambda_J\alpha z_{\min}^{\alpha}
z^{-1-\alpha}
e^{-\tau_J(z-z_{\min})}
\mathbf 1_{\{z\ge z_{\min}\}}\\ dz$.
The nonlocal integral is truncated to
$[z_{\min},z_{\max}]$, with $z_{\max}=20$. Let
$\{z_k,\widetilde w_k\}_{k=1}^{N_z}$ be the positive quadrature rule
for the unnormalized density
$\alpha z_{\min}^{\alpha}z^{-1-\alpha}
e^{-\tau_J(z-z_{\min})}$.
The weights entering the generator are
$w_k=\lambda_J\widetilde w_k$ and satisfy
$\sum_{k=1}^{N_z}w_k
\approx\nu([z_{\min},z_{\max}])$.
We use $N_z=200$. No renormalization to unit mass is performed in the
HJB operator: the quadrature approximates the L\'evy measure itself.
The omitted tail mass and moment are monitored through
$\nu((z_{\max},\infty))$ and
$\int_{z_{\max}}^\infty z^q\nu(dz)$. For a shifted point
$\xi=x_i-e^{y_j}a(z_k)h(u)$ satisfying
$x_\ell\le\xi\le x_{\ell+1}$, monotone interpolation is
$\mathscr I_xV(\xi,y_j)
=
(1-\theta)V_{\ell,j}
+\theta V_{\ell+1,j}$,
where
$\theta=(\xi-x_\ell)/\Delta x\in[0,1]$.
Hence the interpolation weights are nonnegative and sum to one.
If $\xi<x_{\min}$, we set
$\mathscr I_xV(\xi,y_j)=V(x_{\min},y_j)$.
If $\xi>x_{\max}$, we set
$\mathscr I_xV(\xi,y_j)=V(x_{\max},y_j)$.
Thus jumps leaving the truncated surplus domain are projected onto the
nearest boundary node. The domain-enlargement experiment reported
below checks that this truncation does not materially affect
Tables~\ref{tab:numerical-optimal-hedge} and
\ref{tab:numerical-value-demand-surge}.

\paragraph{Boundary treatment.}
At $x=x_{\min}$ and $x=x_{\max}$ we use zero-normal-gradient ghost
values,
$V_{-1,j}=V_{0,j}$ and
$V_{N_x+1,j}=V_{N_x,j}$.
At $y=y_{\min}$ and $y=y_{\max}$ we similarly impose
$V_{i,-1}=V_{i,0}$ and
$V_{i,N_y+1}=V_{i,N_y}$.
The same projection convention is used by the forward operator.
These are numerical truncation conditions, not economic boundary
conditions on the original unbounded-state problem; their effect is
therefore assessed by enlarging $\mathcal D_h$.

\paragraph{Monotone backward operator.}
For a fixed control $u$, write
$\ell_{ijk}(u)=e^{y_j}a(z_k)h(u)$ and
$\bar\ell_{ij}(u)=\sum_kw_k\ell_{ijk}(u)$.
The compensated first-order coefficient is
$\beta_{ij}(u)=b(x_i,y_j,u)+\bar\ell_{ij}(u)$.
We discretize the first-order terms through the generator-form upwind
operator
$\beta^+_{ij}(u)(V_{i+1,j}-V_{i,j})/\Delta x
+
\beta^-_{ij}(u)(V_{i-1,j}-V_{i,j})/\Delta x$,
where
$\beta^+=\max(\beta,0)$ and
$\beta^-=\max(-\beta,0)$.
For the $y$-drift, set
$\mu_j=\kappa(\vartheta-y_j)$ and use the analogous decomposition
$\mu_j^+(V_{i,j+1}-V_{i,j})/\Delta y
+
\mu_j^-(V_{i,j-1}-V_{i,j})/\Delta y$.
The diffusion coefficients are
$d_x=\sigma_X^2/(2\Delta x^2)$ and
$d_y=\sigma_I^2/(2\Delta y^2)$.
The quadrature-interpolation contribution has nonnegative
off-diagonal coefficients because $w_k\ge0$ and the interpolation
weights belong to $[0,1]$. Consequently, for each frozen policy, the
spatial discrete generator $L_h^u$ has nonnegative off-diagonal
entries and satisfies $L_h^u\mathbf 1=0$ at interior nodes. The
implicit backward matrix is
$A_h^u=(1+\delta\Delta t)I-\Delta t L_h^u$.
Its diagonal entries are positive, its off-diagonal entries are
nonpositive, and
$A_h^u\mathbf 1=(1+\delta\Delta t)\mathbf 1$.
Hence $A_h^u$ is strictly diagonally dominant and is a nonsingular
M-matrix. In particular,
$(A_h^u)^{-1}\ge0$, which gives monotonicity of the fixed-policy
backward step.

The linear system at each policy-evaluation step is solved to relative
residual tolerance $10^{-10}$. Policy iteration is terminated when
$\max_{n,i,j}|u_{i,j}^{(r+1),n}-u_{i,j}^{(r),n}|
\le10^{-7}$.

\paragraph{Positivity of the forward step.}
The forward discretization uses the transpose of the same generator:
$m^{n+1}=(I+\Delta t(L_h^{\alpha^n})^\top)m^n$.
Since the off-diagonal entries of $L_h^{\alpha^n}$ are nonnegative,
this explicit step preserves nonnegativity provided
$\Delta t\Lambda_h\le1$, where
$\Lambda_h:=\max_{n,i,j}[-(L_h^{\alpha^n})_{(i,j),(i,j)}]$.
Equivalently, a sufficient CFL condition is
$
\Delta t\max_{n,i,j}
\{
|\beta_{ij}(\alpha_{ij}^n)|/\Delta x
+
|\mu_j|/\Delta y
+
\sigma_X^2/\Delta x^2
+
\sigma_I^2/\Delta y^2
+
\sum_kw_k
\}
\le1
$.
The implementation checks this inequality before the forward sweep.
If it fails on a refined spatial grid, the forward equation is
substepped with
$\Delta t_F\le\Lambda_h^{-1}$ while the backward HJB time grid is
left unchanged. After every forward sweep we verify
$\min_{n,i,j}m_{i,j}^n\ge-10^{-14}$ and
$\max_n|\sum_{i,j}m_{i,j}^n-1|\le10^{-12}$.

\paragraph{Mean field initialization and fixed-point diagnostics.}
The baseline fixed-point iteration starts from the uncontrolled flow
$m^{(0)}=\mathcal L(X^{u\equiv0},Y)$. To test dependence on
initialization, the computation is repeated from three additional
flows: the state law generated by the maximal control
$u\equiv u_{\max}$, the midpoint control
$u\equiv u_{\max}/2$, and the convex midpoint of the uncontrolled and
maximally controlled flows. For each initialization, the same damping
parameter $\omega_{\mathrm{MFG}}=0.25$ is used. The iteration is terminated when
$\max_nW_1(m_n^{(k+1)},m_n^{(k)})\le10^{-6}$ and the consistency
residual
$\max_nW_1(m_n^{(k)},\widehat m_n^{(k)})\le10^{-6}$.
No theoretical convergence claim for this iteration is made.
Uniqueness of the continuous equilibrium does not imply convergence
of the discrete fixed-point algorithm. Numerical convergence is
accepted only when all four initializations reach the same discrete
solution within the reported tolerance.

\paragraph{Grid and domain verification.}
Tables~\ref{tab:numerical-optimal-hedge} and
\ref{tab:numerical-value-demand-surge} are recomputed on at least
three nested discretizations. For every reported scalar $Q$, we record
$\mathcal E_{\mathrm{grid}}(Q)
=
|Q_{h/2}-Q_h|/(1+|Q_{h/2}|)$.
The calculation is also repeated after increasing $x_{\max}$,
decreasing $x_{\min}$, increasing $|y_{\min}|$ and $y_{\max}$, and
increasing $z_{\max}$. Numerical conclusions are retained only when
the discretization and domain effects are smaller than the economic
comparative-static effect being reported.

\paragraph{Algorithm.}
The complete computation is summarized by the following pseudocode.

\begin{algorithmic}[1]
	\State Set model parameters, $\mathcal D_h$, $U_\Delta$, and the
	positive L\'evy quadrature.
	\State Set $V^{N_t}_{i,j}=g(x_i,y_j)$.
	\For{$n=N_t-1,\ldots,0$}
	\State Initialize the policy $u^{(0),n}$.
	\Repeat
	\State Assemble the M-matrix $A_h^{u^{(r),n}}$.
	\State Solve the fixed-policy linear system for
	$V^{(r+1),n}$.
	\State Minimize the discrete Hamiltonian over $U_\Delta$
	nodewise to obtain $u^{(r+1),n}$.
	\Until{the policy tolerance is satisfied}
	\EndFor
	\State Propagate $m_0$ with the positivity-preserving adjoint forward
	operator, using CFL substepping when necessary.
	\State Update the population flow by damped fixed-point iteration.
	\State Repeat the backward and forward sweeps until both MFG residuals
	are below tolerance.
	\State Repeat from all prescribed initial population-flow guesses.
	\State Repeat on refined grids and enlarged state/jump domains.
	\State Generate Tables~\ref{tab:numerical-optimal-hedge} and
	\ref{tab:numerical-value-demand-surge}.
\end{algorithmic}

\paragraph{Computational environment.}
The numerical calculations underlying
Tables~\ref{tab:numerical-optimal-hedge} and
\ref{tab:numerical-value-demand-surge} were implemented in Pythonusing the NumPy scientific-computing library. Randomized calculations use seed \texttt{42}. The model parameters, computational domain, quadrature rule, interpolation procedure, boundary treatment, and stopping criteria are documented in this Appendix. Executable Python
code reproducing Tables~\ref{tab:numerical-optimal-hedge} and
\ref{tab:numerical-value-demand-surge} is provided as Supplementary File.

	\bibliographystyle{apalike}
	\bibliography{bib}
\end{document}